\documentclass[11pt]{article}
\pdfoutput=1
\usepackage{epsfig}
\usepackage{amsfonts}
\usepackage{mathrsfs}
\usepackage{amssymb}
\usepackage{amstext}
\usepackage{amsmath}
\usepackage{multirow}
\usepackage{xspace}
\usepackage{theorem}
\usepackage{comment}
\usepackage{subfigure}
\usepackage{bigints}
\usepackage{algorithm}
\usepackage{algpseudocode}
\usepackage[sort]{natbib}
\usepackage[hyperref]{xcolor}
\usepackage[pagebackref=false, bookmarks]{hyperref}
\hypersetup{
backref,
pageanchor=true,
bookmarks,
bookmarksnumbered,
colorlinks=true,
menubordercolor= [rgb]{0, 0, 1},
linkbordercolor= [rgb]{0, 1, 1},
citecolor = red,
linkcolor= blue}

\makeatletter
 \renewcommand{\section}{\@startsection{section}{1}{0mm}%
                                   {2ex plus -1ex minus -.2ex}%
                                   {1.3ex plus .2ex}%
                                   {\normalfont\Large\bfseries}}%
 \renewcommand{\subsection}{\@startsection{subsection}{2}{0mm}%
                                     {1ex plus -1ex minus -.2ex}%
                                     {1ex plus .2ex}%
                                     {\normalfont\large\bfseries}}%
 \renewcommand{\subsubsection}{\@startsection{subsubsection}{3}{0mm}%
                                     {1ex plus -1ex minus -.2ex}%
                                     {1ex plus .2ex}%
                                     {\normalfont\normalsize\bfseries}}
 \renewcommand\paragraph{\@startsection{paragraph}{4}{0mm}%
                                    {1ex \@plus1ex \@minus.2ex}%
                                    {-1em}%
                                    {\normalfont\normalsize\bfseries}}
 \renewcommand\subparagraph{\@startsection{subparagraph}{5}{\parindent}%
                                       {2.0ex \@plus1ex \@minus .2ex}%
                                       {-1em}%
                                      {\normalfont\normalsize\bfseries}}
\makeatother

\newenvironment{proof}{{\bf Proof:  }}{\hfill\rule{2mm}{2mm}}
\newenvironment{proofof}[1]{{\bf Proof of #1:  }}{\hfill\rule{2mm}{2mm}}

\newtheorem{remark}{Remark}
\newtheorem{lemma}{Lemma}
\newtheorem{theorem}{Theorem}
\newtheorem{corollary}{Corollary}
\newtheorem{proposition}{Proposition}
\newtheorem{conjecture}{Conjecture}

\newtheorem{definition}{Definition}

\newcommand{\R}{\ensuremath{\mathbb R}}
\newcommand{\Z}{\ensuremath{\mathbb Z}}
\newcommand{\N}{\ensuremath{\mathbb N}}

\newcommand{\prob}[1]{\ensuremath{\text{{\bf Pr}$\left[#1\right]$}}}

\newcommand{\expct}[1]{\ensuremath{\text{{\bf E}$\left[#1\right]$}}}

\newcommand{\ceil}[1]{\ensuremath{\left\lceil#1\right\rceil}}

\newcommand{\junk}[1]{}

\newcommand{\ignore}[1]{}

\newcommand{\fl}[1]{\lfloor{#1}\rfloor} 
\newcommand{\cl}[1]{\lceil{#1}\rceil} 
\newcommand{\id}[1]{{1}_{\{{#1}\}}} 

\newcommand{\E}{\mathbb{E}} 
\newcommand{\dto}{\Rightarrow} 

\newcommand{\D}{\mathbf{D}} 

\newcommand{\prs}[1]{\mathbb{P}\left({#1}\right)} 
\newcommand{\fls}[1]{\bar{#1}^{(r)}} 
\newcommand{\sfs}[1]{\bar{#1}^{(r,l)}} 
\newcommand{\ds}[1]{\hat{#1}^{(r)}} 
\newcommand{\ssp}[1]{{#1}^{(r)}} 

\newcommand{\ser}{\mathcal{Z}} 
\newcommand{\buf}{\mathcal{Q}} 

\newcommand{\M}{\mathbf{M}} 

\newcommand{\pov}{\mathbf{d}} 
\newcommand{\inn}[2]{\langle{#1},{#2}\rangle} 

\newcommand{\cratio}{\frac{c^2_s+1}{c^2_s+c^2_a}} 

\def\argmax{\mathop{\rm arg\,max}}%
\def\argmin{\mathop{\rm arg\,min}}%

\usepackage[normalem]{ulem}
\newcommand{\ins}[1]{\textcolor{red}{\v{}{#1}}} 

\graphicspath{{./figures/}}

\begin{document}
\title{Approximations and Optimal Control for State-dependent Limited Processor Sharing Queues}
\author{Varun Gupta \\  \small{Booth School of Business} \\  \small{University of Chicago} \\ \small{\texttt{varun.gupta@chicagobooth.edu}} \and Jiheng Zhang \\ \small{Department of Industrial Engg. and Logistics Management} \\ \small{The Hong Kong University of Science and Technology} \\ \small{\texttt{j.zhang@ust.hk}}}
\date{}
\maketitle
\rule[0.1in]{\textwidth}{0.025in}

\begin{abstract}
The paper studies approximations and control of a processor sharing (PS) server where the service rate depends on the number of jobs occupying the server. The control of such a system is implemented by imposing a limit on the number of jobs that can share the server concurrently, with the rest of the jobs waiting in a first-in-first-out (FIFO) buffer. A desirable control scheme should strike the right balance between efficiency (operating at a high service rate) and parallelism (preventing small jobs from getting stuck behind large ones). 

We employ the framework of \emph{heavy-traffic} diffusion analysis to devise near optimal 
 control heuristics for such a queueing system. However, while the literature on diffusion control of state-dependent queueing systems begins with a sequence of systems and an exogenously defined drift function, we begin with a finite discrete PS server and propose an axiomatic recipe to explicitly construct a sequence of state-dependent PS servers which then yields a drift function. 
We establish diffusion approximations and use them to obtain insightful and closed-form approximations for the original system under a static concurrency limit control policy.

We extend our study to control policies that dynamically adjust the concurrency limit. We provide two novel numerical algorithms to solve the associated diffusion control problem.
Our algorithms can be viewed as ``average cost'' iteration: The first algorithm uses binary-search on the average cost and can find an $\epsilon$-optimal policy in time $O\left( \left(\log \frac{1}{\epsilon} \right)^2\right)$; the second algorithm uses the Newton-Raphson method for root-finding and requires $O\left( \log \frac{1}{\epsilon} \log \log \frac{1}{\epsilon} \right)$ time.

Numerical experiments demonstrate the accuracy of our approximation for choosing optimal or near-optimal static and dynamic concurrency control heuristics.

\end{abstract}




\section{Introduction}

Consider an \emph{emergency room} where doctors, nurses, and diagnostic equipment make up a shared resource for admitted patients. It has been empirically observed that the service rate of such service systems is state-dependent (e.g., \cite{BattTerwiesch2012}). Human operators tend to speed up service when there is congestion. As another example, consider a typical \emph{web server or an online transaction processing system}. In such resource sharing systems, as the number of tasks (also called active threads) concurrently sharing the server increases, the server throughput initially increases due to more efficient utilization of resources. However, as the server switches from one task to another, it needs to make room for the new task's data in its cache memory by evicting an older task's data (only to fetch it again later). Without a limit on the number of concurrent tasks, this contention for the limited memory can lead to a phenomenon called \emph{thrashing} which causes the system throughput to drop drastically (e.g., \cite{Denningetal76, Blake82, Agrawal85, Elniketyetal04, HeissWagner91, Welsh01}). 

The resource sharing system examples we have described above fall into the category of the so-called \emph{State-dependent Limited Processor Sharing} (Sd-LPS) systems. To specify an Sd-LPS system, we begin with a processor sharing (PS) server whose service rate varies as a function of the number of jobs at the server. For example,
\begin{align}
\label{eqn:mu_example}
\mu(1) = 1, \mu(2) = 1.5, \mu(3) = 1.25, \mu(4) = 1, \mu(5) = 0.75, \ldots
\end{align}
When there are $n$ jobs at the PS server, each job gets served at a rate of $\frac{\mu(n)}{n}$ jobs/second. To ensure efficient operation, we impose a limit on the maximum number of jobs that can be served in parallel. We call this the concurrency limit, $K$. Arriving jobs that find the server busy with $K$ jobs wait in a first-come-first-served (FCFS) buffer. A \emph{static concurrency control policy} is one where the concurrency limit is independent of the state. If the concurrency level can vary with the system state (e.g., the queue length of the FCFS buffer), we call it a \emph{dynamic concurrency control policy}.

To understand the tradeoff involved in choosing the optimal concurrency level, suppose there are 3 jobs in the system described above. Even though the server is capable of serving at an aggregate rate of 1.5 jobs/second by limiting the concurrency level to 2, we may choose to increase the concurrency level to 3 and operate below peak capacity. Why might we want to do that? It is well known that if the job size distribution has high variability, then pure PS outperforms FCFS scheduling by allowing small jobs to overtake large ones. Therefore, it may be beneficial to increase the concurrency level beyond the peak efficiency even if some capacity will be lost. Similarly, for job size distributions with low variability, it may be beneficial to operate at $K=1$. Thus Sd-LPS systems are not ``work-conserving'' queueing systems.

\subsection*{Contributions}
Naturally, our goal is to choose the `best' concurrency control policy. In this work we aim to develop a diffusion approximation framework for Sd-LPS queues, and to utilize the proposed diffusion approximation to find concurrency control policies that minimize the mean sojourn time. This immediately leads to the question: 
\emph{Given that we want to control the state-dependent PS server (exemplified by \eqref{eqn:mu_example}), what is a `meaningful' asymptotic scaling and diffusion approximation?} 

While there are some works on heavy-traffic asymptotics for queues with state-dependent rates, $(i)$ they begin with a sequence of systems with exogenously given limiting drift functions whereas we begin with a discrete PS server of the kind shown in \eqref{eqn:mu_example}, $(ii)$ they are limited to models where the server can only serve one job at a time whereas multiple jobs are processed in parallel by the PS server; and $(iii)$ they only analyze a Jackson network type of system and do not solve a diffusion control problem. 
The present paper fills these gaps in the literature.  

Our main contributions are as follows:
\begin{enumerate}
\item We propose an axiomatic approach to ``reverse-engineer'' a sequence of Sd-LPS queueing systems starting with a discrete state-dependent PS server (Section~\ref{sec:model}). This sequence yields a limiting state-dependent drift function which we utilize to develop diffusion approximations. 
All prior literaure on diffusion analysis of state-dependent queues assumes that the drift function is given exogeneously.
\item We propose an approximation for the distribution of the number of jobs in the Sd-LPS system for a static concurrency limit under a $GI$ arrival process and $GI$ job sizes. This approximation is used to choose a near-optimal static concurrency limit to minimize any cost that is a function of the number of jobs in the system. 
\item We extend our framework by proposing a more general scaling for developing dynamic (state-dependent) control policies. We present two numerical algorithms for solving the resulting diffusion control problem to minimize the steady-state mean number of jobs in the system. Our algorithms can be viewed as `average cost iteration' (as opposed to value function iteration or policy iteration) and are novel to the best of our knowledge. In our simulation experiments, the dynamic policies based on diffusion control perform remarkably close to the true optimal dynamic policies (for input distributions where the true optimal policy can be computed numerically). 
\end{enumerate}

\subsection*{Related work on control of LPS systems}
The literature on LPS-type systems has mostly focused on the constant rate LPS queue where the server speed is independent of the state. Yamazaki and Sakasegawa \cite{YamazakiS1987} show qualitatively the effect of increasing the concurrency level on the mean sojourn time for NWU (New-Worse-than-Used) and Erlang job size distributions. Avi-Itzhak and Halfin \cite{AH88} derive an approximation for the mean sojourn time for the constant rate LPS queue with $M/GI/$ input process, while Zhang and Zwart \cite{ZhangZwart08} derive one for $GI/GI/$ input. Nair et al. \cite{NairWZ2010} expose the power of LPS scheduling by analyzing the tail of sojourn time under light-tailed and heavy-tailed job size distributions. They prove that with an appropriate choice of the concurrency level as a function of the load, LPS queues can achieve robustness to the distribution of job sizes (their tail to be precise). 

For Sd-LPS queues, Rege and Sengupta \cite{RegeSengupta85} derive expressions for the moments and distribution of the sojourn time under $M/M/$ input. Gupta and Harchol-Balter \cite{PSMPL_paper} propose an approximation for the mean sojourn time for $GI/GI/$ input by approximating the interarrival times and job size distribution by the tractable degenerate hyperexponential distribution. They also propose heuristic dynamic admission control policies under $M/GI/$ input.

In this paper, we propose the first diffusion approximation for Sd-LPS queues with a $GI/GI$ input and a static concurrency level. In addition, we propose the first heuristic dynamic admission control policies for Sd-LPS queues.

\subsection*{Related work on control of queueing systems}
There is a considerable literature on the control of the arrival and service rates of queueing systems, but the majority of this work focuses on control of $M/M/1$ or $M/M/s$ systems via Markov decision process formulation, e.g., \cite{GeorgeHarrison_2001, AtaShneorson_2006, AdusumilliHasenbein_2010, LeeKulkarni_2014}. Ward and Kumar \cite{WardKumar_2008} look at the diffusion control formulation for admission control in a $GI/GI/1$ with impatient customers. Our model differs significantly from those in the literature: in our model, the space of actions is the number of jobs admitted to the PS server and is therefore state-dependent. The action space for control problems studied in the literature is usually either the probability of admitting jobs or the server speed, neither of which is state-dependent. The state-dependence of the action space means that the value function may not even be monotonic in the state. We establish this result for our problem and present a simple criterion under which monotonicity holds for general control problems with state-dependent action spaces (see proof of Proposition~\ref{prop:monotonic_V}). In addition, the rather arbitrary nature of the service rate curve precludes elegant structural results for the optimal value function which leads us to propose novel and efficient numerical algorithms for solving the resulting diffusion control problem.

\subsection*{Related work on heavy-traffic analysis of systems with state-dependent rates}
Our heavy-traffic scaling is most closely related to the recent work of Lee and Puhalskii \cite{LeePuhalskii_2012}, who analyze a queueing network of FCFS queues in the critically loaded regime and under non-Markovian arrival and service processes. Yamada \cite{Yamada_1995} also analyzes Markovian state-dependent queueing networks under a similar scaling of state-dependent service and arrival rates. Whereas \cite{LeePuhalskii_2012, Yamada_1995} assume an exogenously given limiting drift function, we propose a method to calculate it from the finite queueing system which is the object of the control problem. Further, the scheduling policy we consider is Processor Sharing. Other works on analysis of heavy-traffic asymptotics of state-dependent Markovian queues include Krichagina \cite{Krichagina_1992}, Mandelbaum and Pats \cite{MandelbaumPats_1998}, Janssen et al.~\cite{JanssenLS_2013}.  

\subsection*{Outline}
In Section \ref{sec:model} we present details of the Sd-LPS model, introduce the notation used in the paper, and describe our approach towards arriving at the asymptotic regime for diffusion analysis.
In Section~\ref{sec:static_control}, we present our results on diffusion approximation for the Sd-LPS queue under a static concurrency control policy. We defer the proofs of convergence to the appendix. 
In Section~\ref{sec:dynamic_control} we turn to dynamic concurrency control policies for the Sd-LPS queue. We set up a diffusion control problem for the limiting diffusion-scaled system, and propose two numerical algorithms to solve the diffusion control problem. 
We make our concluding remarks in Section~\ref{sec:conclusions}.

\section{Model and Diffusion Scaling}
\label{sec:model}

\subsection{Stochastic model and Notation}
\label{sec:measure-vauled-model}
We begin with a description of the Sd-LPS system for which we want to find the optimal control.  Let $X(t)$ denote the total number of jobs in the system at time $t$. The control of such a system is implemented by imposing a concurrency limit $K$. Only $Z(t)=X(t)\wedge K$ jobs are in service and server capacity of $\mu(Z(t))$ is shared equally among the jobs. The remaining $Q(t)=(X(t)-K)^+$ jobs wait in a FCFS queue. A job, once in service, stays in service until completion. The rate of the server $\mu(Z(t))$ is understood to be the speed at which it drains the workload.  So the \emph{cumulative service amount} a job in service can receive from time $s$ to $t$ is
\begin{equation}\label{eq:cumulative}
S(s,t)=\int_s^t\psi(Z(\tau))d\tau,
\end{equation}
where 
\begin{equation}
  \label{eq:sharing}
  \psi(z) = \left\{
    \begin{array}{ll}
      \frac{\mu(z)}{z}, & \textrm{if }z \ge 0,\\
      0, & \textrm{if }z = 0.
    \end{array}
    \right.
\end{equation}
Without loss of generality, we assume that there is no intrinsic limit on the number of jobs the server can serve as we can set the service rate to 0 to model such a limit. 
Note that for the regular state-independent system whose service rate $\mu(\cdot)$ is a constant, say 1, $\frac{\mu(z)}{z}$ in the above will simply become $1/z$.
The state-dependent service rate makes the system \emph{non}-work-conserving, which brings a fundamental challenge to their study. Existing studies of PS or LPS systems crucially rely on the fact that the system is work-conserving, which implies that the workload process is equivalent to that of a simple $G/G/1$ queue. However, this is not the case for our Sd-LPS model. 


The number of job arrivals in time $[0,t]$ is denoted by $\Lambda(t)$. We assume that $\Lambda(\cdot)$ is a renewal process with rate $\lambda$, and $c^2_a$ denotes the squared coefficient of variation (SCV) for the $i.i.d.$ inter-arrival times. The system is allowed to be non-empty initially. We index jobs by $i = -X(0)+1,-X(0)+2,\ldots,0,1,\ldots$. The first $X(0)$ jobs are initially in the system, with jobs $i =-X(0)+1,\ldots,-Q(0)$ in service and jobs $i =-Q(0)+1,\ldots,0$ waiting in the queue. Arriving jobs are indexed by $i = 1, 2, \ldots$. The size of the $i$th job is denoted by $v_i$. We assume job sizes are $i.i.d.$ random variables with mean size $m$ (in the chosen unit measuring work) and SCV $c^2_s$. Jobs leave the system once the cumulative amount of service they have received from the server exceeds their job sizes.

In this study, we are interested in how the system performance (e.g., expected number of jobs in steady state) depends on the state-dependent service rate function $\mu(\cdot)$, the parameters $(\lambda, c^2_a, m, c^2_s) $ of the stochastic primitives, and the concurrency level $K$, which is a decision variable we can control and optimize.

\subsubsection*{Measure-valued state descriptor}
Analyzing the stochastic processes underlying the Sd-LPS model with generally distributed service times requires tracking of more information about the system state than just the number of jobs. Following the framework in \cite{ZDZ2009, ZhangDaiZwart2011}, we introduce a \emph{measure-valued} state descriptor to describe the full state of the system. At any time $t$ and for any Borel set $A\subset(0,\infty)$, let $\buf(t)(A)$ denote the total number of jobs in the buffer whose job size belongs to $A$ and $\ser(t)(A)$ denote the total number of jobs in service whose residual job size belongs to set $A$. Thus, $\buf(\cdot)$ and $\ser(\cdot)$ are measure-valued stochastic processes. 
Let $\delta_a$ denote the Dirac measure of point $a$ on $\R$ and $A+y \doteq \{a+y:a\in A\}$.  
By introducing the measure-valued processes, we can characterize the evolution of the system via the following \emph{stochastic dynamic equations}: 
\begin{align}
  \buf(t)(A)&=\sum_{i={B}(t)+1}^{\Lambda(t)}\delta_{v_i}(A),\label{eq:stoc-dym-eqn-B}\\
  \ser(t)(A)&=\ser(0)(A+{S}(0,t))
  +\sum_{i={B}(0)+1}^{{B}(t)}\delta_{v_i}(A+{S}(\tau_i,t)),\label{eq:stoc-dym-eqn-S}
\end{align}
where $\tau_i$ is the time when the $i$th job starts to receive service and
\begin{align}\label{eq:B(t)}
  B(t) = \Lambda(t)-Q(t),
\end{align}
which can be intuitively interpreted as the index of the last job to enter service by time $t$. For any Borel measurable function $f:\R_+\to\R$, the integral of this function with respect to a measure $\nu$ is denoted by $\inn{f}{\nu}$. Then, both $Q(t)$ and $Z(t)$ can be represented using the measure-valued descriptors: 
\begin{align*} 
  Q(t)=\inn{1}{\buf(t)},\quad Z(t)=\inn{1}{\ser(t)}.
\end{align*}
Let $W(t)$ denote the workload of the system at time $t$ which is defined as the sum of the sizes of all jobs in queue and the remaining sizes of all jobs in service. Due to the varying service rate of the server, the dynamics of the workload process is represented by
\begin{equation}
  \label{eq:workload-dynamics}
  W(t) = W(0) + \sum_{i=1}^{\Lambda(t)} v_i - \int_0^t\mu(Z(s))\id{W(s)>0}ds.
\end{equation}
Again, we can express the workload $W(t)$ in terms of the measure-valued descriptors:
\begin{equation}
  W(t)=\inn{\chi}{\buf(t)+\ser(t)},
\end{equation}
where $\chi$ denotes the identity function on $\R$.

\subsection{Proposed Asymptotic Regime for Diffusion Approximation of Sd-LPS systems}
\label{sec:regime-scaling}
We refer to the system introduced in Section~\ref{sec:measure-vauled-model} as our \emph{original} system. We now propose an asymptotic regime where a sequence of Sd-LPS systems, parametrized by $r\in \Z^+$, will be studied under an appropriate scaling. The objective is to a obtain a meaningful approximation of the original system with the goal of choosing the `best' concurrency control policy. This leads to the question:
\begin{quote}
\emph{What is the appropriate scaling to analyze the Sd-LPS queue? That is, what asymptotic regime captures the entire service-rate curve of the original Sd-LPS system, and thus can be used to find a near-optimal concurrency limit?}
\end{quote} 
As we mentioned earlier, the scaling we develop is very close to the scaling proposed by Yamada \cite{Yamada_1995} and Lee et al. \cite{LeePuhalskii_2012}. To provide an axiomatic justification for why this is the appropriate scaling for Sd-LPS systems we begin by examining \emph{two special cases of Sd-LPS systems}: $(i)$ multiserver systems, and $(ii)$ the constant rate LPS queue.

\begin{description}

\item[\bf The $G/GI/k$ multiserver system] A $G/GI/k$ multiserver system with a service rate of $\mu$ jobs/second per server and a central buffer can be viewed as an Sd-LPS system with  $\mu(n)=n \mu$ and a concurrency limit of $K=k$. There is a rich literature on the question of whether having many slow servers is better than having a few fast servers (e.g., Brumelle \cite{Brumelle1971}, Daley and Rolski \cite{DR84}), which is similar in spirit to the concurrency control problem. The work on diffusion approximations for multiserver systems started with K{\"o}llerstr{\"o}m \cite{Kollerstrom74} for the classical heavy-traffic regime where $k$ and $\mu$ are held constant while $\lambda \uparrow k \mu$. A more refined heavy-traffic regime is the Halfin-Whitt regime (starting with \cite{HalfinWhitt1981} and more recently \cite{Reed2009}, \cite{GamarnikMomcilovic08}) where one fixes $\mu$ and creates a sequence of multiserver systems parametrized by $r$, where the number of servers grows according to $\ssp{k} = rk$ while the mean arrival rate $\ssp{\lambda}$ increases so that $\frac{\ssp{k} \mu - \ssp{\lambda} }{\sqrt{\ssp{k}}} \to \beta$. The constant $\beta$ is chosen so that the probability that an arrival gets blocked converges to a non-degenerate limit (bounded away from $0$ and $1$). An extremely accurate diffusion approximation for a given $G/M/k$ system can be obtained by matching the blocking probability under the Halfin-Whitt regime. 

\item[\bf State-independent (constant rate) LPS queue] In the state-independent LPS queue, the service rate of the server is a constant $\mu$ irrespective of the number of jobs at the server, and there is a fixed concurrency limit $k$. Recently, Zhang et al.~\cite{ZhangDaiZwart2011} have proposed and analyzed a diffusion approximation for the LPS system where a sequence of LPS systems (parametrized by $r$) is devised so that the service rate remains fixed at $\mu$, the concurrency limit increases according to $\ssp{k} = r k $ and the arrival rate increases so that $\ssp{k} (\mu - \ssp{\lambda} ) \to \theta$, a constant. As in the Halfin-Whitt regime for the multiserver systems, under the proposed scaling for LPS systems the probability that an arrival finds all slots at the PS server occupied converges to a constant. In addition, the queue length scaled by $\frac{1}{\ssp{k}}$ also converges to a non-degenerate distribution, unlike Halfin-Whitt where the queue lengths are smaller and must be scaled by $\frac{1}{\sqrt{\ssp{k} }}$.
\end{description}

It is not obvious how either of these scalings can be extended to the Sd-LPS system, but among the chief desiderata is that the diffusion-scaled system should in some sense be a \emph{faithful proxy} for the original system. As an example of a regime that is not quite faithful enough, consider the following example: We scale the concurrency limit as $k^{(r)}=rk$, leave the mean arrival rate $\lambda$ unchanged, and `stretch' the service rate curve so that for the $r$th Sd-LPS system, $\mu^{(r)}(rx) = \hat{\mu}(x)$ where $\hat{\mu}(\cdot)$ is a continuous interpolation of $\mu(\cdot)$. This would correspond to a fluid limit where the steady state `gets stuck' around $x^*$, where $\hat{\mu}(x^*)=\lambda$. Therefore this fluid regime cannot be used to devise a control policy for the original Sd-LPS system.

Instead, we adopt an axiomatic approach to devising the asymptotic regime: Under some \emph{reasonable non-trivial assumptions} $\mathcal{A}$, the \emph{behavior} $\mathcal{B}$ of the diffusion-scaled system \emph{should mimic} the original discrete system we want to approximate. The choice of $\mathcal{A}$ and $\mathcal{B}$ can be seen as the axioms of our scaling which we will use to reverse-engineer a diffusion scaling. We now formalize our choice of the assumptions $\mathcal{A}$ and behavior $\mathcal{B}$ for Sd-LPS systems with static concurrency control. Later we will use the intuition gained from this exercise to engineer a scaling for developing dynamic control policies.

\noindent \textbf{Axioms for the Sd-LPS diffusion scaling}\\
We construct a sequence of Sd-LPS systems parametrized by $r\in \Z^+$ such that the $r$th system has a concurrency level of 
\begin{equation}
  \label{eq:K-r}
  \ssp{k} = rK.
\end{equation}
Further, the sequence of service rate curves $\ssp{\mu}(\cdot)$ satisfies:
\begin{description}
\item[\bf ($\mathcal{A}$, the assumptions)] under a Poisson arrival process with rate $\lambda$ and $i.i.d$. Exponentially distributed job sizes with mean size $m$ (i.e., under $M/M/$ input),
\item[\bf ($\mathcal{B}$, the behavior)] the distribution of the scaled number of jobs in the system (scaled by $\frac{1}{\ssp{k}}$) converges to a non-degenerate limit.
\end{description}
\ \\
\noindent \textbf{Consequences of the scaling axioms and an alternate characterization} \\
Since we start by fixing the concurrency levels for the sequence of Sd-LPS systems, the only design flexibility we have to satisfy the scaling axioms is the choice of state-dependent service rate curves. Let us denote the service rate curve for the $r$th system by $\ssp{\mu}(i)$, and the resulting distribution of the number of jobs in the $r$th system under $M/M/$ input by $\ssp{F}$. Our goal is to find the sequence $\ssp{\mu}(\cdot)$ so that
\begin{align}
\label{eqn:F_convergence_condition}
\lim_{r \to \infty} \ssp{F}(\ceil{r x} ) & = \hat{F}(x)  \quad \forall x \in [0, \infty)
\end{align}
for some distribution function $\hat{F}(\cdot)$. This gives us our \textbf{first way} of deriving the scaling: Fix $\hat{F}(\cdot)$ to be a smooth, strictly increasing interpolation of the distribution function of the original system  under $M/M/$ input and reverse-engineer the sequence of service rate functions $\ssp{\mu}(i)$. As we will show in the next section, the requisite service rate functions satisfy
\begin{align}
\label{eqn:mu_condition_F}
\lim_{r\to \infty} r \left( \lambda m - \ssp{\mu}(\ceil{r x}  \right) &= \lambda m \frac{d \log f(x)}{dx} \quad \forall x \in [0,\infty).
\end{align}
where $f(x) = \frac{d}{dx}\hat{F}(x)$.

Of course, by reverse-engineering the scaling, we guarantee ourselves a non-degenerate limit that is interesting in that it captures the effect of the entire $\mu(\cdot)$ function. Further, it turns out we never really need to compute the service rate functions $\ssp{\mu}(\cdot)$! In Section~\ref{sec:static_control}, we will show that we can directly express the limiting steady-state quantities in terms of the distribution $\hat{F}(\cdot)$, which can be easily obtained from that of the original system.

However, the method described above does not generalize to dynamic control policies since the distribution $F(\cdot)$ (and its smooth interpolation $\hat{F}(\cdot)$) was obtained under the assumption of a static concurrency limit of $K$. For this case, we propose a \textbf{second way} of deriving the scaling, that still guarantees \eqref{eqn:F_convergence_condition}:

Begin with $\hat{\mu}(\cdot):\R^+_0 \to \R^+_0$ satisfying:
\begin{enumerate}
\item $\hat{\mu}(\cdot)$ agrees with $\mu(\cdot)$ at integer arguments: $\hat{\mu}(i)=\mu(i)$ for $i=\{1,2,\ldots\}$
\item $\hat{\mu}(\cdot)$ is continuous and smooth
\end{enumerate}
The sequence of service rate functions $\{\ssp{\mu}(\cdot)\}$ is chosen to satisfy
\begin{align}
\label{eqn:mu_condition_mu}
\lim_{r \to \infty} r  \left( \lambda m - \ssp{\mu}(\ceil{r x}  \right) &= \lambda m \log\frac{\lambda m}{\hat{\mu}(x)}  \quad \forall x \in [0,\infty).
\end{align}
For either way of arriving at the diffusion scaling, we see that $r(\lambda-\ssp{\mu}(\ceil{r x)})$ converges to a non-degenerate \emph{drift function} $-\theta(x)$. In the first case the $\theta(x)$ function is reverse-engineered by fixing a limiting distribution, and in the second case it is obtained more directly using a continuous extension of  $\mu(i)$. The first is more appropriate for finding static control policies, while the second is more appropriate for computing dynamic control policies. In both cases there is limited flexibility in extending a discrete function to a continuous smooth function.
\ \\

\noindent \textbf{Intuitive explanation for the choice of service rate curves $\ssp{\mu}(\cdot)$}\\
We begin by explaining our first choice of the asymptotic regime \eqref{eqn:mu_condition_F}. Consider the $r$th Sd-LPS system operating under $M/M/$ input. Let $\ssp \pi(i)$ be the probability mass function for the steady-state number of jobs in the $r$th system. Flow-balance equations imply
\[ 
\frac{\ssp \pi (\ceil{rx+1})}{\ssp \pi(\ceil{rx})} = \frac{\lambda m}{ \ssp \mu(\ceil{rx})}.
\]
Since, by design, we want $r \ssp \pi(\ceil{rx})$ to converge to the density function $f(x)$, we should have:
\[ 
\frac{\lambda m}{\ssp \mu(\ceil{rx})} \approx \frac{f\left(x+\frac{1}{r}\right)}{f(x)} \approx 1 + \frac{1}{r} \frac{f'(x)}{f(x)}.
\]
Equivalently, $r(\lambda m - \ssp \mu(\ceil{rx}) \to \lambda m \frac{d \log f(x)}{dx}$.

To motivate the second proposal of $\ssp{\mu}(\cdot)$, note that if $\frac{\lambda m}{\ssp{\mu}(r\cdot x)} \approx 1- \frac{\theta(x)}{\lambda m r} \approx e^{-\frac{\theta(x)}{\lambda m r}}$, then
\[ \frac{\ssp{\pi}(ry) }{\ssp{\pi}(rx)} \approx e^{- \frac{1}{\lambda m}\int_{x}^{y} \theta(u) du} \to \frac{\pi(y)}{\pi(x)} = \Pi_{i=x+1}^{y} \frac{\lambda m}{\mu(i)}.\]
Or,
\[ - \frac{1}{\lambda m}\int_{x}^{y} \theta(u) du \approx \log \frac{\ssp \pi(ry)}{\ssp \pi(rx)} \approx \log \frac{\pi(y)}{\pi(x)} = \sum_{i=x+1}^{y} \log \frac{\lambda m}{\mu(i)}. \]
Comparing the first and last expressions above gives us an approximation for $\theta(u)$ in terms of a continuous extension $\hat{\mu}$ of $\mu$: $r(\lambda m - \ssp \mu(\ceil{rx}) \to - \lambda m \log \frac{\lambda m}{\hat{\mu}(x)}$.

\noindent \textbf{Comparison with existing diffusion scalings}
For the two specific examples of Sd-LPS systems we pointed out earlier we can now ask: What axioms do the existing diffusion scalings satisfy?

The \ins{implicit} axioms used by \cite{HalfinWhitt1981} posit that each system in the sequence is itself a homogeneous multiserver system, and further, under $M/M/$ input the blocking probability of the sequence converges to a non-degenerate limit (for example, to the blocking probability of the finite system being approximated). If we use our proposed scaling to approximate a multiserver system, the sequence of Sd-LPS systems would not be a homogeneous multiserver system. Indeed, $r(\lambda m-\ssp{\mu}(rx))$ grows as $\log(x)$ for small $x$. On the positive side, this flexibility allows us to capture the entire distribution of number of jobs, not just the blocking probability. We should also point out that while the Halfin-Whitt scaling is more useful for capacity provisioning, the goal of our proposed scaling is to solve the admission/concurrency control problem.

For the constant rate LPS system, our scaling matches the diffusion scaling of Zhang et al. \cite{ZhangDaiZwart2011}, and thus can be seen as an extension of their scaling to Sd-LPS. Our convergence proofs follow their outline as well.

\section{Diffusion approximation for the Sd-LPS queue with a static concurrency level}
\label{sec:static_control}
The goal of this section is to provide approximations for the steady-state performance of the Sd-LPS queue with a static concurrency level under the proposed scaling \eqref{eqn:mu_condition_F}. 
In Section~\ref{sec:diff_summary} we first summarize the results of this section by giving an approximation for the mean number of jobs in an Sd-LPS system under a static concurrency level (equation~\eqref{eqn:EN_static_approx}), and providing some simulation results which show the utility of the approximation for choosing a near-optimal concurrency level.
In Section~\ref{sec:diffusion_analysis}, we prove process-level limits for diffusion-scaled workload and head count processes. In Section~\ref{sec:ss_diffusion}, we justify using the steady state of the limiting processes as an approximation for the limit of the steady state of the diffusion-scaled processes by establishing the required interchange of limits. We also present closed-form formulae for these steady-state distributions. All the proofs for this section can be found in the appendix.

\subsection{An approximation and simulation results} \label{sec:diff_summary}
Let $N$ denote the steady-state number of jobs in the Sd-LPS system for a given static concurrency level $K$. Our main result of this section yields the following simple approximation formula for the expectation of $N$ as a function of the concurrency level and other system parameters (see Proposition~\ref{prop:final_approximation} for the formal statement)  
\begin{align}
\expct{N} & \approx 
\frac
{\sum_{n=0}^{\infty} (n \wedge K) \pi(n)^{\frac{c^2_s+1}{c^2_s+c^2_a}}}
{\sum_{n=0}^{\infty} \pi(n)^{\frac{c^2_s+1}{c^2_s+c^2_a}}}
+ \left(\frac{c^2_s+1}{2} \right)
\frac
{\sum_{n=0}^{\infty} (n-K)^+ \pi(n)^{\frac{c^2_s+1}{c^2_s+c^2_a}}}
{\sum_{n=0}^{\infty} \pi(n)^{\frac{c^2_s+1}{c^2_s+c^2_a}}},
\label{eqn:EN_static_approx}
\end{align}
where $\pi(n)$ denotes the steady-state probability of there being $n$ jobs in the Sd-LPS system under $M/M/$ input (that is, Poisson arrivals with mean rate $\lambda$ and $i.i.d.$ Exponentially distributed job sizes with mean size $m$).

\begin{figure}[ht]
\begin{center}
\includegraphics[width=2.45in]{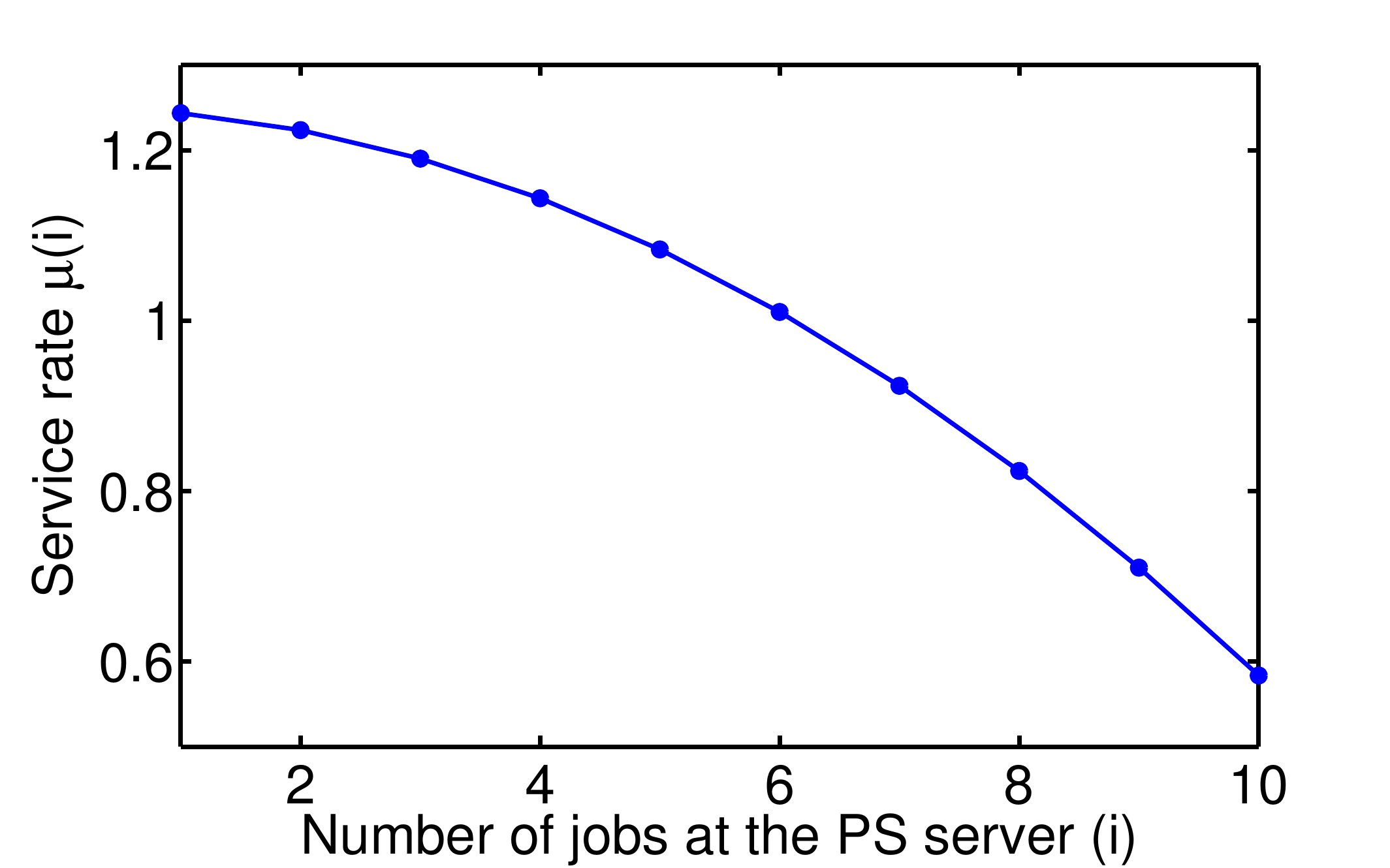}
\caption{State-dependent service rate function used for simulation results} \label{fig:mu_function_quad}
\end{center}
\end{figure}
Figure~\ref{fig:mu_function_quad} shows a hypothetical service rate function for a PS server. The service rate has the functional form $\mu(i) = 1.25 - \frac{i^2}{150}$, and is monotonically decreasing in the concurrency level. Figure~\ref{fig:performance_static} shows the simulation results for the steady-state mean number of jobs as a function of the concurrency level $K$. The arrival process is Poisson with mean arrival rate shown below the figures. We simulated three distributions, each with mean $m=1$ and SCV $c^2_s=19$. The solid curve shows the diffusion approximation \eqref{eqn:EN_static_approx} for the mean number of jobs. For each value of $\lambda$ and each distribution, the optimal concurrency level obtained via approximation \eqref{eqn:EN_static_approx} matches the one obtained from simulating the LPS system. We should point out the caveat that while the proposed diffusion approximation accurately captures the \emph{shape} of $\expct{N}$ versus the concurrency level curve and thus provides good guidance for concurrency control, the actual numerical values for $\expct{N}$ are not always very accurate for all values of $K$.
\begin{figure}[h!]
\hspace{-0.3in}
\begin{minipage}{7.5in}
\subfigure[$\lambda=0.7$]{\includegraphics[width=2.45in]{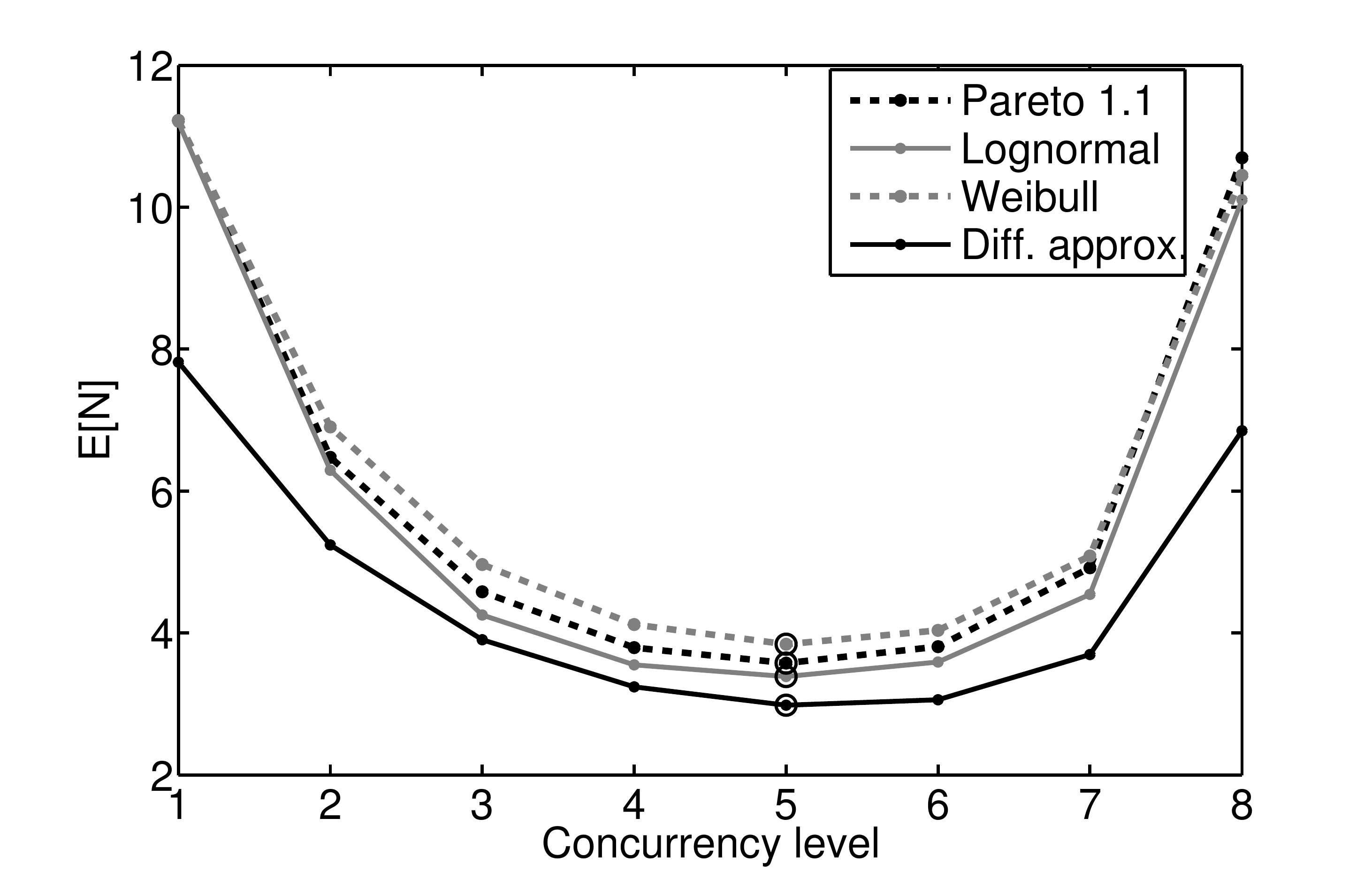}}
\hspace{-0.3in}
\subfigure[$\lambda=0.8$]{\includegraphics[width=2.45in]{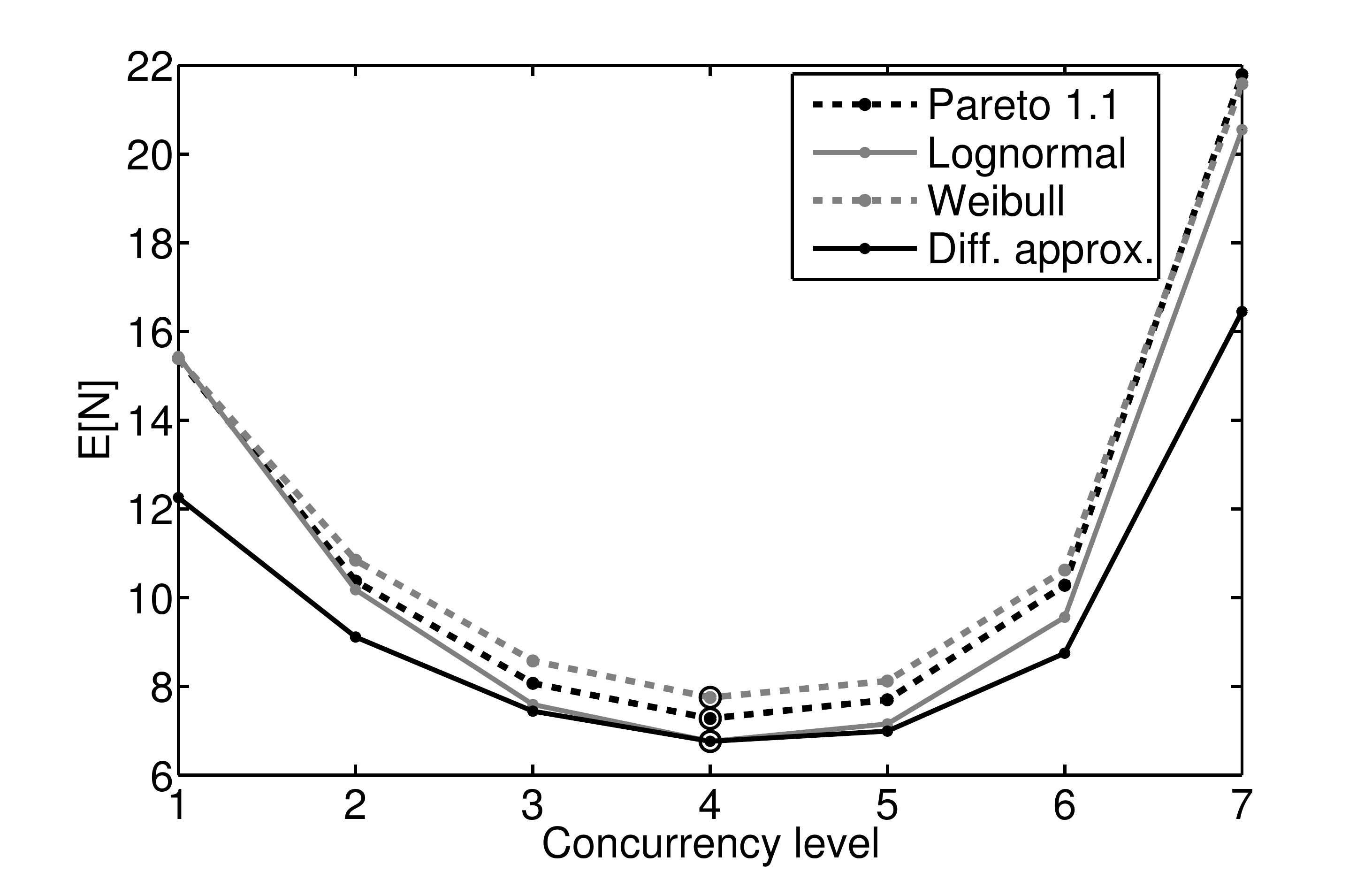}}
\hspace{-0.3in}
\subfigure[$\lambda=0.9$]{\includegraphics[width=2.45in]{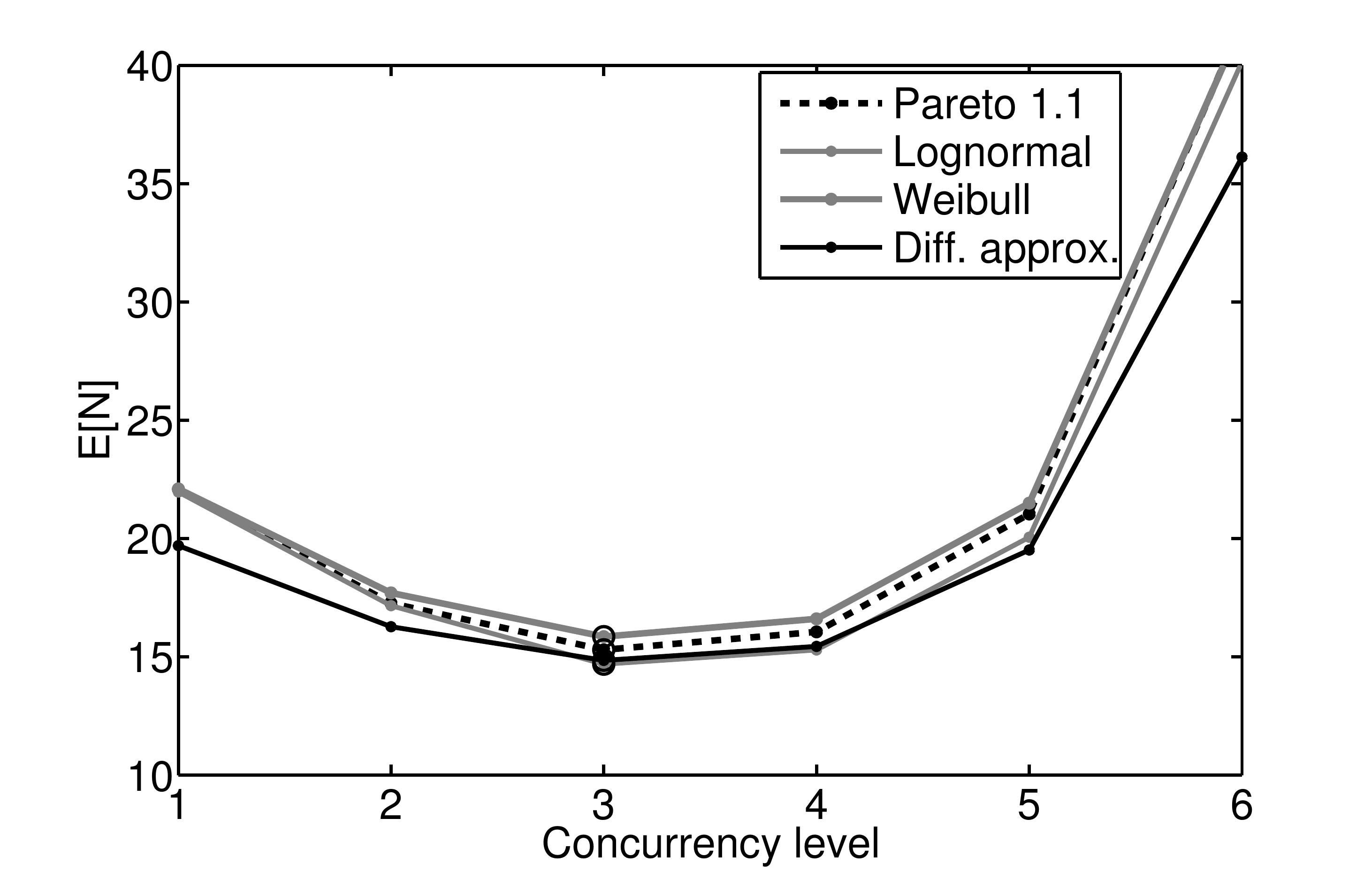}}
\end{minipage}
\caption{Simulation results for mean number of jobs in the system versus the concurrency level for the service rate function shown in Figure~\ref{fig:mu_function_quad} for various job size distributions, all with mean $m=1$ and SCV $c^2_s=19$. The arrival process is Poisson with indicated mean arrival rate $\lambda$. Also shown is the diffusion approximation from equation~\eqref{eqn:EN_static_approx}. The optimal concurrency level for each curve is shown with a circle.}
\label{fig:performance_static}
\end{figure}

\subsection{Diffusion analysis of Sd-LPS system}\label{sec:diffusion_analysis}
We now present the analysis of the Sd-LPS system under the asymptotic regime described in \eqref{eqn:mu_condition_F}. For generality and notational convenience, we present all the analysis in terms of the general drift function $\theta(x)$, and then translate the result into a form involving $f(x)$ (Proposition~\ref{prop:final_approximation}) for convenience.

Consider the sequence of Sd-LPS systems indexed by $r$. We append a superscript $(r)$ to all the quantities associated with the $r$th system. 
The concurrency level $\ssp{k}$ is specified as \eqref{eq:K-r}. 

Assume that the arrival process $\ssp{\Lambda}(\cdot)$ satisfies 
\begin{equation}
  \label{eq:HT-arrival}
  \frac{\ssp{\Lambda}(r^2t)-r^2\lambda t}{r}\dto M_a(t),\quad \textrm{as }r\to\infty,
\end{equation}
where $M_a(\cdot)$ is a Brownian motion with zero drift and variance $c_a^2$. 
Further, we assume that the sizes of arriving jobs follow distribution $G$ which satisfies
\begin{equation}
  \label{eq:cond-ser-dist}
  G \text{ is a continuous distribution function with mean } m.
\end{equation}
Introduce the drift function
\begin{equation*}
  \ssp{\theta}(x) = \left\{
    \begin{array}{ll}
       r\left(\ssp{\mu}(\ceil{rx})-\lambda m \right) & x>0,\\
       0 & x=0.
    \end{array}
    \right.
\end{equation*}
The above definition is only for technical convenience, since otherwise $\ssp\mu(0)$ would be undefined. However, this does not matter since the server idles when there are no jobs in the system.
The heavy traffic condition is specified by 
\begin{equation}
  \label{eq:cond-HT}
  \ssp{\theta}(x)\stackrel{u.o.c}{\longrightarrow}\theta(x)
  \quad\textrm{as }n\to\infty,
\end{equation}
for some locally Lipschitz continuous function $\theta(\cdot)$ on $(0,\infty)$ satisfying 
\begin{equation}
  \label{eq:cond-stability}
  \theta(K)>0.
\end{equation}
The notation $\stackrel{u.o.c}{\longrightarrow}$ means uniform convergence on compact sets, which is only required for technical reasons. Condition \eqref{eq:cond-stability} ensures that the system is stable (see the proof of Theorem~\ref{thm:XW_ss_convergence}). 
As a quick remark, we make a connection with the traditional single server system where the server speed is constant, say $\ssp{\mu}(\cdot)\equiv1$, and the drift is created by constructing a sequence of $\ssp{\lambda}$ which converges to $\lambda$ at the rate of $1/r$. The heavy traffic condition for this constant rate LPS system then becomes 
\begin{equation*}
  r\left(1-\ssp{\lambda}m\right)\to\theta>0, \quad \textrm{as }r\to\infty.
\end{equation*}
We are interested in the asymptotic behavior of the diffusion-scaled processes for the $r$th system, defined as
\begin{align}
  \label{eq:difusion-scaling}
  \ds X(t)=\frac{1}{r}\ssp{X}(r^2t),\quad
  \ds W(t)=\frac{1}{r}\ssp{W}(r^2t). 
\end{align}
The diffusion scaling for other stochastic processes $\ssp\buf$, $\ssp\ser$, $\ssp{Z}$, $\ssp{Q}$ and $\ssp{B}$ is defined in the same way. To obtain the diffusion limit of the head count process $\ds X$ and workload process $\ds W$, we need to carefully analyze the measure-valued processes introduced. The detailed analysis is presented in Appendix~\ref{sec:state-space-collapse}. 

Since we need to work with the measure-valued process, let $\nu$ denote the probability measure associated with the probability distribution function $G$, and $\nu_e$ denote the probability measure associated with the \emph{equilibrium} distribution $G_e$ of $G$. That is, $G_e(x) = \frac{1}{m}\int_0^x[1-G(y)]dy$ and the mean of $G_e$ is
\begin{equation*}
  m_e = \frac{1+c^2_s}{2} m.
\end{equation*}
Let $\M$ denote the space of all non-negative finite Borel measures on $[0,\infty)$. 
We need the following regularity assumptions on the initial state to rigorously prove the diffusion approximation results.
Assume there exists $(\xi^*,\mu^*)\in\M\times\M$ such that
\begin{align}
  \label{eq:cond-initial}
  (\ds\buf(0),\ds\ser(0))&\dto(\xi^*,\mu^*),\\
  \label{eq:cond-initial-u}
  \inn{\chi^{1+p}}{\ds\buf(0)+\ds\ser(0)}&\dto\inn{\chi^{1+p}}{\xi^*+\mu^*}\quad \textrm{for some }p>0,
\end{align}
as $r\to\infty$, and
\begin{align}
  \label{eq:cond-init-ssc}
  & (\xi^*,\mu^*)=\left(\frac{w^*\wedge K m_e}{m_e}\nu, \frac{(w^*-K m_e)^+}{m}\nu_e\right),
\end{align}
where $w^*=\inn{\chi}{\xi^*+\mu^*}$. The above regularity assumptions \eqref{eq:cond-initial}--\eqref{eq:cond-init-ssc} basically require that the sequence of initial states is well behaved. These assumptions, together with the heavy traffic assumptions \eqref{eq:HT-arrival}--\eqref{eq:cond-HT}, are made throughout the rest of this paper. 

The first result we present is an asymptotic relationship, called State Space Collapse (SSC), between the workload process and the head count process. Define a map $\Delta_K(\cdot):\R_+\to\R_+$ by
\begin{align}
\label{eq:delta}
\Delta_K(w) = \frac{w\wedge K m_e}{m_e} + \frac{(w-K m_e)^+}{m}. 
\end{align}
The SSC result states that the total number of jobs in the system $\ds X$ can be \emph{asymptotically} represented using the workload $\ds W$ via the map $\Delta_K$, which is a bijective map meaning that workload can also be represented using the total number of jobs. SSC is described as follows: 
\begin{proposition}[State Space Collapse]
\label{prop:state_space_collapse} For the sequence of Sd-LPS systems parametrized by $r \in \Z^+$ and satisfying initial conditions \eqref{eq:cond-initial}-\eqref{eq:cond-init-ssc}, as $r \to \infty$, 
  \begin{equation}
    \label{eq:SSC}
    \sup_{t\in[0,T]}\big|(\ds X(t)\wedge K)m_e+(\ds X(t)-K)^+m
    -\ds W(t)\big|\dto 0.
  \end{equation}
\end{proposition}
Note that \[ \Delta_K^{-1}(x) = (x \wedge K)m_e + (x-K)^+ m \] is the inverse of the map $\Delta_K(\cdot)$.
A full version of the SSC, which demonstrates a bijective map between the workload $\ds W$ and the measure-valued status $(\ds\buf,\ds\ser)$, is presented and proved in Appendix~\ref{sec:state-space-collapse}. 
Roughly speaking, SSC reveals that the residual sizes of jobs in service follow the equilibrium distribution $G_e$. 
The simpler SSC of Proposition~\ref{prop:state_space_collapse} can be derived from the full version proved in Appendix~\ref{sec:state-space-collapse}. For the purpose of performance analysis and for optimal control in this paper, we only need the simple version of SSC. 

The next step is the analysis of the workload process defined in \eqref{eq:workload-dynamics}. The challenge here is that the evolution of workload depends on the number of jobs in service due to the state-dependent service rate. The simple SSC result allows us to overcome this difficulty. The following theorem establishes the diffusion limit of the workload process $\ds W(t)$ and the number of jobs $\ds X$ as reflected Brownian motion (RBM) with state-dependent drifts. 

\begin{theorem}[Weak convergence to RBMs with state-dependent drift]
\label{thm:W_process_convergence}
For the sequence of Sd-LPS systems parametrized by $r \in \Z^+$ satisfying \eqref{eq:cond-initial}-\eqref{eq:cond-init-ssc}, as $r \to \infty$, 
  \begin{align}
    \label{eq:worload-difussion-limit}
    \ds W &\dto W^*,
  \end{align}
  where $W^*$ is an RBM with initial value $W^*(0)=w^*$, drift $-\theta\left(\Delta_K(W^*(t)) \wedge K\right)$ and variance $\sigma^2 = \lambda m^2(c_a^2+c_s^2)$. Moreover, as $r \to \infty$, 
  \begin{equation}
    \label{eq:total-jobs-diffusion-limit}
    \ds X \dto X^*= \Delta_K(W^*).
  \end{equation}
\end{theorem}
The proof of Theorem~\ref{thm:W_process_convergence} is presented in Appendix~\ref{sec:state-space-collapse}.

\subsection{Steady State of the Diffusion Limit}\label{sec:ss_diffusion}

The entire goal of heavy traffic analysis is to obtain a tractable process, an RBM with state-dependent drift, as an approximation of the complicated stochastic process underlying the original model. That is, the steady state of the limiting RBM can be computed. We begin with a basic result on the steady-state distribution of an RBM with state-dependent drift and variance (Lemma~\ref{lemma:RBM_ss}). We then use this lemma to derive the steady-state distribution and mean of the limiting workload and number of jobs ($W^*$ and $X^*$) in terms of the primitives of the original Sd-LPS model.

\begin{lemma}
\label{lemma:RBM_ss}
The stationary distribution of a one-dimensional RBM, $W$, with state-dependent drift $-\beta(\cdot)$ and state-dependent variance $s(\cdot)$ is given by
\begin{equation}
  \label{eq:RBM_ss}
  \prob{W(\infty) \leq w} 
   = \alpha  \int_{0}^w e^{-\int_0^u \frac{\beta(v) + \frac{1}{2}s'(v)}{\frac{1}{2}s(v)} dv} du 
  = \alpha \int_{0}^w \frac{1}{s(u)} e^{-\int_0^u \frac{\beta(v)}{\frac{1}{2}s(v)} dv} du,
\end{equation}
where $\alpha$ is a normalization constant.
\end{lemma}
The proof of this lemma is presented in Appendix~\ref{sec:state-space-collapse}. In our setting, the drift $\beta(w)=\theta\left(\Delta_K(w)\right)$ and the variance $s(w)=\sigma^2 = \lambda m^2(c^2_a+c^2_s)$ is a constant. Using Lemma~\ref{lemma:RBM_ss}, we immediately have the closed-form expression for the steady-state $W^*(\infty)$ of the diffusion limit $W^*$ in terms of the drift function $\theta(\cdot)$:
\begin{corollary}
\label{cor:gen_approximation}
With $W^*$ as defined in Theorem~\ref{thm:W_process_convergence},
\begin{align}
  \label{eq:W_RBM_ss}
  \prob{W^*(\infty) \leq w} 
  &= \begin{cases}
    \alpha \bigintss_0^w e^{-\frac{2\int_0^u \theta(v/m_e)}{\sigma^2}dv} du  
    & w \leq K \cdot m_e, \\
    \alpha \bigintss_0^w \left( e^{-\frac{2\int_0^{Km_e} \theta(v/m_e)}{\sigma^2} dv}  \right) 
    e^{ - \frac{2 \theta(K) (u-Km_e)}{\sigma^2}}  du
    & w > K\cdot m_e.
  \end{cases} 
\end{align}
\end{corollary}
To present the result for the limiting steady-state quantities in a form that is easier to apply in practice, we express the drift function $\theta(\cdot)$ as follows:
\begin{equation}
  \label{eq:theta-f}
  - \theta(x) = \lambda m \frac{d \log f(x)}{dx}
\end{equation}
where, recall, the function $f(\cdot)$ represents the derivative of a twice-differentiable interpolation of the distribution of steady-state number of jobs for the original Sd-LPS system with concurrency level $K$ under $M/M/$ input (see the discussion on the asymptotic regime in Section~\ref{sec:regime-scaling} preceding equation~\eqref{eqn:mu_condition_F}). We assume that $\frac{d \log f(x)}{dx} $ is a constant less than 0 on the interval $[K,\infty)$ (it is easy to verify that such an extension exists). It turns out that while we need $f(\cdot)$ to be differentiable to define $\theta(\cdot)$, the limiting steady-state distribution is well defined even without this condition.
Finally, we obtain the following result on the steady-state distribution of workload and number of jobs in the system:
\begin{proposition}
\label{prop:final_approximation} 
Let $W^*$ and $X^*$ be the workload and number of jobs for the limiting Sd-LPS system (as defined in Theorem~\ref{thm:W_process_convergence}). Let the drift function $\theta(x)$ be given by $-\theta(x) = \lambda m \frac{d\log f(x)}{dx}$. 

The steady-state distributions of $W^*$ and $X^*$ are given by
\begin{align}
	\label{eqn:W_RBM_ss}
	\prob{W^*(\infty) \leq w} &= \alpha \int_{0}^{\frac{w}{m_e}} f(x)^{\frac{c^2_s+1}{c^2_s+c^2_a}} dx, \\
  \label{eq:X_RBM_ss}
  \prob{X^*(\infty) \leq x} 
  &=
  \begin{cases}
    \alpha \int_0^x f(u)^{\frac{c^2_s+1}{c^2_s+c^2_a}} du & x\le K,\\
    \alpha \int_{0}^{K+(x-K)\frac{m}{m_e}} f(u)^{\frac{c^2_s+1}{c^2_s+c^2_a}} du  & x>K,
  \end{cases}
\end{align}
where $\alpha$ is the normalization constant. The mean of the limiting scaled number of jobs is given by
\begin{align}\label{eq:X-expectation-approx}
  \expct{X^*(\infty)} 
  &= \frac{ \int_{x=0}^{\infty} (x \wedge K) f(x)^{\frac{c^2_s+1}{c^2_s+c^2_a} }dx }
          {\int_{x=0}^{\infty} f\left( x\right)^{\frac{c^2_s+1}{c^2_s+c^2_a}} dx } 
   + \frac{c^2_s+1}{2} \cdot 
     \frac{ \int_{x=0}^{\infty} (x-K)^+ f(x)^{\frac{c^2_s+1}{c^2_s+c^2_a} }dx }
          {\int_{x=0}^{\infty} f\left( x\right)^{\frac{c^2_s+1}{c^2_s+c^2_a}} dx }.
\end{align}
\end{proposition}
We have thus obtained closed-form formulae (approximations) for steady-state quantities based on the limiting diffusion process. The approximation \eqref{eqn:EN_static_approx} at the beginning of this section is obtained from \eqref{eq:X-expectation-approx} by further using the probability mass function, $\pi(\cdot)$, for the number of jobs corresponding to the original Sd-LPS system in place of the density function $f(\cdot)$.

We now close the loop by translating the convergence at the process level to convergence of steady-state distributions in the following theorem. The proof is presented in Section~\ref{lem-reflection-map}. This justifies the formulae in Proposition~\ref{prop:final_approximation} as an approximation for the steady state of the original Sd-LPS system. The quality of the approximation is demonstrated in the numerical experiment presented at the beginning of this section (see Figure~\ref{fig:performance_static}).

\begin{theorem}[Convergence of steady-state distributions]
\label{thm:XW_ss_convergence}
For all large enough $r$, the stochastic process $\ds{X}$ has a steady state, denoted by $\ds{X}(\infty)$. Moreover, 
  \begin{align*}
    \ds W(\infty) &\dto W^*(\infty),\\
    \ds X(\infty) &\dto X^*(\infty),
  \end{align*}
  where $W^*(\infty)$ and $X^*(\infty)$ are characterized in \eqref{eq:W_RBM_ss} and \eqref{eq:X_RBM_ss}. 
\end{theorem}

\section{Dynamic concurrency control for the Sd-LPS queue}
\label{sec:dynamic_control}

In Section~\ref{sec:static_control}, we established approximations for the steady-state number of jobs and workload in an Sd-LPS system operating under a \emph{static} concurrency level. 
Our numerical experiments showed that the optimal static level based on the approximations yields near-optimal performance for the original Sd-LPS system. In this section we go further by allowing a \emph{dynamically adjustable} concurrency level. 

In Section~\ref{sec:setup_dynamic} we first summarize our translation of the discrete state space control problem for the original system to a continuous state space diffusion control problem, and the translation of the resulting control back to that for the original Sd-LPS system. To demonstrate the efficacy of our approach, we present results of numerical experiments comparing the performance of the proposed diffusion limit based control policies against the true optimal dynamic control policy for a special non-trivial input process for which the true optimal policy can be computed numerically.
In Section~\ref{sec:diff_control_prob} we formulate the diffusion control problem and show how this can help solve the dynamic control problem for the original system. 
We then describe two novel numerical algorithms to solve the diffusion control problem: an algorithm that iteratively refines its estimate of the average cost of the optimal policy using binary search in Section~\ref{sec:bin_search}, and an algorithm that uses the Newton-Raphson root finding method to search for the average cost of the optimal policy in Section~\ref{sec:Newton}.  


\subsection{Overview of our approach and Simulation results} \label{sec:setup_dynamic}
The following steps outline our approach to obtaining a heuristic dynamic control policy for the original Sd-LPS server:
\begin{enumerate}
\item Convert the discrete service rate vector $\mu(i)$ for the original state-dependent PS server into a drift function according to \eqref{eqn:mu_condition_mu}:
\begin{equation*}
  \theta(x) \doteq -\lambda m \log{\frac{\lambda m}{\hat{\mu}(x)}},
\end{equation*}
where $\hat{\mu}$ is a continuous extension of $\mu(i)$.
\item Formulate a diffusion control problem to minimize the steady-state mean number of jobs. The action/control will be the concurrency level as a function of the state. For convenience, we frame the diffusion control problem with workload as the state variable since the variance of workload is a constant (i.e., independent of state or action), and the control affects the drift of the workload through $\theta(x)$. 
\item Given $k^*(w)$, the optimal concurrency level as a function of the workload for the diffusion control problem, to obtain a control policy for the original (discrete) Sd-LPS system, we first obtain a control function $\widetilde{k}(w)$ with discrete concurrency levels by rounding $k(w)$ to the nearest integer for all $w$. The control algorithm is implemented by using $\widetilde{W}(t) = m_e  Z(t) + mQ(t)$ for the original system as the proxy for the current workload, and then taking action to reach the concurrency level dictated by the diffusion control problem: $\widetilde{k}(\widetilde{W}(t))$. In controling the original system we only take actions upon job arrivals and departures, do not preempt jobs once they enter service, and do not increase the concurrency level by more than one in any arrival/departure event. The precise policy is given as follows:
\begin{itemize}
\item \textbf{On arrival at $t$:} Let $\widetilde{W} = m_e Z(t_-)+m(Q(t_-)+1)$, where $(Z(t_-),Q(t_-))$ denotes the system state immediately before the event. If $\widetilde{k}(\widetilde{W}) \geq (Z(t_-)+1)$ then admit one job to the server at $t$, otherwise do nothing.
\item \textbf{On departure at $t$:} Let $\widetilde{W} = m_e (Z(t_-)-1) + mQ(t_-)$. Admit $\min\big\{ \big(\widetilde{k}(\widetilde{W}) - Z(t_-)+1\big)^+, 2\big\}$ jobs at $t$.
\end{itemize}
\end{enumerate}

\subsection*{Simulation Results}\label{sec:sim_dynamic}

Table~\ref{table:dynamic_results} shows experimental results comparing the performance of the dynamic policies produced using the proposed diffusion scaling and the true optimal dynamic policy.  We focus on a special class of input processes: Poisson arrivals and a degenerate Hyperexponential job size distribution (a mix of a point mass at 0 and an Exponential distribution). This allows us to compute optimal dynamic policies using the algorithm proposed by \cite{PSMPL_paper}. The dynamic policy for the diffusion control problem was computed using the Newton-Raphson method (Algorithm~\ref{alg:Newtons_method}, Section~\ref{sec:Newton}). The service rate curve is the one shown in Figure~\ref{fig:mu_function_quad}, which gives the drift function as
\begin{align}
- \theta(x) \doteq \lambda m \log \frac{\lambda m}{\hat{\mu}(x)} \doteq \lambda m \log \frac{\lambda m}{1.25 - \frac{x^2}{150}}.
\end{align}
We used MATLAB's \texttt{ode45} function to solve the differential equations involved in Algorithm~\ref{alg:Newtons_method}. The performance of the diffusion control policy was evaluated by simulating it for a Poisson arrival process and Hyperexponential job size distribution with the indicated $c^2_s$. 

For each of the six cases shown, the steady-state mean number of jobs for the diffusion control heuristic is within $2\%$ of the optimal dynamic policy, demonstrating the validity of our proposed scaling for computing control policies for Sd-LPS systems across a range of traffic intensities.

\begin{table}[ht]
\renewcommand*{\arraystretch}{1.2}
\centering
\begin{tabular}{|c|c||c|c|c|}
\hline
\multicolumn{2}{|c||}{} & \multicolumn{3}{c|}{Steady-state mean number of jobs $\expct{N}$} \\
\cline{3-5}
\multicolumn{2}{|c||}{} & Opt. dynamic policy & Diffusion control policy & Suboptimality (\%)\\
\hline
\hline
\multirow{3}{*}{$c^2_s=4$} & $\lambda = 0.7$ & 1.739 & 1.744 & 0.29 \\ 
\cline{2-5}
 & $\lambda = 0.8$ & 2.854 & 2.885 & 1.09\\ 
\cline{2-5}
& $\lambda = 0.9$ & 4.873 & 4.893 & 0.41 \\
\hline
\multirow{3}{*}{$c^2_s=19$} & $\lambda = 0.7$ & 2.90 & 2.94 & 1.38 \\
\cline{2-5}
 & $\lambda = 0.8$ & 6.51 & 6.63 & 1.84 \\
\cline{2-5}
& $\lambda = 0.9$ & 14.24 & 14.33 & 1.01 \\
\hline
\end{tabular}
\caption{Simulation results comparing the performance of dynamic policies for Poisson arrivals with rate $\lambda$ and a degenerate hyperexponential ($H^*$) job size distribution with $m=1$ and SCV $c^2_s$. The first column shows the steady-state mean number of jobs for the optimal dynamic policy. The second column shows the same metric for the heuristic policy obtained from the diffusion control problem. For each case, the diffusion control policy yields an $\expct{N}$ of at most 2\% larger than the optimal policy.}
\label{table:dynamic_results}
\end{table}

\subsection{The diffusion control problem} \label{sec:diff_control_prob}
In this section we set up the diffusion control problem for dynamic concurrency control of the LPS server. We begin by generalizing the scaling of the concurrency limit $\ssp{k}$ given in \eqref{eq:K-r} so that it becomes a function of the workload in the system:
\begin{equation}
  \label{eq:K-r-dyn}
  \ssp{k}(\ssp{W}(t)) = r k\left(\frac{\ssp{W}(t)}{r} \right),
\end{equation}
where $k:\R_+\to\R_+$ and $k(w)\le w/m_e$ for any $w$. The restriction $k(w) \leq w/m_e$ on the choice of concurrency level is driven by the state space collapse (see Conjecture~\ref{thm:W_process_convergence-dyn}).

The objective is to find the optimal state-dependent concurrency level function $k(\cdot)$.  For technical reasons, we restrict our consideration to the following family of dynamic controls
\begin{equation}
  \label{eq:policy-family-dyn}
  \mathcal K = \left\{ k:\R_+\to\R_+| k(w)\le w/m_e ; \textrm{  $k$ is Lipschitz continuous} ; \int_{v=0}^\infty e^{-\int_0^v \theta(k(w)) dw} dv < \infty \right\}. 
\end{equation}
The Lipschitz continuity requirement is for technical reasons, and the last condition above is only to ensure that a stationary distribution for the diffusion-scaled workload under $k(w)$ exists. We use the same heavy traffic regime as in Section~\ref{sec:diffusion_analysis} except that stability condition \eqref{eq:cond-stability} is replaced by 
\begin{align}
  \label{eq:cond-stability-dyn}
\sup_{x\in[0,M]} \theta(x) > 0,
\end{align}
for some $M<\infty$. That is, a service rate strictly larger than the arrival rate is achievable at a finite concurrency limit and hence at a finite workload. In fact, we will make a stronger assumption on $\sup_{x} \theta(x)$. Define
\begin{align*}
\hat \theta \doteq \sup_{x\in\R_+} \theta(x) \ ;\quad \hat{k} \doteq \argmax_{k} \{ \theta(k) \}.
\end{align*}
Here $\hat{k}$ denotes the most efficient concurrency level, which we will assume to be finite.
For any $k\in\mathcal K$, define the mapping $\Delta_k:\R_+\to\R_+$ by 
\begin{equation}
  \label{eq:map-delta-dyn}
  \Delta_k(w) = \frac{w \wedge k(w)m_e}{m_e} + \frac{(w-k(w)m_e)^+}{m}.
\end{equation}
Note that we use $\Delta_K$ to denote the mapping under a static concurrency level $K$, and $\Delta_k$ to denote the mapping under a dynamic concurrency policy $k\in\mathcal K$.
Extending the diffusion limit result in Section~\ref{sec:diffusion_analysis}, we have the following conjecture:

\begin{conjecture}[Diffusion limits under a dynamic policy]
\label{thm:W_process_convergence-dyn}
For the sequence of Sd-LPS systems parametrized by $r \in \Z^+$ under the dynamic policy \eqref{eq:K-r-dyn} for some $k\in\mathcal K$, as $r \to \infty$, 
  \begin{align}
    \label{eq:worload-difussion-limit-dyn}
    \ds W &\dto W^*,
  \end{align}
  where $W^*$ is an RBM with initial value $W^*(0)=w^*$, drift $-\theta\left(\Delta_K(W^*(t))\wedge k(W^*(t)) \right)$ and variance $\sigma^2 = \lambda m^2(c_a^2+c_s^2)$. Moreover, as $r \to \infty$, 
  \begin{equation}
    \label{eq:total-jobs-diffusion-limit-dyn}
    \ds X \dto X^* = \Delta_k(W^*).
  \end{equation}
\end{conjecture}
In other words, we conjecture that the state space collapse result still holds and the state-dependent concurrency level function $k(\cdot)$ only plays a role in modifying the drift of the diffusion limit of the workload.
The key to proving this conjecture is to extend the state space collapse result to allow a dynamic concurrency level and analyze the underlying fluid model (as in \cite{ZhangDaiZwart2011}). Due to the technical intricacies involved, proving the conjecture is beyond the scope of this paper. Instead, we focus on utilizing the conjectured diffusion limit to identify a near-optimal policy for the original LPS system. 


As mentioned earlier, we will formulate the diffusion control problem with the limiting workload process $W^*$ as the state variable and the map $\Delta_k(\cdot)$ as the state-dependent cost function (using the state space collapse conjecture \eqref{eq:total-jobs-diffusion-limit-dyn}). There are two reasons for choosing $W^*$ over the head count process $X^*$ as the state variable: $(i)$ the variance of $X^*$ is state-dependent making the computation more complicated, while it is a constant for $W^*$, and $(ii)$ headcount does not carry enough information since two different states $(Q_1, Z_1)$ and $(Q_2,Z_2)$ along the state-space collapse trajectory  may have the same head count but different workloads. Therefore the control is not uniquely obtained as a function of the number of jobs in the system.

Let $V_{\gamma}(w)$ denote the discounted total cost (with discount rate $\gamma$) for the limiting process $W^*$ under a control policy $k(\cdot)$ when the workload starts in state $w$:
\begin{align}
  \label{eqn:Bellman_eqn}
  V_{\gamma}(w) &= \E_w \Big[\int_{0}^{\infty} e^{-\gamma t} \Delta_{k}(W^*(t)) dt\Big].
\end{align}

\noindent \textbf{Optimality Equations}\\
Consider a small $\delta>0$. According to It\={o} calculus
\begin{align*}
  V_\gamma(w) &= \Delta_k(w) + (1-\gamma\delta) \E\left[V_\gamma\left(W^*(\delta)\right)\right]+o(\delta)\\
  &= \Delta_k(w) + (1-\gamma\delta) \E\left[V_\gamma(w)+V'_\gamma(w)(W^*(\delta)-w)
  +\frac{V''_\gamma(w)}{2}(W^*(\delta)-w)^2+o(\delta)\right]+o(\delta)\\
  &= \Delta_k(w) + (1-\gamma\delta) \left[V_\gamma(w)+V'_\gamma(w)\theta(k(w))\delta
  +\frac{V''_\gamma(w)}{2}\sigma^2\delta\right]+o(\delta),
\end{align*}
where, recall,  $\sigma^2 \doteq \lambda m^2 (c^2_s+c^2_a)$. 
We thus have the following relation for the discounted value function $V_\gamma$:
\begin{equation}
  \label{eq:Bellman}
  \gamma V_\gamma(w) = \Delta_k(w) - \theta(k(w)) V'_\gamma(w) + \frac{\sigma^2}{2} V''_\gamma(w).
\end{equation}
Letting $\gamma \to 0$, define
\begin{align*}
v = \lim_{\gamma \to 0} \gamma V_\gamma(w) \ , \ \mbox{ and }\  G(w) = \lim_{\gamma \to 0} V'_\gamma(w),
\end{align*}
where $v$ is the average cost of policy $k(\cdot)$, and the value function gradient $G(w)$ solves the following ordinary differential equation (ODE):
\begin{align}
\label{eqn:G_ODE}
v &= \Delta_k(w) - \theta(k(w)) G(w) + \frac{\sigma^2}{2} G'(w).
\end{align}
Above, we have provided a heuristic derivation to arrive at the average cost optimal control problem as a limit of the discounted cost problem. For a formal treatment of the relation between discounted and average cost problems (i.e., by defining discounted relative cost functions $h_\gamma(w) = V_\gamma(w)-V_\gamma(\tilde{w})$ for some positive recurrent state $\tilde{w}$, taking limit $h(w) = \lim_{\gamma \downarrow 0} h_\gamma(w)$ and $v = \lim_{\gamma \downarrow 0} \gamma V_{\gamma}(\tilde{w})$), we refer readers to \cite[Chapter 5]{Hernandez-LermaDiscreteMCP}, \cite{BertsekasDP}.

\begin{proposition}
\label{prop:monotonic_V}
The discounted value function $V_{\gamma}(w)$ is non-decreasing in $w$ for all $\gamma$, and hence $G(w) \geq 0$. 
\end{proposition}

\begin{remark} For a given control policy $k(w)$, equation \eqref{eqn:G_ODE} is a first order ODE for $G(w)$. However, to solve $G(\cdot)$ we also need to know the average cost $v$. This is to be expected since we started from a second order ODE where we would need two boundary conditions to completely specify $V_\gamma$. In our case, one boundary condition is easy to get hold of: since we have a reflecting boundary at $w=0$, we must have (see, for example, \cite[page VIII]{Mandl_analytical}): 
\begin{align}
V'_\gamma(0) &= 0
\end{align}
and therefore, also $G(0)=0$. This observation will be critical in the development of our algorithms.
\end{remark}

Returning to equation \eqref{eq:Bellman}, let $V^*_{\gamma}$ denote the value function for the optimal policy. Then Bellman's principle of optimality becomes:
\begin{align}
\gamma V^*_\gamma(w) &= \min_{ k \in [0, w/m_e] } \left\{  \Delta_{k}(w) - \theta(k) {V^*_\gamma}'(w)   + \frac{\sigma^2}{2} {V^*_\gamma}''(w) \right\}.
\end{align}
If we let $\gamma \to 0$, then
\begin{align}
\label{eqn:G*_ODE}
v^* &= \min_{ k \in [0, w/m_e] } \left\{  \Delta_{k}(w) - \theta(k) {G^*}(w)   + \frac{\sigma^2}{2} {G^*}'(w) \right\},
\end{align}
where again, as remarked earlier, we have the boundary condition $G^*(0)=0$, leaving $v^*$ the only unknown.

Though many diffusion control problems addressed in the literature have a nice structure allowing a closed-form solution, e.g., \cite{HST1983,HarrisonTaksar1983}, the problem \eqref{eqn:G*_ODE} is intrinsically difficult mainly due to the generality of the service rate curve. Thus we seek numerical algorithms, which presents another challenge. For diffusion control problems where a closed-form solution can be found, one of the boundary conditions is imposed by setting the coefficient of the exponential term in the solution of the second order ODE to zero. This captures the physical constraint that the optimal value function should asymptotically grow at a polynomial rate and not exponentially. However, this trick cannot be applied when searching for a numerical solution, which led us to develop the algorithms in Sections~\ref{sec:bin_search} and \ref{sec:Newton} to get around this obstacle. While the majority of numerical algorithms for solving diffusion control problems rely on the Markov chain method where time and space are discretized and a probability transition matrix is engineered to satisfy local consistency requirements (e.g., \cite{KushnerDupuis_NumericalMethods}), we directly work with the ODE in \eqref{eqn:G*_ODE}.
  


\begin{figure}[ht]
\hspace{-0.3in}
\begin{minipage}{7.5in}
\subfigure[Drift function $\theta(\cdot)$]{
\includegraphics[width=2.45in]{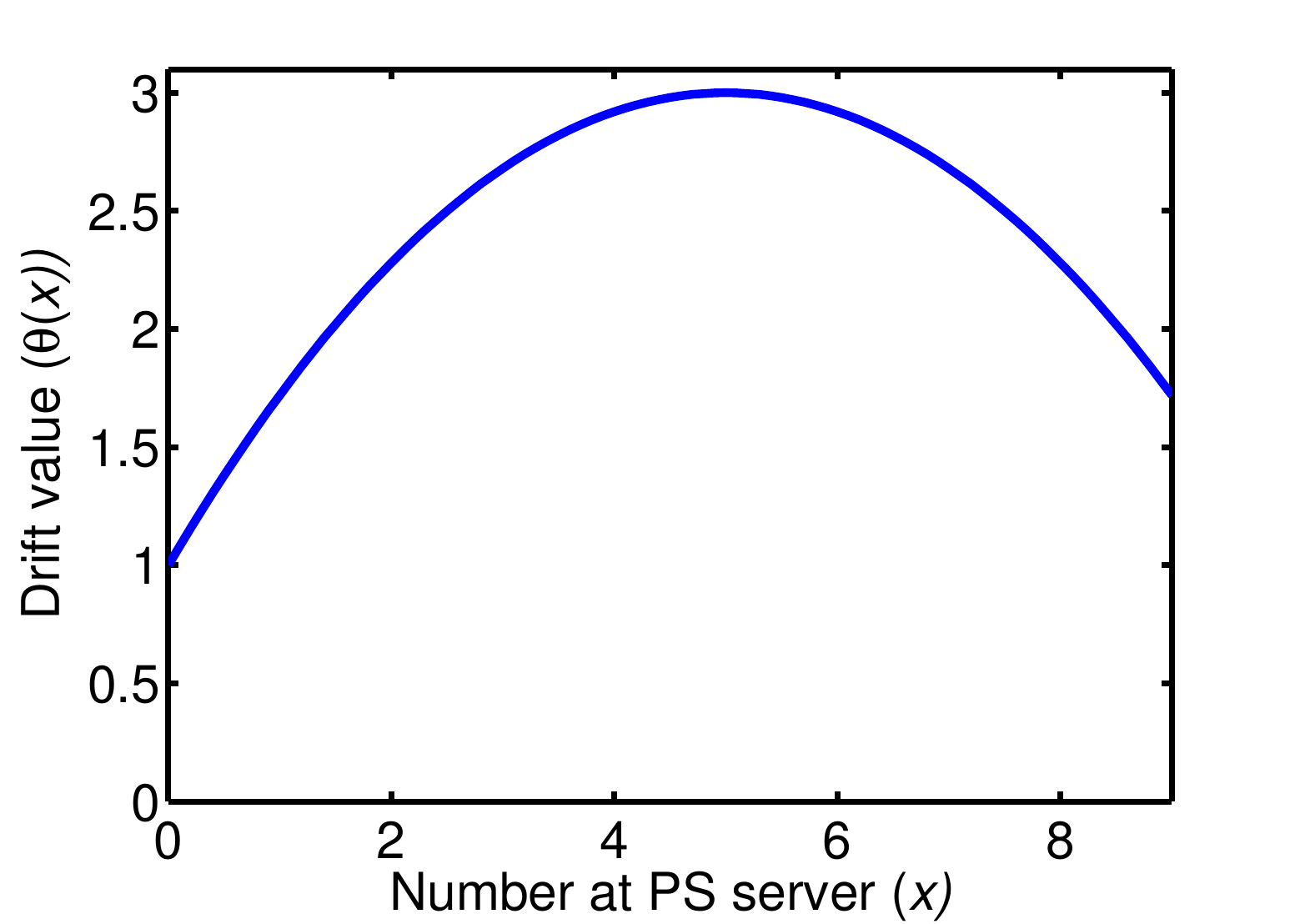}}
\hspace{-0.3in}
\subfigure[$c^2_a=c^2_s=10$]{
\includegraphics[width=2.45in]{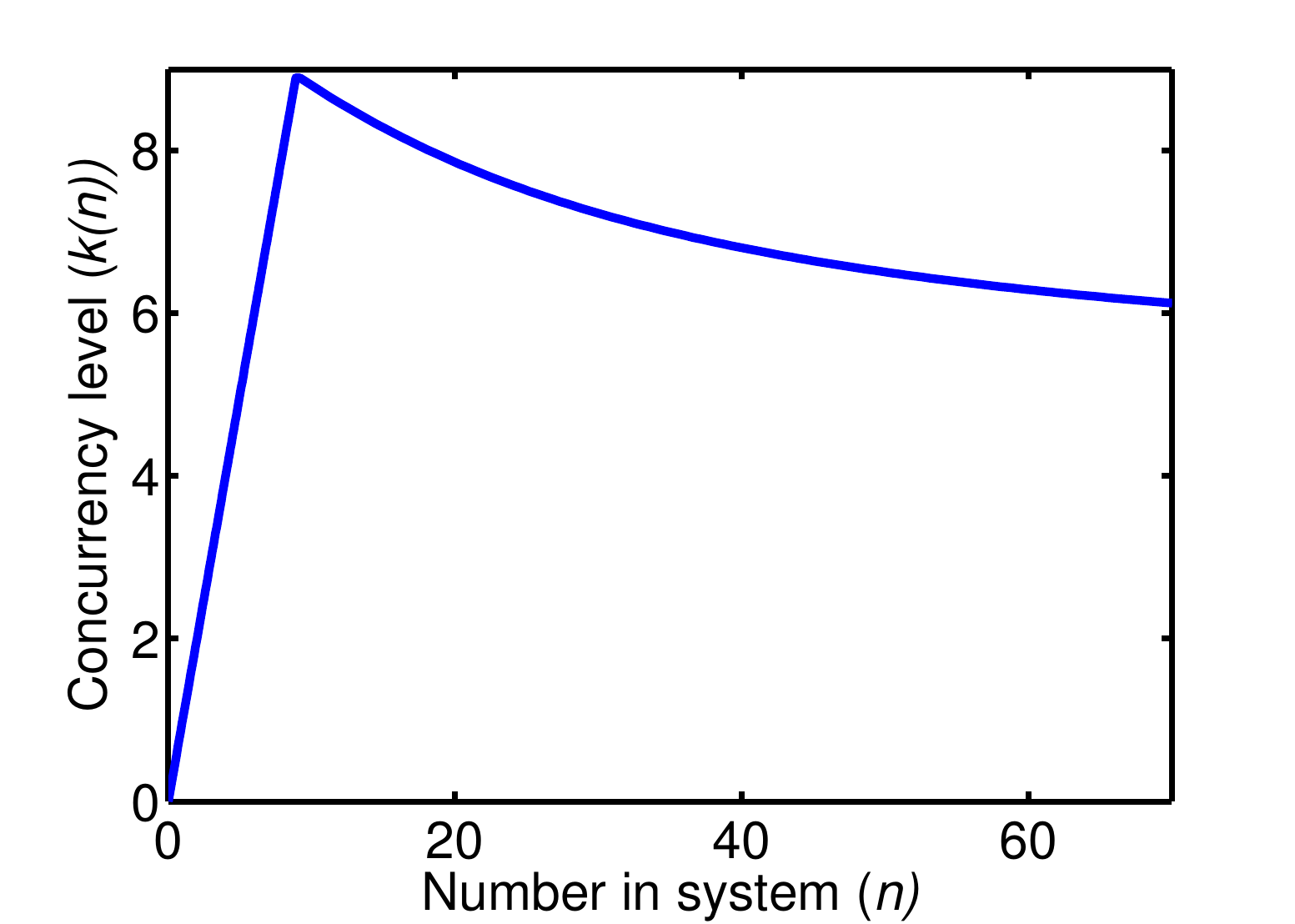}}
\hspace{-0.3in}
\subfigure[$c^2_a=c^2_s=0.3$]{
\includegraphics[width=2.45in]{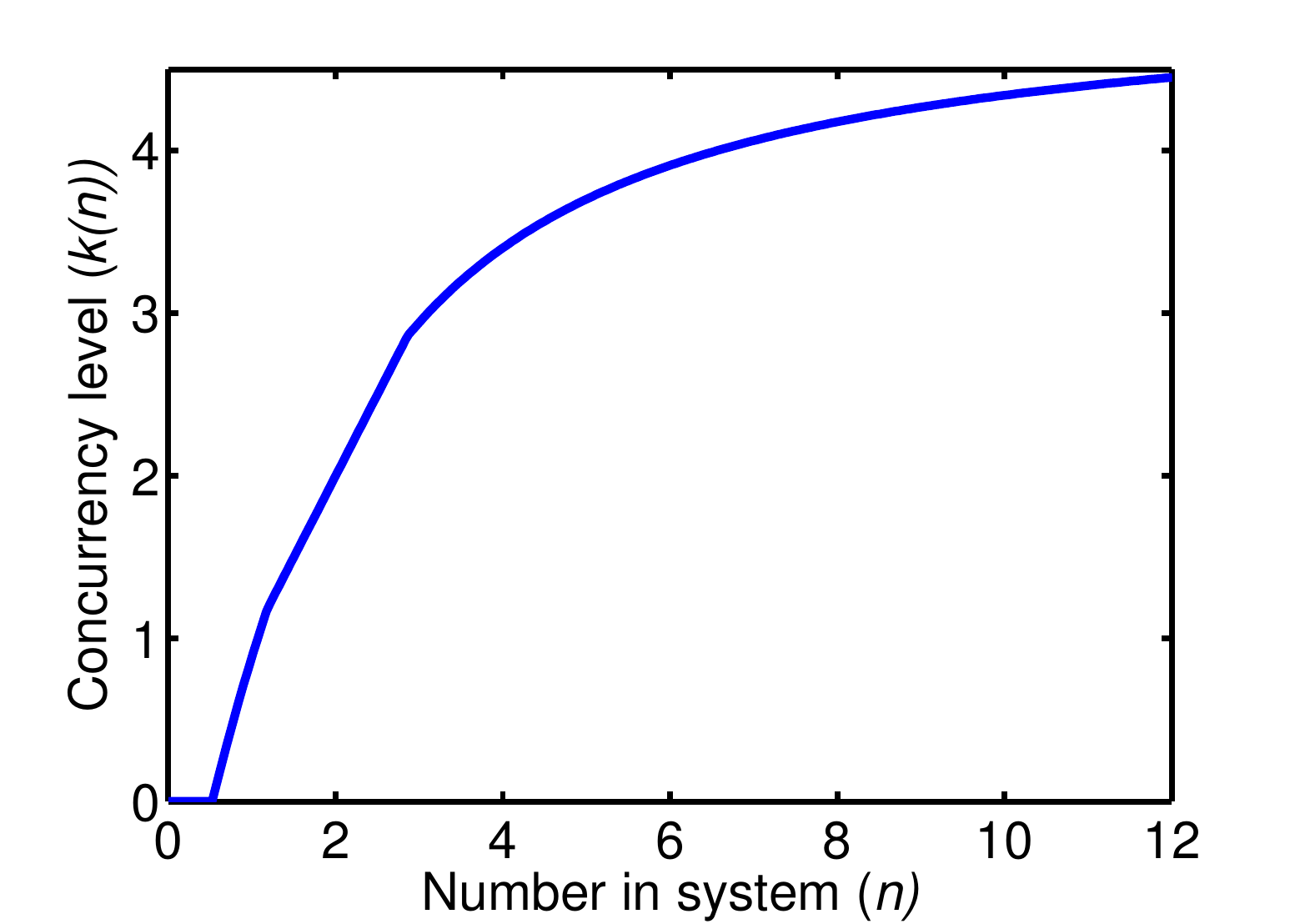}}
\end{minipage}
\caption{A hypothetical drift function $\theta(x)$ and the optimal diffusion control policies for two choices of workload parameters $c^2_a,c^2_s$. }
\label{fig:sample_policies}
\end{figure}
For an illustration of what an optimal dynamic policy might look like, see Figure~\ref{fig:sample_policies}. The first figure shows an illustrative example of the $\theta(x)$ function for the PS server. As can be seen, the PS server is most efficient when there are $\hat{k}=5$ jobs at the server, and the speed drops on either side of this point. The second figure shows the optimal dynamic policy (translated from $k(w)$ to $k(n)$, that is, as a function of the number of jobs in the system, for clarity) 
when $c^2_s=c^2_a=10$. This corresponds to a workload that has significant variability, and the optimal policy increases the concurrency level to approximately 9 when the number of jobs in the system is small but scales it back when there is a long queue. The third figure shows the policy for $c^2_s=c^2_a=0.3$. This is a low variability workload, and as the number of jobs in the system increases, initially the PS server acts as an FCFS server and thus compromises speed to keep the concurrency level small. At $n \approx 0.5$, the system switches to a controlled PS behavior by gradually increasing the concurrency level to increase service rate. At $n \approx 1.2$ the system switches to a pure PS behavior admitting everyone in queue, and finally at $n=3$ it switches back to a controlled PS behavior, gradually increasing the concurrency level to $\hat{k}=5$ as queue becomes longer. The graphs shown were produced using the Newton-Raphson average cost iteration algorithm described in Section~\ref{sec:Newton}.

\subsection{Binary search algorithm for solving the diffusion control problem} \label{sec:bin_search}
Before giving the algorithm, we discuss the main intuition and ideas behind it.
Let us assume that an oracle reveals to us the average cost $v^*$ of the optimal policy. This, together with the boundary condition $G^*(0)=0$, would allow us to numerically solve for the optimal control by evolving $G^*(\cdot)$ forward: Assuming we have solved $G^*(w)$ for $w \in [0, x]$, we first find 
\begin{equation}
\label{eqn:optimal_k}
  k^*(x) = \argmin_{k \in [0,x/m_e]} \left\{ k\left(1-\frac{m_e}{m} \right) -\theta(k) G^*(x) \right\} 
\end{equation}
and then
\[ \frac{\sigma^2}{2}{G^*}'(x) = v^* - \left[ \frac{x}{m} + k^*(x)\left( 1 - \frac{m_e}{m}\right) - \theta(k^*(x)) G^*(x) \right] \]
allows us to evolve $G^*(w)$ forward in a small enough interval $(x, x+\delta x]$.
Here then is the idea of the binary search algorithm in a nutshell: we maintain an interval $[L,U]$ within which $v^*$ is known to lie. We test if $\frac{L+U}{2}$ is the average cost of a \emph{feasible} control (we describe the feasibility test shortly). If it is, then $v^*$ is at most $\frac{L+U}{2}$ and we update $U$ to this value, otherwise we update $L$ to this value. Therefore, within $O(\log\frac{1}{\epsilon})$ iterations, we would have $(U-L) \leq \epsilon$, at which point we return a control corresponding to the average cost $U$ (which is feasible). (If we were interested in solving the discounted cost problem, the only thing to change would be to search the value of $\gamma V^*_{\gamma}(0)$ instead of the average cost.) 

\begin{description} 
\item[\bf Detecting infeasibility ($v < v^*$) : ] Let us assume that we guess a value $v $ that is smaller than the optimal average cost $v^*$ and solve the ODE \eqref{eqn:G*_ODE} forward starting from $0$. What would go wrong? It turns out that in this case $G^*(w) < 0$ for some $w>0$,  contradicting Proposition~\ref{prop:monotonic_V}. Indeed, while $v<v^*$ is infeasible for the original LPS system, it is still the cost of a feasible policy for a \emph{finite buffer LPS loss system}. If $\underline{W}(v)$ denotes the smallest $w$ at which $G^*(w)<0$, then $v$ is the optimal cost of the finite buffer system with buffer size $\underline{W}(v)$, and during the forward evolution of $G^*(\cdot)$ we have in fact found the optimal control policy for this finite buffer loss system. This provides us with a one-sided test of infeasibility: if ever $G^*(w) < 0$ for some $w>0$, our guess $v$ is too optimistic (i.e., $v<v^*$). We formalize these statements in the proposition below:
\begin{proposition}
\label{prop:infeasible_v}
Let $v^*$ be the average cost of the optimal control for the diffusion control problem \eqref{eqn:G*_ODE}, and let $v<v^*$. Let $G_v(\cdot)$ be the solution of the Bellman equation:
\begin{align}
  \label{eqn:ODE_infeasible}
  v &= \min_{k \in [0, w/m_e]} \left[ \frac{w}{m} + k \left(1 - \frac{m_e}{m} \right) - \theta(k) G_v(w) \right] 
  + \frac{\sigma^2}{2}G'_v(w)
\end{align} 
with initial condition $G_{v}(0)=0$, and $k^*_v(w)$ be the policy obtained while solving for $G_v(\cdot)$. Then $k^*_v(w)$ is the optimal control policy for a finite (workload) buffer system with buffer limit $\underline{W}(v)$, where 
\begin{align}
  \label{eq:tech-W-v}
  \underline{W}(v) &= \inf \{ w>0 : G_{v}(w) < 0\}.
\end{align}
Further, $\underline{W}(v) = O \left(\log \frac{1}{v^*-v} \right)$.
\end{proposition}
Since we do not have a priori a concrete bound on the value of $w$ for which $G_v(w)<0$ for $v<v^*$, this infeasibility test alone cannot be translated into an efficient algorithm to find $v^*$. We therefore need a test of feasibility, which is provided next. 

\item[\bf Detecting feasibility ($v > v^*$) : ] Detecting infeasibility relied on \emph{(a)} choosing a set of alternate systems (finite buffer LPS loss systems) which lower bound the cost of the original problem, but have the same HJB equation as the original system albeit with a different boundary condition ($G(\underline{W})=0$), and \emph{(b)} being able to identify in bounded time which finite buffer system our guess $v$ maps to through satisfaction of a second boundary condition (although we only expressed this bound in order notation instead of a concrete bound).

To test if a guess $v$ is larger than the true optimal cost $v^*$, we will employ a class of suboptimal policies we call \emph{fluid continuation policies}.
\begin{definition}
\label{def:fluid-cont-policy}
The set of fluid continuation policies with \emph{fluid continuation point} $W$ is defined as
\begin{align}
  \label{eqn:fluid_continuation_class}
  \mathcal{F}_W = \left\{ k \in \mathcal{K} : k(w) = k_f(w) \doteq \argmax_{x \leq w/m_e} \{\theta(x)\} , w \geq W \right\}. 
\end{align}
That is, beyond the fluid continuation workload point $W$, the control is chosen to be the most efficient service rate available. Denote the cost of the optimal (minimum cost)  policy in $\mathcal{F}_W$ by $v_f(W)$. Let $\overline{W}(v) = \min\{W \geq 0 : v = v_f(W)\}$.
\end{definition}

Clearly all policies in $\mathcal{F}_W$ are stable due to condition \eqref{eq:cond-stability-dyn}. In fact, the policy $k_f \in \mathcal F_0$ is optimal when $c^2_s=1$ since in this case $m_e=m$, and \eqref{eqn:optimal_k} simplifies to \[k^*(w) = \argmin_{k\in [0,w/m_e]} \theta(k) G^*(w) = \argmin_{k\in[0,w/m_e]} \theta(k).\] We will use the next proposition to answer the question whether or not any given guess of the average cost is feasible.

\begin{proposition}
\label{prop:feasible_v} 
Let $v^* \leq v \leq v_f(0)$. Then $v$ is the average cost of an optimal fluid continuation policy $k_v$ with continuation point $\overline{W}(v)$. That is:
\begin{align}
\label{eqn:fluid_continuation_policies}
k_{v}(w) &= \begin{cases}
\argmin_{k \in [0, w/m_e]} \left\{ k\left(1 - \frac{m_e}{m} \right) -\theta(k) G_{v}(w) \right\} & w \leq \overline{W}(v), \\
k_f(w) & w > \overline{W}(v),
\end{cases}
\intertext{where $G_{v}$ is the value function gradient for policy $k_{v}$ and satisfies the ODE}
\label{eq:ODE-contingency}
v &=  \frac{w}{m} + k_{v}(w) \left(1 - \frac{m_e}{m} \right) - \theta(k_{v}(w)) {G_{v}}(w)   + \frac{\sigma^2}{2} G'_{v}(w). 
\end{align}
Further, $\overline{W}(v) = O\left(\log \frac{1}{v-v^*}\right)$.
\end{proposition}

The advantage of the class of feasible policies described by \eqref{eqn:fluid_continuation_policies} is that for any $v$, the function $G_{v}(w)$ for $w \geq \overline{W}(v)$ can be easily computed, and is in fact independent of $\overline{W}(v)$. Let us call this the \emph{fluid continuation of the value function gradient} and denote it by $\overline{G}_v(w)$. The function $\overline{G}_v$ will act as the boundary condition for detecting the feasibility of $v$.

$\overline{G}_v(w)$ solves the ODE
\begin{align}
\label{eqn:Gf_ODE}
v &= \frac{w}{m} + k_{f}(w) \left(1 - \frac{m_e}{m} \right) - \theta(k_{f}(w)) {\overline{G}_{v}}(w)   + \frac{\sigma^2}{2} {\overline{G}}'_{v}(w)
\end{align}
with an as-yet-unspecified boundary condition. Note that $\overline{G}_v(0)$ is not necessarily 0, unless $v$ happens to be the average cost of policy $k_f$ (that is, $v=v_f(0)$). We first consider the range $w \in [\hat{k}m_e, \infty)$ where $k_f(w)=\hat{k}$. In this part, the ODE for $\overline{G}_v$ becomes
\begin{align}
v &= \frac{w}{m} + \hat{k} \left(1 - \frac{m_e}{m} \right) - \hat{\theta} {\overline{G}_{v}}(w)   + \frac{\sigma^2}{2} \overline{G}'_{v}(w)m,
\end{align}
which is a first order non-homogeneous differential equation with constant coefficients, and a non-zero part that is a linear function of $w$. The solution to \eqref{eqn:Gf_ODE} has a homogeneous (general) part and a particular solution with two unknown constants. To determine the unknown constants we need two boundary conditions, one of which is $\overline{G}_v(0)=0$. We obtain our equivalent of the second boundary condition by setting the coefficient of the homogeneous part of the solution (which is $e^{\frac{2 \hat{\theta}w}{\sigma^2}}$ ) to zero because the value function cannot grow exponentially. This gives
\begin{align}
\overline{G}_v(w) &=  \frac{w}{m \hat{\theta}} + \left( \hat{k}\left(1-\frac{m_e}{m} \right) + \frac{\sigma^2}{2m\hat{\theta}} - v \right)\frac{1}{\hat{\theta}},
\quad  w \geq \hat{k} m_e.
\end{align} 
To solve for $\overline{G}_v$ for $w\in [0, \hat{k}m_e]$, we solve the ODE \eqref{eqn:Gf_ODE} backwards starting with the terminal condition 
\[ \overline{G}_v( \hat{k}m_e) = \left(\hat{k} - v + \frac{\sigma^2}{2m\hat{\theta}}\right) \frac{1}{\hat{\theta}}\ .\]  
If it turns out that $\overline{G}_v(0) < 0$, then our guess $v > v_f(0)$ and therefore $v$ is the average cost of some feasible policy. Otherwise, we solve the ODE \eqref{eqn:G*_ODE} forward by starting with initial condition $G_{v}(0)=0$ and substituting our guess $v^* = v$. If $G_{v}$ `hits' $\overline{G}_v$ from below, that is $G_{v}(W)=\overline{G}_v(W)$ for some $W$, then $\overline{W}(v)=W$, and following policy $k_{v}(w)$ for $w\in[0,\overline{W}(v)]$ and $k_f$ for $w \geq \overline{W}(v)$ is a policy with average cost $v$, indicating feasibility of $v$. 
\end{description}
The step-by-step procedure is given in Algorithm~\ref{alg:binary_search}.

\begin{algorithm}
\caption{Average cost iteration (binary search method)}\label{alg:binary_search}
\algblock{Solve}{EndSolve}
\algnewcommand\algorithmicsolve{\textbf{solve}}
\algnewcommand\algorithmicendsolve{\textbf{end\ solve}}
\algrenewtext{Solve}{\algorithmicsolve}
\algrenewtext{EndSolve}{\algorithmicendsolve}

\begin{algorithmic}[0]
\State {\bf define} $\hat{k} \doteq \argmax_k \theta(k)$; $\hat{\theta} \doteq \theta(\hat{k})$
\State {\bf define} $k_f(w) \doteq \argmax_{k \in[0,w/m_e]} \theta(k)$; $\theta_f(w) \doteq \theta(k_f(w))$ \Comment(Fluid optimal policy)
\State {\bf initialize} $L \gets 0, U \gets \emptyset, v \gets 1$ \Comment(Search interval $[L,U]$ and initial guess)
\While{$U-L \geq \epsilon$ }
	\Solve{ for the fluid continuation $\overline{G}_v(w)$ with average cost $v$:}
	\State $\overline{G}_v(w) = \frac{w}{m \hat{\theta}} + \left( \hat{k}\left(1-\frac{m_e}{m} \right) + \frac{\sigma^2}{2m \hat{\theta}} - v \right)\frac{1}{\hat{\theta}}  \quad \ldots \ w \in [\hat{k}m_e, \infty)$ \Comment(Terminal condition)
	\State $ v = \frac{w}{m}+k_f(w) (1-\frac{m_e}{m}) -\theta_f(w) \overline{G}_v(w) + \frac{\sigma^2}{2}\overline{G}'_v(w) \quad \ldots \ w \in [0, \hat{k}m_e]$ \Comment(ODE)
	\EndSolve
\If{$\overline{G}_v(0) < 0$}	\Comment($v$ is larger than avg. cost of $k_f$)
	\State $U \gets v$ \Comment(Therefore, $v^* \leq v$)
	\State $v \gets \frac{L+U}{2}$ \Comment(Updated guess for next iteration)
\Else
	\Solve{ for policy $k_{v}(w)$ and $G_{v}(w)$:}
		\State $G_{v}(0) = 0$ \Comment(Initial condition)
		\State $ v = \min_{k \in [0,w/m_e]}\left\{  \frac{w}{m}+k (1-\frac{m_e}{m}) -\theta(k) G_{v}(w) + \frac{\sigma^2}{2}G'_{v}(w) \right\} $ \Comment(ODE)
		\State {\bf until} $W = \inf\{ w:  (G_{v}(w) \geq  \overline{G}_v(w)) \mbox{ OR } ( G_{v}(w) < -\epsilon) \}$ \Comment(Terminal event)
	\EndSolve
	\If{ $G_{v}(W) \geq \overline{G}_v(W)$} \Comment($v$ is feasible)
		\State $U \gets v$ \Comment(Therefore, $v^* \leq v$)
		\State $v \gets \frac{L+U}{2}$ \Comment(Update guess for next iteration)
	\Else	\Comment($v$ is infeasible)
		\If{$U=\emptyset$}
			\State $v \gets 2v$ \Comment(Double the guess until we find one feasible value)
		\Else
			\State $L \gets v$	\Comment($v^* \geq v$)
			\State $v \gets \frac{L+U}{2}$ \Comment(Update guess for next  iteration)
		\EndIf	
	\EndIf 
\EndIf
\EndWhile
\State \Return Cost $v=U$; Policy $k_U(w)$
\end{algorithmic}
\end{algorithm}

\subsection{Newton-Raphson method for solving the diffusion control problem} \label{sec:Newton}
The binary search algorithm we proposed in Section~\ref{sec:bin_search} was based on first guessing an average cost value $v$, forward evolving ODE \eqref{eqn:G*_ODE} with the (initial) boundary condition $G_{v}(0)=0$ until a terminal boundary condition was met thereby verifying feasibility or infeasibility of $v$ as the average cost, and then updating the guess for $v$. The algorithm we propose in this section will be based on backward evolution of ODE \eqref{eqn:G*_ODE}.

As in Section~\ref{sec:bin_search}, here we will find an optimal policy in the class of fluid continuation policies (Definition~\ref{def:fluid-cont-policy}) $\mathcal{F}_W$ for some fluid continuation point $W$. However, this time we will first fix a large enough value of $W$ ($W \geq \hat{k}m_e$) and seek the optimal policy and the optimal average cost in $\mathcal{F}_W$.
Recall that we denote the average cost of the optimal policy in $\mathcal{F}_W$ by $v_f(W)$.
As mentioned later, we can use a standard doubling trick to settle on a `large enough' $W$.
Next, we will guess an average cost value $v$ and devise a test to compare $v$ with $v_f(W)$. For this, we evolve ODE \eqref{eqn:G*_ODE} backwards with the (terminal) boundary condition
\begin{align}
\label{eqn:NR_terminal_condition}
G_{v}(W) &= \left(\hat{k} \left(1- \frac{m_e}{m} \right) - v + \frac{\sigma^2}{2m\hat{\theta}}\right) \frac{1}{\hat{\theta}}  + \frac{1}{m\hat{\theta}}\cdot W 
\end{align}
(For notational simplicity, we have suppressed the dependence of $G_v()$ on $W$). If indeed $v=v_f(W)$ then we must have $G_v(0)=0$, and the sign of $G_v(0)$ can tell us if $v<v_f(W)$ or $v>v_f(W)$. This would be similar to the binary search iterative algorithm with a linear convergence rate. However, since we know $G_{v_f(W)}(0)=0$, we can (and will) find $v_f(W)$ by solving for the root of the equation $G_{v}(0)=0$ (in $v$) using the Newton-Raphson method which has a faster quadratic convergence rate.

Recall that to solve for the root of a function $h(x)=0$ with a current estimate of $x_n$, the Newton-Raphson update step is given by
\[ x_{n+1} = x_n - \frac{h(x_n)}{h'(x_n)}.\]
Let us assume that our current guess for $v_f(W)$ is $v_n$. To generate the next guess $v_{n+1}$ via the Newton-Raphson method, we need the derivative of $G_{v}(0)$ at $v=v_n$. With some abuse of notation, define
\[ g_{v}(w) \doteq \frac{d G_{v}(w)}{dv} .\]
(What we really mean by the above is that $g_{v}(w) \doteq \frac{\partial G(v,w)}{\partial v}$, where $G(v,w) = G_v(w)$.) As we will show in the proof of Proposition~\ref{prop:NR_convergence} (see step 2 of the proof), $G_v(w)$ is Lipschitz continuous and decreasing in $v$ for all $w$ and therefore $g_v(w)$ exists almost everywhere, and further it is bounded away from 0.
With $W \geq \hat{k}m_e$ representing the point at which we switch to the fluid policy $k_f$, we can write $G_{v}(w)$ as the following integral: for $w\leq W$, 
\begin{align}
  \label{eqn:G_integral}
  G_{v_n}(w) &= G_{v_n}(W) + \frac{2}{\sigma^2}\int_{W}^w 
  \left[ v_n - \min_{k \in [0,u/m_e]}  \left\{ \Delta_k(u) - \theta(k) G_{v_n}(u)  \right\} \right] du,
\end{align}
where $G_{v_n}(W)$ is given by \eqref{eqn:NR_terminal_condition}. Differentiating \eqref{eqn:G_integral} with respect to $v_n$ yields
\begin{align*}
  g_{v_n}(w) &= g_{v_n}(W) + \frac{2}{\sigma^2}
  \int_{W}^w \left[ 1 - \frac{d}{dv_n} \min_{k \in [0,u/m_e]}  
    \left\{ \Delta_k(u) - \theta(k) G_{v_n}(u)  \right\} \right] du \\
  &= -\frac{1}{\hat{\theta}} + \frac{2}{\sigma^2}
  \int_{W}^w \left[ 1 + \theta(k_{v_n}(u)) g_{v_n}(u) \right] du.
\end{align*}
Since the policy $k_{v_n}()$ also depends on $v_n$, to arrive at the last equality, we have used the \emph{envelope theorem}: If $k^*(v) = \argmin_{k} \phi(k,v)$ and $\phi^*(v)=\phi(k^*(v),v)$, then $\frac{d \phi^*(v)}{dv} = \frac{\partial \phi(k^*(v),v)}{\partial v}$ (where $\frac{\partial \phi(k,v)}{\partial v}$ is the partial derivative with respect to $v$).
Therefore, very similar to $G_{v_n}$, $g_{v_n}$ satisfies the following ODE
\begin{align}
1 &=  - \theta(k_{v_n}(w)) {g_{v_n}}(w)  + \frac{\sigma^2}{2} g'_{v_n}(w)
\intertext{with the terminal condition}
\label{eqn:NR_g_terminal_condition}
g_{v_n}(W) &= -\frac{1}{\hat{\theta}}.
\end{align}
The updated guess for average cost is then
\begin{align}
\label{eqn:NR_equation}
v_{n+1} = v_n - \frac{G_{v_n}(0)}{g_{v_n}(0)}.
\end{align}
Remarkably, it turns out that $v_{n+1}$ is exactly the average cost of the policy, call it $k_{v_n}(w)$, that is implicitly generated when solving for $G_{v_n}$ and $g_{v_n}$. This is because if we fix the policy  $k(w)=k_{v_n}(w)$, then from \eqref{eqn:G_integral} we can see that $G_{v}$ is linear in $v$. Therefore the Newton-Raphson update is effectively solving for that $v$ for which $G_{v}(0)=0$ when $k(w)=k_{v_n}(w)$, which is the average cost of $k_{v_n}$. Therefore the sequence of average cost iterates produced by the algorithm are in fact the average costs of a sequence of feasible policies. The next proposition formally states the result on convergence of the Newton-Raphson average cost iteration algorithm.
\begin{proposition}
\label{prop:NR_convergence}
Let $v_1, v_2, \ldots$ denote the average cost iterates generated by the Newton-Raphson method \eqref{eqn:NR_equation}. Let 
\[ d_\theta \doteq \sup_k \theta(k) - \inf_k \theta(k) < \infty. \]
The sequence $\{v_n\}$ monotonically decreases to $v_f(W)$, which is the average cost of the optimal diffusion control policy in the set $\mathcal{F}_W$ of fluid continuation policies with fluid continuation point $W$. Further, assuming that $\theta(k)$ is twice differentiable everywhere, and that 
\begin{enumerate}
\item the first derivative of $\theta(k)$ is finite, i.e., 
  \begin{equation}
    \label{eq:tech-1st-der-bd}
    S_\theta \doteq \sup_k \left| \frac{d\theta(k)}{dk} \right| < \infty,
  \end{equation}
\item the second derivative of $\theta(k)$ is bounded away from 0, i.e., 
  \begin{equation}
    \label{eq:tech-2nd-der-bd}
    D_\theta \doteq \inf_k \left| \frac{d^2\theta(k)}{dk^2} \right| > 0,
  \end{equation}
\end{enumerate}
the errors of the Newton-Raphson iterates, $\epsilon_n = (v_n-v_f(W))$, decrease quadratically.
\end{proposition} 
Since close to the root, the error roughly squares in each iteration, it takes $O(\log \log \frac{1}{\epsilon})$ iterations to reach an $\epsilon$-optimal policy within $\mathcal{F}_W$. To find the $\epsilon$-optimal policy among all policies, we can keep doubling the value of $W$ until the error between successive iterates is sufficiently small. By our earlier result, we need a $W = O(\log \frac{1}{\epsilon})$ to arrive at an $\epsilon$-optimal policy. Since each iteration of the Newton-Raphson method takes $O(W)$ time, the overall time complexity of the algorithm to find an $\epsilon$-optimal policy is $O\left( \log\frac{1}{\epsilon}  \log \log \frac{1}{\epsilon}\right)$. On the other hand, the overall time complexity of the binary search algorithm is $O\left( \left(\log\frac{1}{\epsilon}\right)^2 \right)$.

The step-by-step procedure is described in Algorithm~\ref{alg:Newtons_method}. In the description, we have omitted iterating over values of $W$, the fluid continuation point, for clarity and to focus on the core of the algorithm. Note that we choose the average cost of the fluid policy as the initial guess for average cost $v_0$ which is computed using a single step of Newton-Raphson iteration (shown in the initialize block). This is a minor optimization that also takes care of a corner case in the proof of convergence of Algorithm~\ref{alg:Newtons_method}, although we believe this initialization is not necessary for quadratic convergence to hold.

\begin{algorithm}
\caption{Average cost iteration (Newton-Raphson method)}\label{alg:Newtons_method}
\algblock{Solve}{EndSolve}
\algnewcommand\algorithmicsolve{\textbf{solve}}
\algnewcommand\algorithmicendsolve{\textbf{end\ solve}}
\algrenewtext{Solve}{\algorithmicsolve}
\algrenewtext{EndSolve}{\algorithmicendsolve}
\algblock{Init}{EndInit}
\algnewcommand\algorithmicinit{\textbf{initialize}}
\algnewcommand\algorithmicendinit{\textbf{end\ initialize}}
\algrenewtext{Init}{\algorithmicinit}
\algrenewtext{EndInit}{\algorithmicendinit}
\begin{algorithmic}[0]
\State {\bf define} $\hat{k} \doteq \argmax_k \theta(k)$; $\hat{\theta} \doteq \theta(\hat{k})$
\State {\bf require} $W \geq \hat{k}m_e$ \Comment(Fluid continuation point)
\Init \Comment(Compute initial guess for average cost $v_f(0)$, see Defn. \ref{def:fluid-cont-policy})
	\Solve{ functions $G_{f}(w)$ and $g_{f}(w)$ for $w\in[0, \hat{k}m_e]$:}
		\State $G_{f}(\hat{k}m_e) = \left(\hat{k} (1-m_e/m)  + \frac{\sigma^2}{2m\hat{\theta}}\right) \frac{1}{\hat{\theta}}  + \frac{1}{m\hat{\theta}}\cdot \hat{k}m_e $  \Comment(Terminal condition for $G_{f}$)
		\State $g_{f}(\hat{k}m_e) = -\frac{1}{\hat{\theta}}$ \Comment(Terminal condition for $g_{f}$)
		\State $k_{f}(w) = \argmax_{k \in [0,w/m_e]} \theta(k) $ \Comment(Fluid optimal policy)
		\State $ 0 = \frac{w}{m}+k_{f}(w) (1-\frac{m_e}{m}) -\theta(k_{f}(w)) G_{f}(w) + \frac{\sigma^2}{2}G'_{f}(w) $ \Comment(ODE for $G_{f}$)
		\State $1 = -\theta(k_{f}(w)) g_{f}(w) + \frac{\sigma^2}{2} g_{f}'(w)$ \Comment(ODE for $g_{f}$)
	\EndSolve
\State 	$v_0 \gets v_f(0) = - \frac{G_f(0)}{g_f(0)}$
\EndInit
\Repeat
	\Solve{ policy $k_{v_n}(w)$, functions $G_{v_n}(w)$ and $g_{v_n}(w)$ for $w\in[0,W]$:}
		\State $G_{v_n}(W) = \left(\hat{k} (1-m_e/m) - v_n + \frac{\sigma^2}{2m\hat{\theta}}\right) \frac{1}{\hat{\theta}}  + \frac{1}{m\hat{\theta}}\cdot W $  \Comment(Terminal condition for $G_{v_n}$)
		\State $g_{v_n}(W) = -\frac{1}{\hat{\theta}}$ \Comment(Terminal condition for $g_{v_n}$)
		\State $k_{v_n}(w) = \argmin_{k \in [0,w/m_e]} \left\{ k\left(1-\frac{m_e}{m}\right) - \theta(k) G_{v_n}(w)  \right\}$
		\State $ v_n = \frac{w}{m}+k_{v_n}(w) (1-\frac{m_e}{m}) -\theta(k_{v_n}(w)) G_{v_n}(w) + \frac{\sigma^2}{2}G'_{v_n}(w) $ \Comment(ODE for $G_{v_n}$)
		\State $1 = -\theta(k_{v_n}(w)) g_{v_n}(w) + \frac{\sigma^2}{2} g_{v_n}'(w)$ \Comment(ODE for $g_{v_n}$)
	\EndSolve
\State {\bf update } $v_{n+1} \gets v_n - \frac{G_{v_n}(0)}{g_{v_n}(0)} $ \Comment(Newton-Raphson update)
\Until{$|G_{v_n}(0)| \leq \epsilon$}
\State \Return Cost $v_{n+1}$; Policy $k_{v_n}(w)$

\end{algorithmic}
\end{algorithm}


\textbf{Comparison with the policy iteration algorithm:} Puterman and Brumelle \cite{PutermanBrumelle1979convergence} formally proved that the policy iteration algorithm for discrete-time Markov decision processes is equivalent to the Newton-Raphson algorithm for finding the fixed point of the dynamic programming functional equation, but performed in the value function space. Puterman \cite{Puterman1977control} presented a policy iteration algorithm for control of a diffusion process in a bounded region in $\Re^n$ for finite horizon total cost optimization. It is therefore instructive to compare our average cost iteration algorithm with his policy iteration algorithm for control of diffusions. One difference is that we carry out the Newton-Raphson algorithm in the space of average cost. Another major difference is that the policy iteration algorithm alternates between policy evaluation and policy improvement steps. Our algorithm can be viewed as one where we have folded the policy evaluation and policy improvement into one step.

\section{Concluding Remarks}
\label{sec:conclusions}
The primary goal of the present paper was to propose a diffusion scaling to aid the analysis and control of State-dependent Limited Processor Sharing (LPS) systems. 
Our philosophy while designing the scaling was to fix a limiting distribution for the steady-state number of jobs in the system, and then reverse-engineer the sequence of service rate curves that yields this limit. By choosing the limiting distribution as the one of the original state-dependent system under an $M/M/$ input, our scaling captures the effect of the entire service rate curve. 
The resulting diffusion approximation leads to the choice of a near-optimal static concurrency limit.

To compute dynamic control policies, we generalized our scaling by defining it directly in terms of the service rate curves of the original system. Again, as proof-of-concept, we presented experimental results demonstrating that the dynamic policies resulting from the associated diffusion control problem are extremely close to the optimal dynamic control policies.

While carrying out our experiments for dynamic policies, we realized that there were no numerical algorithms for solving diffusion control problems in the literature that we could conveniently use for our unique setting. We thus devised two new algorithms which iterate over the average cost. The first algorithm uses binary search and thus has a linear convergence rate. The second algorithm uses the Newton-Raphson method for finding roots, and thus has a superior quadratic convergence rate. We believe that these algorithms are by themselves very useful.

\bibliographystyle{plain}
\bibliography{references}

\begin{thebibliography}{10}

\bibitem{AdusumilliHasenbein_2010}
Kranthi~Mitra Adusumilli and John~J. Hasenbein.
\newblock Dynamic admission and service rate control of a queue.
\newblock {\em Queueing Syst.}, 66(2):131--154, 2010.

\bibitem{Agrawal85}
Rakesh Agrawal, Michael~J. Carey, and Miron Livny.
\newblock Models for studying concurrency control performance: alternatives and
  implications.
\newblock {\em SIGMOD Rec.}, 14(4):108--121, 1985.

\bibitem{AtaShneorson_2006}
Baris Ata and Shiri Shneorson.
\newblock Dynamic control of an {M/M/1} service system with adjustable arrival
  and service rates.
\newblock {\em Management Science}, 52(11):1778--1791, 2006.

\bibitem{AH88}
B.~Avi-Itzhak and S.~Halfin.
\newblock Expected response times in a non-symmetric time sharing queue with a
  limited number of service positions.
\newblock {\em In Proceedings of ITC}, 12:5.4B.2.1--7, 1988.

\bibitem{BattTerwiesch2012}
Robert~J. Batt and Christian Terwiesch.
\newblock Doctors under load: An empirical study of state-dependent service
  times.
\newblock 2012.

\bibitem{BertsekasDP}
D.~Bertsekas.
\newblock {\em Dynamic Programming and Optimal Control}, volume 1-2.
\newblock Athena Scientific, 3rd edition, 2007.

\bibitem{Blake82}
Russ Blake.
\newblock Optimal control of thrashing.
\newblock In {\em Proceedings of the 1982 ACM SIGMETRICS Conference on
  Measurements and Modeling of Computer Systems}, 1982.

\bibitem{Bramson1998}
Maury Bramson.
\newblock State space collapse with application to heavy traffic limits for
  multiclass queueing networks.
\newblock {\em Queueing Syst.}, 30(1-2):89--148, 1998.

\bibitem{Brumelle1971}
Shelby~L. Brumelle.
\newblock Some inequalities for parallel-server queues.
\newblock {\em Operations Research}, 19(2):402--413, 1971.

\bibitem{BudhirajaGhosh2006}
Amarjit Budhiraja and Arka~Prasanna Ghosh.
\newblock Diffusion approximations for controlled stochastic networks: an
  asymptotic bound for the value function.
\newblock {\em Ann. Appl. Probab.}, 16(4):1962--2006, 2006.

\bibitem{DR84}
D.J. Daley and T.~Rolski.
\newblock Some comparibility results for waiting times in single- and
  many-server queues.
\newblock {\em J. Appl. Prob.}, 21:887--900, 1984.

\bibitem{Denningetal76}
Peter~J. Denning, Kevin~C. Kahn, Jacques Leroudier, Dominique Potier, and Rajan
  Suri.
\newblock Optimal multiprogramming.
\newblock {\em Acta Informatica}, 7:197--216, 1976.

\bibitem{Elniketyetal04}
Sameh Elnikety, Erich Nahum, John Tracy, and Willy Zwaenepoel.
\newblock A method for trasparent admission control and request schedulin in
  e-commerce web sites.
\newblock In {\em World-Wide-Web Conference}, 2004.

\bibitem{GamarnikMomcilovic08}
David Gamarnik and Petar Mom\v{c}ilovi\'{c}.
\newblock Steady-state analysis of a multiserver queue in the halfin-whitt
  regime.
\newblock {\em Adv. Appl. Probab.}, 40(2):548--577, 2008.

\bibitem{GamarnikZeevi2006}
David Gamarnik and Assaf Zeevi.
\newblock Validity of heavy traffic steady-state approximation in generalized
  {J}ackson networks.
\newblock {\em Ann. Appl. Probab.}, 16(1):56--90, 2006.

\bibitem{GeorgeHarrison_2001}
Jennifer~M. George and J.~Michael Harrison.
\newblock Dynamic control of a queue with adjustable service rate.
\newblock {\em Operations Research}, 49(5):pp. 720--731, 2001.

\bibitem{PSMPL_paper}
Varun Gupta and Mor Harchol-Balter.
\newblock Self-adaptive admission control policies for resource-sharing
  systems.
\newblock In {\em Proceedings of ACM SIGMETRICS '09}, pages 311--322, New York,
  NY, USA, 2009.

\bibitem{HalfinWhitt1981}
Shlomo Halfin and Ward Whitt.
\newblock Heavy-traffic limits for queues with many exponential servers.
\newblock {\em Operations research}, 29(3):567--588, 1981.

\bibitem{HST1983}
J.~Michael Harrison, Thomas~M. Sellke, and Allison~J. Taylor.
\newblock Impulse control of {B}rownian motion.
\newblock {\em Math. Oper. Res.}, 8(3):454--466, 1983.

\bibitem{HarrisonTaksar1983}
J.~Michael Harrison and Michael~I. Taksar.
\newblock Instantaneous control of {B}rownian motion.
\newblock {\em Math. Oper. Res.}, 8(3):439--453, 1983.

\bibitem{HeissWagner91}
Hans-Ulrich Heiss and Roger Wagner.
\newblock Adaptive load control in transaction processing systems.
\newblock In {\em Proceedings of the 17th International Conference on Large
  Data Bases (VLDB)}, 1991.

\bibitem{Hernandez-LermaDiscreteMCP}
On\'{e}simo Hern\'{a}ndez-Lerma and Jean~Bernard Lasserre.
\newblock {\em Discrete-Time Markoc Control Processes}.
\newblock Springer, 1995.

\bibitem{JanssenLS_2013}
A.J.E.M. Janssen, J.S.H.~van Leeuwaarden, and Jaron Sanders.
\newblock Scaled control in the qed regime.
\newblock {\em Performance Evaluation}, 70(10):750--769, 2013.

\bibitem{KarlinTaylor_Second}
Samuel Karlin and Howard~M Taylor.
\newblock {\em A Second Course in Stochastic Processes}.
\newblock Academic Press, 1981.

\bibitem{Kollerstrom74}
Julian K{\"o}llerstr{\"o}m.
\newblock Heavy traffic theory for queues with several servers. {I}.
\newblock {\em J. Appl. Prob.}, 11:544--552, 1974.

\bibitem{KrichaginaPuhalskii1997}
E.~V. Krichagina and A.~A. Puhalskii.
\newblock A heavy-traffic analysis of a closed queueing system with a
  {$GI/\infty$} service center.
\newblock {\em Queueing Syst.}, 25(1-4):235--280, 1997.

\bibitem{Krichagina_1992}
E.V. Krichagina.
\newblock Asymptotic analysis of queueing networks.
\newblock {\em Stochastics and Stochastic Reports}, 40:43--76, 1992.

\bibitem{KushnerDupuis_NumericalMethods}
Harold~J. Kushner and Paul Dupuis.
\newblock {\em Numerical Methods for Stochastic Control Problems in Continuous
  Time}.
\newblock Springer, 2001.

\bibitem{LeePuhalskii_2012}
Chihoon Lee and Anatolii~A. Puhalskii.
\newblock Non-markovian state dependent networks in critical loading.
\newblock arXiv preprint arXiv:1212:4078, 2012.

\bibitem{LeeWeerasinghe2011}
Chihoon Lee and Ananda Weerasinghe.
\newblock Convergence of a queueing system in heavy traffic with general
  patience-time distributions.
\newblock {\em Stochastic Processes and their Applications},
  121(11):2507--2552, 2011.

\bibitem{LeeKulkarni_2014}
Nelson Lee and VidyadharG. Kulkarni.
\newblock Optimal arrival rate and service rate control of multi-server queues.
\newblock {\em Queueing Syst.}, 76(1):37--50, 2014.

\bibitem{MandelbaumPats_1998}
Avi Mandelbaum and Gennady Pats.
\newblock State-dependent stochastic networks. part i. approximations and
  applications with continuous diffusion limits.
\newblock {\em Ann. Appl. Probab.}, 8(2):569--646, 05 1998.

\bibitem{Mandl_analytical}
Petr Mandl.
\newblock {\em Analytical Treatment of One-Dimensional Markov Processes}.
\newblock Academia, Springer, 1968.

\bibitem{NairWZ2010}
Jayakrishnan Nair, Adam Wierman, and Bert Zwart.
\newblock Tail-robust scheduling via limited processor sharing.
\newblock {\em Perform. Eval.}, 67(11):978--995, 2010.

\bibitem{PutermanBrumelle1979convergence}
Martin~L. Puterman and Shelby~L. Brumelle.
\newblock On the convergence of policy iteration in stationary dynamic
  programming.
\newblock {\em Mathematics of Operations Research}, 4(1):pp. 60--69, 1979.

\bibitem{Puterman1977control}
M.L. Puterman.
\newblock Optimal control of diffusion processes with reflection.
\newblock {\em Journal of Optimization Theory and Applications},
  22(1):103--116, 1977.

\bibitem{Reed2009}
Josh~E Reed.
\newblock The {$G/GI/N$} queue in the {H}alfin-{W}hitt regime.
\newblock {\em Ann. Appl. Probab.}, 19(6):2211--2269, 2009.

\bibitem{RegeSengupta85}
K.M. Rege and M.~Sengupta.
\newblock Sojourn time distribution in a multiprogrammed computer system.
\newblock {\em AT\&T Tech. J.}, 64:1077--1090, 1985.

\bibitem{WardKumar_2008}
Amy~R. Ward and Sunil Kumar.
\newblock Asymptotically optimal admission control of a queue with impatient
  customers.
\newblock {\em Mathematics of Operations Research}, 33(1):pp. 167--202, 2008.

\bibitem{Welsh01}
Matt Welsh, David Culler, and Eric Brewer.
\newblock Seda: an architecture for well-conditioned, scalable internet
  services.
\newblock {\em SIGOPS Oper. Syst. Rev.}, 35(5):230--243, 2001.

\bibitem{Yamada_1995}
Keigo Yamada.
\newblock Diffusion approximation for open state-dependent queueing networks in
  the heavy traffic situation.
\newblock {\em Ann. Appl. Probab.}, 5(4):958--982, 1995.

\bibitem{YamazakiS1987}
Genji Yamazaki and Hirotaka Sakasegawa.
\newblock An optimal design problem for limited processor sharing systems.
\newblock {\em Management Science}, 33(8):pp. 1010--1019, 1987.

\bibitem{ZhangZwart08}
J.~Zhang and B.~Zwart.
\newblock Steady state approximations of limited processor sharing queues in
  heavy traffic.
\newblock {\em Queueing Syst.}, 60:227--246, 2008.

\bibitem{ZDZ2009}
Jiheng Zhang, J.~G. Dai, and Bert Zwart.
\newblock {Law of Large Number Limits of Limited Processor-Sharing Queues}.
\newblock {\em Math. Oper. Res.}, 34(4):937--970, 2009.

\bibitem{ZhangDaiZwart2011}
Jiheng Zhang, J.~G. Dai, and Bert Zwart.
\newblock {Diffusion Limits of Limited Processor-Sharing Queues}.
\newblock {\em Ann. Appl. Probab.}, 21(2):745--799, 2011.

\end{thebibliography}

\appendix

\section{Diffusion and Steady State Analysis for the Workload Processes}

Following from the dynamic equation \eqref{eq:workload-dynamics}, the diffusion-scaled workload is
\begin{equation}\label{eq:workload-diffusion}
  \ds W(t) = \ds W(0)+\frac{1}{r}\sum_{i=1}^{\ssp{\Lambda}(r^2t)}\ssp{v}_i 
  - \frac{1}{r}\int_0^{r^2t}\ssp{\mu}(\ssp{Z}(s))\id{\ssp{W}(s)>0}ds
\end{equation}
Now, introduce the notations
\begin{align}
  \label{eq:Kiefer-fs}
  \fls K(t,x) &= \frac{1}{r^2}\sum_{i=1}^{\cl{r^2t}}\id{\ssp{v}_i\le x},\\
  \label{eq:Kiefer-ds}
  \ds K(t,x) &= r[\fls K(t,x) - t G(x)].
\end{align}
The second term on the right-hand side of \eqref{eq:workload-diffusion} can be written as
\begin{align*}
  &\quad r\int_0^t\int_0^\infty x d\fls K(\frac{1}{r^2}\ssp{\Lambda}(r^2s),x)\\
  &=\int_0^t\int_0^\infty x d\ds K(\frac{1}{r^2}\ssp{\Lambda}(r^2s),x)
  +\frac{1}{r}\int_0^t\int_0^\infty x dG(x) d\ssp{\Lambda}(r^2s)\\
  &=\int_0^t\int_0^\infty x d\ds K(\frac{1}{r^2}\ssp{\Lambda}(r^2s),x)
  + m\int_0^t d \frac{1}{r}[\ssp{\Lambda}(r^2s)-\lambda r^2s]+\lambda rmt.
\end{align*}
The last term on the right-hand side of \eqref{eq:workload-diffusion} can be written as
\begin{align*}
  &\quad - r\int_0^t \ssp{\mu}(r\ds Z(s)) \id{\ssp{W}(r^2s)>0}ds\\
  &= \int_0^t r[\lambda m-\ssp{\mu}(r\ds Z(s))] ds + r\int_0^t \id{\ds W(s)=0}ds-\lambda rmt\\
  &= \int_0^t r[\lambda m-\ssp{\mu}(r\Delta(\ds W(s))\wedge \ssp{k})] ds 
  +  \int_0^t r[\ssp{\mu}(r\ds Z(s))-\ssp{\mu}(r\Delta(\ds W(s)) \wedge \ssp{k})] ds \\
  &\quad  + r\int_0^t \id{\ds W(s)=0}ds-\lambda rmt\\
  &= \int_0^t \ssp{\theta}\left(\Delta(\ds W(s)) \wedge \frac{\ssp{k}}{r}\right) -
              \theta\left(\Delta(\ds W(s)) \wedge \frac{\ssp{k}}{r}\right)  ds
  + r\int_0^t \id{\ds W(s)=0}ds-\lambda rmt \\
  &\quad +  \int_0^t r[\ssp{\mu}(r\ds Z(s))-\ssp{\mu}(r\Delta(\ds W(s)) \wedge \ssp{k})] ds 
  + \int_0^t \theta\left(\Delta(\ds W(s)) \wedge \frac{\ssp{k}}{r}\right) ds
\end{align*}
In summary, we can write the workload process as
\begin{equation}\label{eq:workload-diffusion-map}
  \begin{split}
    \ds W(t)
    &= \ds W(0) + \ds M_s(t) + \ds M_a(t) + \ds G_1(t) + \ds G_2(t)\\
    &\quad + \int_0^t \theta\left(\Delta(\ds W(s))\wedge \frac{\ssp{k}}{r}\right) ds 
    + r\int_0^t \id{\ds W(s)=0}ds,
  \end{split}
\end{equation}
where
\begin{align}
  \label{eq:3}
  \ds M_s(t) &= \int_0^t\int_0^\infty x d\ds K(\frac{1}{r^2}\ssp{\Lambda}(r^2s),x),\\
  \ds M_a(t) &= m\int_0^t d \frac{1}{r}[\ssp{\Lambda}(r^2s)-\lambda r^2s],\\
  \ds G_1(t) &= \int_0^t r[\ssp{\mu}(r\ds Z(s))-\ssp{\mu}(r\Delta(\ds W(s))\wedge \ssp{k})] ds,\\
  \ds G_2(t) &= \int_0^t \ssp{\theta}\left(\Delta(\ds W(s) \wedge \frac{\ssp{k}}{r})\right) -
              \theta\left(\Delta(\ds W(s)) \wedge \frac{\ssp{k}}{r}\right)  ds.
\end{align}

The following lemma is an extension of the classical one-dimensional Skorohod problem. The proof can be found in \cite{LeeWeerasinghe2011}.  
\begin{lemma}\label{lem-reflection-map}
  Suppose $g$ is a Lipschitz continuous function. For any $x\in\D(\R^+)$, there exists a unique pair $(y,z)\in\D^2(\R^+)$ satisfying
  \begin{align}
    \label{eq:var-sko-1}
    & z(t) = \int_0^tg(z(s))ds + x(t) + y(t),\\
    \label{eq:var-sko-1.5}
    & z(t)\ge 0,\quad \textrm{for all }t\ge 0,\\
    \label{eq:var-sko-2}
    & y(0) = 0 \textrm{ and $y$ is non-decreasing},\\
    \label{eq:var-sko-3}
    & \int_0^t z(s) dy(s) = 0.
  \end{align}
  More over, denote $z=\psi(x)$. The mapping $\psi:\D(\R^+)\to \D(\R^+)$ is continuous in the uniform topology on compact set. 
\end{lemma}

\proofof{Theorem~\ref{thm:W_process_convergence}}
  We first study the first four terms on the right-hand side of equation \eqref{eq:workload-diffusion-map}.  For the initial condition $\ds W(0)$, its convergence to some random variable $w_0$ is part of the assumption \eqref{eq:cond-initial-u} on the initial state.

  According to Lemma~3.8 in \cite{KrichaginaPuhalskii1997}, 
  \begin{equation*}
    \int_0^t\int_0^\infty x d\ds K(\frac{1}{r^2}\ssp{\Lambda}(r^2s),x)
    \dto \sqrt\lambda m c_s M_s(t),\quad\textrm{as }r\to\infty,
  \end{equation*}
  where $M_s(t)$ is a standard Brownian motion (with zero drift and variance $1$). 
  
  It follows from the assumption~\eqref{eq:HT-arrival} that
  \begin{align*}
    \ds M_a(t)
    =m \ds \Lambda(t)\dto \sqrt\lambda m c_a  M_a(t), \quad \textrm{as }r\to\infty.
  \end{align*}
  
  We now study the terms $\ds G_1$ and $\ds G_2$. By the stochastic bound (Lemma~\ref{lem:workload-upperbound}) proved in Section~\ref{sec:state-space-collapse}, for any $\epsilon>0$, there exists $C$ such that $\prs{\Omega_r}\ge 1-\epsilon$, where $\Omega_r = \left\{\sup_{t\in[0,T]}\max\left(\ds Z(s), \Delta \ds W(s)\right)\le C\right\}$ (noting that we naturally have $\ds Z(\cdot)\le \ssp{k}/r$). According to condition \eqref{eq:cond-HT}, for any sample path in the event $\Omega_r$, we have
  \begin{align*}
    \ds G_1(t) \dto 0,\ \quad \ 
    \ds G_2(t) \dto 0,\quad \textrm{as }r\to\infty.
  \end{align*}
  
  Let $\ds Y(t)=r\int_0^t \id{\ds W(s)=0}ds$. It is easy to see that 
  \begin{equation}
    \int_0^t \ds W(s)d\ds Y(s) = 0.
  \end{equation}
  Thus $(\ds W,\ds Y)$ is the solution to the reflection mapping in Lemma~\ref{lem-reflection-map}. So 
  \begin{equation*}
    \ds W = \psi\left(\ds W(0) + \ds M_s+\ds M_a+\ds G_1+\ds G_2\right).
  \end{equation*}
  By the continuous mapping theorem, $\ds W\dto W^*$, where $W^*=\psi(w_0+ \sqrt\lambda m c_s M_s(t) + \sqrt\lambda m c_a M_a(t))$. In other words, the limit $W^*$ satisfies
  \begin{align}
    W^*(t) = w_0 + \sqrt\lambda m c_s M_s(t) + \sqrt\lambda m c_a M_a(t) -\theta(\Delta(W^*))(t) + Y^*(t),
  \end{align}
  with $Y^*(0)=0$ and being non-decreasing and 
  \begin{equation}
    \int_0^t W^*(s)d Y^*(s) = 0.
  \end{equation}
  Thus, we have shown that the diffusion limit of the workload process is an RBM with state-dependent drift $-\theta(\Delta_K(W^*(t)) \wedge K)$ and variance $\lambda m^2(c_s^2+c_a^2)$. 
  The proof of \eqref{eq:total-jobs-diffusion-limit} follows immediately from the continuous mapping theorem.
\endproof

{\vskip 12pt}
\proofof{Lemma~\ref{lemma:RBM_ss}} 
\cite[Chapter 15]{KarlinTaylor_Second} prescribed an approach based on the Kolmogorov equation to compute the stationary distribution for general diffusion processes. Here we provide an alternate derivation using the basic adjoint relationship for an RBM with state-dependent drift and variance. 

We can write $W(t)$ as
\begin{equation}
  \label{eq:RBM-state-dep-drift}
  W(t)=W(0)-\int_0^t\beta(W(\tau))ds+\int_0^t \sqrt{s(W(\tau))} dB(\tau)+Y(t),
\end{equation}
where $B$ is a standard Brownian motion and $Y$ is the regulator process that prevents $W$ from becoming negative. The process $Y$ is non-decreasing and satisfies
\begin{equation}\label{eq:push}
  \int_0^tW(\tau)dY(\tau)=0.
\end{equation}
Let $f$ be a twice differentiable function. By Ito's formula,
\begin{align*}
  f(W(t))-f(W(0))&=\int_0^t \sqrt{s(W(\tau))} f'(W(\tau))dB(\tau)\\
  &\quad+\int_0^t\left[\frac{1}{2} s(W(\tau)) f''(W(\tau))-\beta(W(\tau))f'(W(\tau))\right]d\tau\\
  &\quad+\int_0^tf'(W(\tau))dY(\tau).
\end{align*}
Note that $\int_0^tf'(W(\tau))dB(\tau)$ is a martingale, and that 
\begin{equation*}
  \int_0^tf'(W(\tau))dY(\tau) = f'(0)Y(t),
\end{equation*}
due to regulation \eqref{eq:push}.  Taking the conditional expectation with respect to the stationary distribution $\pi$ on both sides of the above formula, we have
\begin{align*}
  0=\E_{\pi}\int_0^t\left[\frac{1}{2}s(W(\tau))f''(W(\tau))-\beta(W(\tau))f'(W(\tau))\right]d\tau
  +f'(0)\E_{\pi}[Y(t)].
\end{align*}
So
\begin{equation}
  \label{eq:BAR}
  \int_0^\infty\left[\frac{1}{2}s^2(W(\tau))f''(w)-\beta(w)f'(w)\right]d\pi(w)
  +f'(0) \frac{\E_{\pi}[Y(t)]}{t}=0.
\end{equation}
This is known as the \emph{basic adjoint relation} (BAR) in the literature. Our goal is to guess a functional form for $\pi$ so that the integral in the above expression can be decomposed as $f'(0)$ times a term independent of $f$.
 

Consider the following derivative:
\begin{align*}
  &\quad \left[ c(w) h(w) e^{\int_{0}^w g(u)du}\right]' \\
  &= \Big[ c(w) h'(w) + c(w) h(w) g(w) + c'(w) h(w)  \Big] e^{\int_{0}^w g(u)du} \\
  &=\Big[ c(w) h'(w) + \big[ c(w) g(w) + c'(w)\big] h(w) \Big] e^{\int_{0}^w g(u)du} 
\end{align*}
Now substituting $h(w) = f'(w)$ and $c(w) = \frac{1}{2}s(w)$, we obtain the expression inside the square brackets in the integral term of BAR if $ c'(w) + c(w)g(w) = -\mu(w)$. Equivalently, $g(w) = \frac{-\mu(w) - \frac{1}{2}s'(w)}{\frac{1}{2}s(w)}$.
Therefore, letting $d\pi(w) = \alpha \cdot e^{-\int_{0}^w \frac{\mu(v) + \frac{1}{2}s'(v)}{\frac{1}{2}s(v)} dv}$,  
\begin{align*}
  &\quad \int_{0}^{\infty} \left[\frac{1}{2}s(w) f''(w) - \mu(w) f'(w) \right] d\pi(w) \\
  &= \int_{0}^{\infty} \left[\frac{1}{2}s(w) f''(w) - \mu(w) f'(w) \right] \alpha \cdot 
  e^{-\int_{0}^w \frac{\mu(v) + \frac{1}{2}s'(v)}{\frac{1}{2}s'(v)} dv} \\
  &= \alpha \int_{0}^{\infty} d \left[ \frac{1}{2} s(w) f'(w) e^{-\int_{0}^w \frac{\mu(v) + 0.5 s'(v)}{0.5 s(v)} dv}   \right] \\
  &= \alpha \left[ \frac{1}{2} s(\infty) f'(\infty) e^{-\int_{0}^{\infty} \frac{\mu(v)+0.5 s'(v)}{0.5 s(v)} dv}
    - \frac{1}{2} s(0) f'(0) \right] \\
  &= -\frac{1}{2} s(0) \alpha f'(0).
\end{align*}
If we plug in $f(w)=w$ into the BAR \eqref{eq:BAR}, we will get $\frac{\E_{\pi}[Y(t)]}{t} = \frac{1}{2} s(0) \alpha$. This proves the lemma. 
\endproof


\proofof{Proposition~\ref{prop:final_approximation}}
We start with Lemma~\ref{lemma:RBM_ss} and substitute state-dependent variance and drift as
\begin{align*}
s(w) &= \lambda m^2 (c^2_a+c^2_s) \\
\beta(w) &= 
\begin{cases} 
\theta(w/m_e)  = - \lambda m \left.\frac{d \log f(x)}{dx}\right|_{x=\frac{w}{m_e}} & w \leq K\cdot m_e \\
\theta(K)  = - \lambda m \left. \frac{d \log f(x)}{dx}\right|_{x=K} & w > K\cdot m_e 
\end{cases}
\end{align*}
To obtain a further simplification, we use our assumption that $\frac{d \log f(x)}{dx}$ is a constant for $x \geq K$, and therefore 
\[ \beta(w) = -\lambda m \left.\frac{d \log f(x)}{dx}  \right|_{x = \frac{w}{m_e}}, \quad \forall\ w \in [0, \infty) \] 
We then get
\begin{align*}
\prob{W^*(\infty) \leq w} &= \frac{\alpha'}{\lambda m^2 (c^2_a+c^2_s)} \int_0^w e^{ \frac{2}{\lambda m^2 (c^2_a+c^2_s)} \int_0^u \lambda m d \log f(v/m_e)   } du \\
&= \frac{\alpha'}{\lambda m^2 (c^2_a+c^2_s)} \int_0^w e^{ \frac{2 \lambda m m_e}{\lambda m^2 (c^2_a+c^2_s)} \int_0^{u/m_e} d \log f(z)   } du \\
&= \frac{\alpha''}{\lambda m^2 (c^2_a+c^2_s)} \int_0^w e^{ \frac{1+c^2_s}{ c^2_a+c^2_s} \log f(u/m_e)   } du \\
&= \frac{\alpha''}{\lambda m^2 (c^2_a+c^2_s)} \int_0^w f\left( \frac{u}{m_e}\right)^{ \frac{1+c^2_s}{ c^2_a+c^2_s}   } du \\
&= \alpha  \int_0^\frac{w}{m_e} f(u)^{ \frac{1+c^2_s}{ c^2_a+c^2_s} } du 
\end{align*}
which proves \eqref{eqn:W_RBM_ss}.\\
From \eqref{eq:total-jobs-diffusion-limit} and the continuous mapping theorem
\begin{equation}
  \label{eq:X-W-ss}
  X^*(\infty) = \frac{W^*(\infty) \wedge Km_e }{m_e} + \frac{(W^*(\infty) - Km_e)^+}{m}.
\end{equation}
It now follows that
\begin{align*}
\prob{X^*(\infty) \leq x} &= 
\begin{cases}
\prob{W^*(\infty) \leq x m_e} & x \leq K \\
\prob{W^*(\infty) \leq K m_e + (x-K) m} & x > K
\end{cases}
\end{align*}
which, together with \eqref{eqn:W_RBM_ss}, gives \eqref{eq:X_RBM_ss}. \\
To find $\expct{X^*(\infty)}$, we will find it convenient to start with \eqref{eqn:W_RBM_ss} and rewrite it as 
\begin{align}
\prob{\frac{W^*(\infty)}{m_e} \leq z} &= \alpha \int_0^z f(x)^{\frac{c^2_s+1}{c^2_s+c^2_a}} dx.
\end{align}
Therefore, $f(x)^{\frac{c^2_s+1}{c^2_s+c^2_a}}$ is the density of $\frac{W^*(\infty)}{m_e}$. Now we again use the map \eqref{eq:X-W-ss} to write
\begin{align*}
\expct{X^*(\infty)} &= \expct{ \frac{W^*(\infty)}{m_e} \wedge K } + \frac{m_e}{m}\expct{ \left(\frac{W^*(\infty)}{m_e}-K \right)^+} \\
&= \frac{\int_{0}^\infty (x\wedge K) f(x)^\cratio dx}{\int_0^\infty f(x)^{\cratio} dx} + \frac{c^2_s+1}{2} \frac{\int_0^\infty (x-K)^+ f(x)^\cratio dx}{\int_0^\infty f(x)^\cratio dx},
\end{align*}
which proves \eqref{eq:X-expectation-approx}.
\endproof


\proofof{Theorem~\ref{thm:XW_ss_convergence}} 
This theorem essentially establishes the interchange of the steady state and heavy traffic limits for the constructed sequence of Sd-LPS models. Proving such an interchange usually involves quite a complicated analysis of a well-constructed Lyapunov function (see, for example, \cite{GamarnikZeevi2006} and \cite{LeeWeerasinghe2011}). Taking advantage of the existing studies, we use a coupling argument to prove the interchange for our model. The proofs for both the workload and queue length essentially follow the same argument. We only focus on the queue length in this proof. 

For each $r$, we construct an auxiliary system which takes exactly the same arrival stream as the $r$th Sd-LPS system and the same initial condition. 
Denote
\begin{equation*}
  \ssp\mu_\dagger = \ssp\mu(\ssp{k}). 
\end{equation*}
When the number of jobs in the auxiliary system is more than $\ssp{k}$, the server works at rate $\ssp\mu_\dagger$. When the number of jobs drops below $\ssp{k}$, the server works at speed 0 (in other words it completely shuts down). Without loss of generality, we assume that the initial number of jobs is larger than $\ssp{k}$. Let $\ssp{Q}(t)$ and $\ssp{Q}_\dagger(t)$ denote the number of jobs in the queue in the Sd-LPS and auxiliary systems, respectively. It is clear that 
\begin{equation}
  \label{eq:couple-Q}
  \ssp{Q}(t)<\ssp{Q}_\dagger(t).
\end{equation}
Due to parallel processing, overtaking can happen in each system, i.e., the $j$th arriving job may leave the system earlier than the $i$th arriving job even if $j>i$. However, due to the coupling, the $i$th arriving job in the auxiliary system can never enter service earlier than the corresponding job in the Sd-LPS system. 

By condition \eqref{eq:cond-stability}, $\ssp\mu_\dagger>\lambda m$ for all large enough $r$. So both $\ssp{Q}$ and $\ssp{Q}_\dagger$ are stationary. Let $\ssp\pi$ denote the stationary probability measure of the diffusion-scaled process $\ds{Q}$. Similarly, Let $\ssp\pi_\dagger$ denote the stationary probability measure of the diffusion-scaled queue length $\ds{Q}_\dagger$ in the coupled system. The key step to showing that $\ssp{X}(\infty)\dto X^*(\infty)$ as $r\to\infty$ is to show that the family of probability measures $\{\ssp\pi\}_{r\in\N}$ is tight. (Since $\ds{X}(t)\le \ds{Q}(t)+\ssp{k}/r$, studying only the queue length suffices.) Readers can refer to the proof of Theorem~8 in \cite{GamarnikZeevi2006} for a standard argument of how to prove the convergence using tightness. We now focus on proving the tightness of probability measures $\{\ssp\pi\}_{r\in\N}$.

We can model the $r$th auxiliary system as if it has $\ssp{k}$ identical servers. All the servers either work or stop in perfect synchronization.  Denote by $\ssp S_{\dagger,i}(\cdot)$, $i=1,\ldots, \ssp{k}$, independent renewal processes with inter-renewal time following distribution $G(\cdot/\ssp{k})$, where $G$ is the distribution of job sizes. In other words, the inter-renewal time has mean $m\ssp{k}$ and SCV $c_s^2$.
The queueing dynamics of the $r$th auxiliary system can be written as 
\begin{align*}
  \ssp Q_\dagger(t) = \ssp Q_\dagger(0)
  + \ssp \Lambda(t) -\sum_{i=1}^{\ssp{k}}\ssp S_{\dagger,i}(\ssp{B}_\dagger(t)),
\end{align*}
where $\ssp{B}_\dagger(t)$ is the cumulative busy time for each of the servers. Applying the diffusion scaling, we have
\begin{align}
  \label{eq:sko-mapping}
  \ds Q_\dagger(t) = \ds Q_\dagger(0)
  + \ds \Lambda(t) -\sum_{i=1}^{\ssp{k}}\ds S_{\dagger,i}(\frac{1}{r^2}\ssp{B}_\dagger(r^2t))
  + r(\lambda - \ssp{\mu}_\dagger /m)t
  + \frac{\ssp{\mu}_\dagger}{rm}(r^2t-\ssp{B}(r^2t))
\end{align}
where 
\begin{align*}
  \ds \Lambda(t) = \frac{1}{r}\left(\ssp\Lambda(r^2t)-\lambda r^2 t\right),
  \quad
  \ds S_{\dagger,i} = \frac{1}{r}\left(\ssp S_{\dagger,i}(r^2t)-\frac{\ssp\mu_\dagger r^2}{m\ssp{k}} t\right).
\end{align*}
Note that $r^2t-\ssp{B}(r^2t)$ increases only when $\ds Q_\dagger(t)=0$, so \eqref{eq:sko-mapping} is the same as the Skorohod mapping for the $G/G/1$ queue except that the service process is the superposition of $\ssp{k}$ renewal processes with a much lower speed (roughly $1/\ssp{k}\approx 1/r$) rather than a single renewal process. We now take advantage of the tools developed in \cite{LeeWeerasinghe2011} by verifying that the processes $\ds \Lambda(t)$ and $\ds S_{\dagger,i}$ satisfy condition (A8.p) there. That is, we want to show
\begin{align}
  \label{eq:A8p-arr}
  &\E \left[\sup_{0\le s\le t}|\ds \Lambda(s)|^2\right]<C(1+t),\\
  \label{eq:A8p-ser}
  &\E \left[\sup_{0\le s\le t}|\ds S_{\dagger,i}(s)|^2\right]<\frac{C}{r}(1+t).
\end{align}
Condition \eqref{eq:A8p-arr} directly follows from assumption \eqref{eq:HT-arrival}, following the same argument as in \cite{LeeWeerasinghe2011} (essentially using Lemma~3.5 in \cite{BudhirajaGhosh2006}). To verify \eqref{eq:A8p-ser}, we need to further investigate the proof of Lemma~3.5 in \cite{BudhirajaGhosh2006}. It follows from (3.31) and (3.32) in the proof that the result of Lemma~3.5 can be enhanced as follows: The right-hand side of the second inequality in (3.20), which is $C^*(1+t)$, can be replaced by $\frac{2}{r}+\frac{C_2}{r^2}+3C_2t$. Since our renewal process $\ssp S_{\dagger,i}$ has speed $1/\ssp{k}$ rather than 1, by time change, we can replace the time $t$ by $\frac{t}{\ssp{k}}$. So $\E \left[\sup_{0\le s\le t}|\ds S_{\dagger,i}(s)|^2\right]<\frac{2}{r}+\frac{C_2}{r^2}+3C_2 \frac{t}{\ssp{k}}\le \frac{3}{r}+6KC_2 \frac{t}{r}$ for all large enough $r$. 
This proves \eqref{eq:A8p-ser}. Then following exactly the same argument, Theorem~3.3 in \cite{LeeWeerasinghe2011} holds for our problem. This implies Theorem~3.2 in \cite{LeeWeerasinghe2011}, i.e., $\sup_r\int_0^\infty w\ssp\pi_\dagger(dw)<\infty$. By the coupling construction \eqref{eq:couple-Q}, we have $\int_0^\infty x\ssp\pi(dx)<\int_0^\infty x\ssp\pi_\dagger(dx)$. This implies tightness of $\{\ssp\pi\}_{r\in\N}$.
\endproof


\section{State Space Collapse for the Sd-LPS system}
\label{sec:state-space-collapse}

We introduce a strengthened version of the mapping $\Delta_K$ as the follows.  Let $\Delta_{K,\nu}:\R_+\to\M\times\M$ be the lifting map associated with the probability measure $\nu$ and constant $K$ given by
\begin{equation*}
  \Delta_{K,\nu} w=\Bigl(\frac{(w-K\beta_e)^+}{\beta}\nu,
  \frac{w\wedge K\beta_e}{\beta_e}\nu_e\Bigr)
  \quad\textrm{ for }w\in\R_+.
\end{equation*}

We aim to prove the following full version of the SSC 
\begin{theorem}[Full State Space Collapse]\label{thm:full-ssc}
  Under the conditions \eqref{eq:HT-arrival}--\eqref{eq:cond-HT} and \eqref{eq:cond-initial}--\eqref{eq:cond-init-ssc}, for any $T>0$,
  \begin{equation*}
    \sup_{t\in[0,T]}\pov[(\ds\buf(t),\ds\ser(t)),\Delta_{K,\nu}\ds W(t)]
    \dto 0\quad\textrm{ as }r\to\infty.    
  \end{equation*}
\end{theorem}
It is clear that Theorem~\ref{thm:full-ssc} implies Proposition~\ref{prop:state_space_collapse}. The rest of this section is devoted to the proof of the full SSC.

\subsection{Tightness of Shifted Fluid-Scaled Processes}
The key to proving SSC, which was originally developed by \cite{Bramson1998}, is to ``chop'' the diffusion-scaled processes into pieces.\\
\noindent\textbf{Shifted Fluid Scaling}
Introduce,
\begin{equation}
  \sfs\buf(t)=\frac{1}{r}\ssp{\buf}(rl+rt),\quad
  \sfs\ser(t)=\frac{1}{r}\ssp{\ser}(rl+rt),
\end{equation}
for all $m\in\N$ and $t\ge 0$.
To see the relationship between these two scalings, consider the diffusion-scaled process on the interval $[0,T]$, which corresponds to the interval $[0,r^2T]$ for the unscaled process.  Fix a constant $L>1$, the interval will be covered by $\fl{rT}+1$ overlapping intervals 
\[[rl,rl+rL]{\hskip 15 pt}l=0,1,\cdots,\fl{rT}.\] 
For each $t\in[0,T]$, there exists an $l\in\{0,\cdots,\fl{rT}\}$ and $s\in[0,L]$ (which may not be unique) such that $r^2t=rl+rs$. Thus 
\begin{equation}\label{eq:relation-frm-dr}
  \ds{\buf}(t)=\sfs{\buf}(s),\quad\ds{\ser}(t)=\sfs{\ser}(s).
\end{equation}
This will serve as a key relationship between fluid and diffusion-scaled processes.

The quantities $\ssp{Q}(\cdot)$, $\ssp{Z}(\cdot)$, $\ssp{X}(\cdot)$, $\ssp{W}(\cdot)$ are essentially functions of $(\ssp{\buf}(\cdot),\ssp{\ser}(\cdot))$, so the scaling for these quantities is defined as the functions of the corresponding scaling for $(\ssp{\buf}(\cdot),\ssp{\ser}(\cdot))$. For example
\begin{align*}
  \sfs{W}(t)&=\inn{\chi}{\sfs\buf(t)+\sfs\ser(t)}=\frac{1}{r}\ssp{W}(rl+rt).
\end{align*}
We define the shifted fluid scaling for the arrival process as
\begin{equation*}
  \sfs \Lambda(t)=\frac{1}{r}\ssp{\Lambda}(rl+rt),
\end{equation*}
for all $t\ge 0$.  By \eqref{eq:B(t)}, the shifted fluid scaling for
$\ssp{B}(\cdot)$ is
\begin{equation*}
  \sfs B(t)=\sfs E(t)-\sfs Q(t),
\end{equation*}
for all $t\ge 0$.
A shifted fluid-scaled version of the stochastic dynamic equations
\eqref{eq:stoc-dym-eqn-B} and \eqref{eq:stoc-dym-eqn-S} can be written
as, for any $A\subset (0,\infty)$, $0\le s\le t$,
\begin{eqnarray}
  \label{eq:mechanism-B}
  &&\begin{split}
    \sfs\buf(t)(A)&=\sfs\buf(s)(A)
    +\frac{1}{r}\sum_{i=r\sfs E(s)+1}^{r\sfs E(t)}\delta_{v_i}(A)\\
    &\quad-\frac{1}{r}\sum_{i=r\sfs B(s)+1}^{r\sfs B(t)} \delta_{v_i}(A),
  \end{split}\\
  \label{eq:mechanism-S}
  &&\begin{split}
    \sfs\ser(t)(A)&=\sfs\ser(s)(A+\ssp{S}(rl+rs,rl+rt))\\
                  &\quad+\frac{1}{r}\sum_{i=r\sfs B(s)+1}^{r\sfs B(t)}
                   \delta_{\ssp{v}_i}(A+\ssp{S}(\ssp{\tau}_i,rl+rt)).
  \end{split}
\end{eqnarray}
We point out that the cumulative service process $\ssp{S}$ is never scaled because it tracks the amount of service received by each individual customer. However, via some algebra we can see that
\begin{equation}
  \label{eq:accum-ser-sfs}
  \ssp{S}(rl+rs,rl+rt) 
  = \int_{rl+rs}^{rl+rt} \frac{\ssp{\mu}(\ssp{Z}(\tau))}{\ssp{Z}(\tau)}d\tau
  = \int_{s}^{t} \frac{\ssp{\mu}(r\sfs{Z}(\tau))}{\sfs{Z}(\tau)}d\tau.
\end{equation}
This gives two interesting observations. First, the shifted fluid scaling is essentially fluid scaling, meaning the shifted fluid-scaled processes should be close to some fluid model solutions. Second, the corresponding fluid model is essentially the same as the fluid model in \cite{ZhangDaiZwart2011} since by \eqref{eq:cond-HT},
\begin{equation*}
  \ssp{\mu}(r\sfs{Z}(\tau)) = 1 + O^+(\frac{1}{r}),
\end{equation*}
where $O^+(1/r)$ means the quantity is positive and of the same order as $1/r$ when $r\to\infty$. So
\begin{equation}
  \label{eq:accum-ser-r-est}
  \ssp{S}(rl+rs,rl+rt) 
  = \int_{s}^{t} \frac{1}{\sfs{Z}(\tau)}d\tau
  +O^+(\frac{1}{r}).
\end{equation}
Intuitively, $\sfs{Z}$ is close to some fluid limit denoted by $\tilde Z$ as $r$ becomes very large (in the mathematical sense of convergence in probability), then 
\begin{equation}
  \ssp{S}(rl+rs,rl+rt) 
  \dto \int_{s}^{t} \frac{1}{\tilde{Z}(\tau)}d\tau.
\end{equation}
So we can conclude that the underlying fluid is the same as the one for the regular LPS system. Thus, we can use existing properties developed in \cite{ZhangDaiZwart2011}.
We hope to make the argument rigorous and concise in the follows.

\noindent\textbf{Some Bound Estimation}\\
The tightness property, which guarantees that the shifted fluid-scaled process $\{\sfs{\buf},\sfs{\ser}\}$ has a convergent subsequence, can be proved in a similar way as in \cite{ZhangDaiZwart2011}. There are two key differences. First is the service process as pointed out before. Second is that \cite{ZhangDaiZwart2011} heavily relies on the known result on the diffusion of the workload (see Proposition~2.1). However, we do not have such a diffusion limit of workload a priori. Instead, we try to prove such a diffusion limit by SSC. Looking into the details of the machinery in \cite{ZhangDaiZwart2011}, what essentially is needed for the workload process is some kind of stochastic bound, which we prove in the following lemma.

\begin{lemma}[An Upper Bound of the Workload]
  \label{lem:workload-upperbound}
  For any $\eta>0$ there exists a constant $M$ such that
  \begin{equation}\label{ineq:W-bound}
    \prs{
      \max_{l\le {rT}}\sup_{t\in[0,L]}\sfs{W}(t)<M
    }>1-\eta.
  \end{equation}
\end{lemma}
\proof
Using the relationship between the shifted fluid scaling and diffusion scaling, we essentially need to prove that 
  \begin{equation*}
    \prs{
      \sup_{t\in[0,L]}\ds{W}(t)<M
    }>1-\eta.
  \end{equation*}
 Recall the representation \eqref{eq:workload-diffusion-map} for the diffusion-scaled workload processes.
Let $\underline\theta = \inf_{x\in[0,K]}\theta(x)$, which is finite due to condition \eqref{eq:cond-HT}, so the process $\ds W_1$ satisfying
 \begin{equation*}
    \ds W_1(t)
    = \ds W(0) - \underline\theta t + \ds M_s(t) + \ds M_a(t) + \ds G_1(t) + \ds G_2(t)
    + r\int_0^t \id{\ds W_1(s)=0}ds
\end{equation*}
is an upperbound of $\ds W$ due to the definition of $\underline\theta$ and condition \eqref{eq:cond-stability}. By Lemma~\ref{lem-reflection-map}, $\ds W_1$ converges to a driftless RBM, which is stochastically bounded. This implies the result.
\endproof

Such a stochastic bound of the workload process helps to establish some useful bound estimates for the stochastic processes underlying the Sd-LPS model. 
\begin{lemma}[Further Bound Estimations]
  \label{lem:bound-estimation}
  For any $\eta>0$, there exists a constant $M>0$ and a probability event $\Omega^r_B(M)$ for each index $r$ such that
  \begin{equation}\label{ineq:Omega-B}
    \liminf_{r\to\infty}\prs{\Omega^r_B(M)}\ge 1-\eta,
  \end{equation}
  and on the event $\Omega^r_B(M)$, we have
  \begin{align}
    \label{eq:Q-bound}
    &\max_{l\le\fl{rT}}\sup_{t\in[0,L]} \sfs Q(t)\le M,\\
    \label{eq:mvp-p-bound}
    &\max_{l\le\fl{rT}}\sup_{t\in[0,L]}
       \inn{\chi^{1+p}}{\sfs\buf(t)+\sfs\ser(t)}\le M.
  \end{align}
\end{lemma}
\proof
The result \eqref{eq:Q-bound} holds due to Lemma~4.2 in \cite{ZhangDaiZwart2011}, which only utilizes the regularity of the arrival process \eqref{eq:HT-arrival} and the stochastic bound \eqref{ineq:W-bound} for the workload process proved in Lemma~\ref{lem:workload-upperbound}. For \eqref{eq:mvp-p-bound}, the first half, $\max_{l\le\fl{rT}}\sup_{t\in[0,L]}\inn{\chi^{1+p}}{\sfs\buf(t)}\le M$,
also follows the same reasoning as Lemma~4.3 in \cite{ZhangDaiZwart2011}. Essentially, any results for the ``queue'' part follows the same argument in \cite{ZhangDaiZwart2011}. 

The challenge with the state-dependent service rate lies in the analysis of the server. It follows from the shifted fluid-scaled dynamic equation \eqref{eq:mechanism-S} that for any Borel set $A\subset (0,\infty)$, 
\begin{align*}
  \frac{1}{r}\ssp\ser(rl+rt)(A)
  &=\frac{1}{r}\ssp\ser(0)(A+\ssp{S}(0,rl+rt))\\
  &+\sum_{j=0}^{m-1}\frac{1}{r}\sum_{i= \ssp{B}(r(l-j-1))+1}^{\ssp{B}(r(l-j))}
  \delta_{v_i}(A+\ssp{S}({\ssp{\tau}_i},rl+rt))\\
  &+\frac{1}{r}\sum_{i=\ssp{B}(rl)+1}^{\ssp{B}(rl+rt)}
  \delta_{v_i}(A+\ssp{S}({\ssp{\tau}_i},rl+rt)).   
\end{align*}
Given $0\le j\le m-1$, for those $i$'s with $\ssp{B}(r(l-j-1))<i\le \ssp{B}(r(l-j))$ we have
\[{\ssp{\tau}_i}\in[r(l-j-1),r(l-j)].\]
For the sake of simplicity, let us assume that $\ssp{Z}(s)>0$ for all $s\in[0,rl+rt]$. If this does not hold, we can use a technical trick presented in the proof of Lemma~4.3 in \cite{ZhangDaiZwart2011} to deal with it.  Here we show the main difference coming from the state-dependent service rate. By \eqref{eq:accum-ser-r-est} and the fact that $\ssp{Z}\le \ssp{k}$, we have a lower bound on the cumulative service amount
\begin{equation}
  \label{eq:accu-ser-lower-bound}
  \ssp{S}(rs,rt) \ge \int_{rs}^{rt}\frac{1}{\ssp{Z}(s)}ds
  \ge \frac{r(t-s)}{\ssp{k}}.
\end{equation}
Thus,
\begin{equation*}
  \ssp{S}({\ssp{\tau}_i},rl+rt)\ge
  \ssp{S}(r(l-j),rl)
  \ge\frac{rj}{\ssp{k}}\ge\frac{j}{2K},    
\end{equation*}
for all large $r$ where the last inequality is due to \eqref{eq:K-r}.  For those $i$'s such that $\ssp{\tau}_i$ is larger than $\ssp{B}(rl)$, we use the trivial lower bound $\ssp{S}({\ssp{\tau}_i},rl+rt)\ge 0$.  Also take the trivial lower bound that $\ssp{S}(0,rl+rt)\ge 0$.  Then we have the following inequality on the $(1+p)$th moment:
\begin{equation}\label{ineq:1+p-moment-bound-est-1}
  \begin{split}
    \inn{\chi^{1+p}}{\frac{1}{r}\ssp{\ser}(rl+rt)}
    &\le\inn{\chi^{1+p}}{\frac{1}{r}\ssp{\ser}(0)}\\
    &+\sum_{j=0}^{m-1}
    \inn{\big((\chi-\frac{j}{2K})^+\big)^{1+p}}
    {\frac{1}{r}\sum_{i= \ssp{B}(r(l-j-1))+1}^{\ssp{B}(r(l-j))}\delta_{v_i}}\\
    &+\inn{\chi^{1+p}}
    {\frac{1}{r}\sum_{i=\ssp{B}(rl)+1}^{\ssp{B}(rl+rt)}\delta_{v_i}}.
  \end{split}
\end{equation}
This is the same as (4.22) in \cite{ZhangDaiZwart2011}. The estimation of the first term on the right-hand side in the above follows directly from the initial condition \eqref{eq:cond-initial-u}. The analysis of the second and third terms follows the same way as in \cite{ZhangDaiZwart2011}. 
\endproof

To prove that a family of measure-valued processes is tight, there are three properties to verify, namely \emph{Compact Containment}, \emph{Asymptotic Regularity} and \emph{Oscillation Bound}. For brevity, we will not repeat the exact mathematical statements and their proofs. For the LPS system, these three properties were proved in Lemmas~4.4--4.6 in \cite{ZhangDaiZwart2011}. We just point out that the proof for the above mentioned three properties for the Sd-LPS system relies on $(a)$ the bound estimate in Lemma~\ref{lem:bound-estimation}; and $(b)$ the fact that \eqref{eq:accum-ser-r-est} implies the lower bound of the cumulative service process \eqref{eq:accu-ser-lower-bound}. The proof of Lemma~\ref{lem:bound-estimation} has demonstrated point (b) clearly, we therefore omit a repeat of the argument used in \cite{ZhangDaiZwart2011}. So we reach the conclusion:

\begin{proposition}[Tightness of Shifted Fluid-scaled Processes]
  The family of shifted fluid-scaled processes $\{(\sfs\buf,\sfs\ser\}_{l\le rT, r\in\N}$ is tight.  
\end{proposition}

Loosely speaking, tightness means that any subsequence from the family of shifted fluid-scaled processes has a convergent subsequence. This is formally stated in Theorem~4.1 in \cite{ZhangDaiZwart2011}. 

\subsection{Bramson's Framework for SSC}
Sd-LPS and LPS essentially use the same measure-valued framework. The difference lies in the cumulative service process as we explained when deriving \eqref{eq:accum-ser-r-est} and the workload process as we studied in Lemma~\ref{lem:workload-upperbound}. After obtaining the tightness, we can apply the framework invented by Bramson \cite{Bramson1998} in the same way as how Section~5 in \cite{ZhangDaiZwart2011} applies it to the measure-valued process. The high level-logic is as the follows: the shifted fluid-scaled processes are ``close'' to the fluid model solution, and the fluid model solution converges to some invariant which exhibits SSC (Theorem~3.1 in \cite{ZhangDaiZwart2011}). Thus SSC, which happens on the diffusion scaling, can be proved based on the relationship \eqref{eq:relation-frm-dr} between diffusion scaling and shifted fluid scaling. We thus refer to Section~5 in \cite{ZhangDaiZwart2011} for the proof of Theorem~\ref{thm:full-ssc}.





\section{Analysis of Algorithms for Finding Optimal Control}

Recall some notation and definitions used in this section.
\begin{align*}
\hat{k} &= \argmax_{k} \theta(k) \\
\hat{\theta} &= \theta(\hat{k}) \\
\Delta_k(w) &= \frac{w}{m} + k \left( 1 - \frac{m_e}{m}\right) \\
d_\theta &= \sup_k \theta(k) - \inf_{k} \theta(k) \\
k_f(w) &= \argmax_{k \in [0, w/m_e]} \theta(k_f) 
\end{align*}
Throughout this section, we assume that $d_\theta$ is finite, and therefore $\theta(k)$ is bounded from above and below. 

\subsection{Some Auxiliary Results}
We first provide some auxiliary results (Lemmas~\ref{lemma:continuous_GvW} and \ref{lemma:G0_Lipschitz}) which will be useful in proving the results in Section~\ref{sec:dynamic_control}.

\begin{lemma}\label{lemma:continuous_GvW}
Consider the solution of the following ODE, parameterized by $v$ and $W$:
\begin{align*}
\intertext{Terminal condition:}
G_{v,W} (w) &= \alpha w +\beta v + \gamma  & \ldots w \geq \max\{W, \hat{k}m_e \} \\
\intertext{ODE:}
v &= \frac{w}{m} +  k_f(w) \left( 1 - \frac{m_e}{m} \right) - \theta(k_f(w)) G_{v,W}(w) + \frac{\sigma^2}{2} G'_{v,W}(w) & \ldots w \in [W, \hat{k}m_e] \\
v &= \min_{k \in [0,w/m_e]} \left\{ \frac{w}{m} +  k \left( 1 - \frac{m_e}{m} \right) - \theta(k) G_{v,W}(w) + \frac{\sigma^2}{2} G'_{v,W}(w)  \right\} & \ldots w \in [0, W]
\end{align*}
Then $G_{v,W}(w)$ is continuous in both $v$ and $W$ for all $w$.
\end{lemma}
\begin{proof}
Let $(v_a, W_a)$ and $(v_b, W_b)$ denote two parameter settings, and for succinctness, denote the corresponding solutions to the ODE as $G_a$ and $G_b$, respectively. We will consider the case $W_a, W_b \geq \hat{k}m_e$ as other cases are analogous.

Let $W_a \leq W_b$.

At $w=W_b$, we have
\begin{align}
\label{eqn:init_condition}
|G_a(W_b) - G_b(W_b)| &= \beta |v_a-v_b|.
\end{align}
For $w \in [W_a, W_b]$, we have
\begin{align}
  \label{eq:tech-EC3-1}
   G_a(w) &= \alpha w + \beta v_a + \gamma, \\
   G'_b(w) &= \frac{2}{\sigma^2} \left( v_b +\theta(k_b(w)) G_b(w) - \frac{w}{m} +k_b\left(1-\frac{m_e}{m} \right) \right),
\intertext{which gives}
 & \frac{2}{\sigma^2}\left( v_b - \frac{W_b}{m \wedge m_e} -d_\theta G_b(w) \right) \leq  G'_b(w) \leq \frac{2}{\sigma^2} \left(v_b - \frac{W_a}{m \vee m_e}  +  d_\theta  G_b(w)  \right).
\intertext{Since the derivatives are bounded, $G_b(w)$ is bounded in the interval $[W_a, W_b]$. Let $D = \sup_{w\in[W_a, W_b]} |G_b(w)|$. Then,}
|G_a(w) - G_b(w)| & \leq |G_a(w) - G_a(W_a)| + |G_a(W_b) - G_b(W_b)| + |G_b(w) - G_b(W_b) | \\
& \leq \alpha |W_b - w| + \beta |v_a-v_b| + (W_b - w) \frac{2}{\sigma^2} \left( v_b + \frac{W_b}{m \wedge m_e} + d_\theta D \right) \\
\label{eqn:bound_1}
& \leq \alpha |W_b - W_a| + \beta |v_a-v_b| + (W_b - W_a) \frac{2}{\sigma^2} \left( v_b + \frac{W_b}{m \wedge m_e} + d_\theta D \right),
\end{align}
which goes to 0 as $|v_a-v_b| + |W_a-W_b| \to 0$.

For $w \in [0,W_a]$, by Lemma~\ref{lemma:bounding}, 
\begin{align}
\nonumber |G'_{a}(w) - G'_b(w)| & \leq \frac{2}{\sigma^2}|v_a-v_b| + \frac{2}{\sigma^2}\left | \min_{k_a \in [0, w/m_e]} \left( k_b (1-m_e/m) - \theta(k_a) G_a(w) \right) \right. \\
\nonumber & \qquad \qquad \qquad \left. - \min_{k_b\in[0,w/m_e]} \left( k_b (1-m_e/m) - \theta(k_b) G_2(w) \right)  \right| \\
\label{eqn:Gronwall_part2} & \leq \frac{2}{\sigma^2}|v_a - v_b| + \frac{2d_\theta}{\sigma^2} | G_a(w) - G_b(w)|.
\end{align}
Applying Gronwall's inequality, for all $w\in[0,W_a]$
\begin{align*}
  |G_a(w) - G_b(w)| 
  \leq |G_a(W_a)-G_b(W_a)| e^{\frac{2d_\theta}{\sigma^2} (W_a-w) } 
  + \frac{|v_a-v_b|}{ d_\theta} \left( e^{\frac{2d_\theta}{\sigma^2} (W_a-w)} -1 \right),
\end{align*}
which, together with \eqref{eqn:bound_1}, implies that for all $w\in[0,W_a]$
\begin{align*}
  |G_a(w)-G_b(w)| & \leq 
|v_a-v_b| \left( |\beta| e^{\frac{2d_\theta}{\sigma^2}(W_a-w)} + \frac{e^{\frac{2d_\theta}{\sigma^2} (W_a-w)} -1}{d_\theta} \right) \\
& \quad + |W_b-W_a| \left( \alpha + \frac{2}{\sigma^2} \left( v_b + \frac{W_b}{m\wedge m_e} + d_\theta D\right) \right)e^{\frac{2d_\theta}{\sigma^2}(W_a-w)},
\end{align*}
which goes to 0 as $|v_a-v_b| + |W_a-W_b| \to 0$.
\end{proof}

\begin{lemma}\label{lemma:G0_Lipschitz}
Consider $G_{v,W}$ defined in Lemma~\ref{lemma:continuous_GvW} for a given $W \geq \hat{k}m_e$. Then $G_{v,W}(w)$ is monotonic and Lipschitz continuous in $v$ for all $w$. 
\end{lemma}
\begin{proof}
Fix $W \geq \hat{k}m_e$, and consider $v_a > v_b$. Let $G_a$ and $G_b$ denote the solutions of the ODE defined in Lemma~\ref{lemma:continuous_GvW} for $v_a$ and $v_b$, respectively. We will show that $G_a(w) < G_b(w)$ for all $w \geq 0$. We rely on the following two facts:
\begin{enumerate}
\item Terminal condition:
\[ G_b(w) - G_a(w) = -\beta(v_a-v_b) \quad w \geq W \]
\item Bounds on $G'_b(w) - G'_a(w)$ for $w \in [0,W]$: 
\[ G'_b(w) - G'_a(w) = - \frac{2}{\sigma^2} \left[ (v_a-v_b) - \min_{k\in [0,w/m_e]}\left( \Delta_k(w)-\theta(k)G_a(w) \right) - \min_{k\in[0,w/m_e]} (\Delta_k(w) - \theta(k)G_b(w)) \right]\]
where recall that $\Delta_k(w) = \frac{w}{m}+k\left(1-\frac{m_e}{m} \right)$. Under the assumption $G_a(w) \leq G_b(w)$, from Lemma~\ref{lemma:bounding}:
\begin{align}
-\frac{2}{\sigma} \left[ (v_a-v_b)+d_\theta(G_b(w)-G_a(w)) \right] \leq G'_b(w)-G'_a(w) \leq -\frac{2}{\sigma^2} \left[ (v_a-v_b) - d_\theta \left( G_b(w)-G_a(w)\right) \right]
\end{align}
\end{enumerate}
Combining these two facts, we get for any $w\in[0,W]$
\begin{align}
(v_a-v_b) \left[ -\beta + \frac{1}{d_\theta} \left(1-e^{-\frac{2d_\theta W}{\sigma^2}} \right) \right]
\leq G_b(w) - G_a(w) \leq
(v_a-v_b) \left[ -\beta + \frac{1}{d_\theta} \left(e^{\frac{2d_\theta W}{\sigma^2}} - 1\right) \right]
\end{align}
\end{proof}

\begin{lemma}\label{lemma:bounding}
Let $x_1 = \argmin_{x \in [u,v]} f_1(x)$ and $x_2 = \argmin_{x \in[u,v]} f_2(x)$. Then,
\[ |f_1(x_1) - f_2(x_2) | \leq \sup_{x\in[u,v] } |f_1(x) - f_2(x)| \]
\end{lemma}
\proof{Proof}
Since $ f_1(x_1) \leq  f_1(x_2)$ and $f_2(x_2) \leq f_2(x_1)$, 
\[ f_1(x_1) - f_2(x_1) \leq f_1(x_1) - f_2(x_2) \leq f_1(x_2) - f_2(x_2)   \]
and therefore, $|f_1(x_1) - f_2(x_2) | \leq \sup_{x\in[u,v] } |f_1(x) - f_2(x)| $.
\endproof

\subsection{Proofs of Results in Section~\ref{sec:dynamic_control}}
\label{sec:proofs-of-prop-alg}

\begin{proofof}{Proposition~\ref{prop:monotonic_V}}
We should point out that the monotonicity of the value function is not immediate because under the optimal policy $k^*(\cdot)$, the state-dependent cost function $\Delta_{k^*}(w)$ need not be monotonic in $w$. If it were, a simple sample path coupling argument could be used to deduce the monotonicity of the discounted value function by considering initial workloads $w_1 \leq w_2$.

Let $k^*_\gamma(\cdot)$ be the optimal policy minimizing expected discounted cost, and $V_\gamma(w)$ be the corresponding value function. Consider $w_1 \leq w_2$. We will create an alternate control policy $\pi_1$ when the initial workload is $w_1$, and denote the corresponding expected discounted cost by $\tilde{v}_1$. We will then show that $\tilde{v}_1 \leq V_\gamma(w_2)$ (in fact, our construction involves stochastic coupling and implies that the discounted reward starting with $w_1$ and using $\pi_1$ is stochastically smaller than the discounted reward starting with $w_2$ and using $k^*_\gamma(\cdot)$).

Construction of $\pi_1$: We simulate two independent systems in parallel: system 1 with initial workload $W_1(0)=w_1$ under control policy $\pi_1$ (which we will describe shortly); and system 2 with initial workload $W_2(0)=w_2$ under the optimal control policy $k^*_\gamma(w)$. The control at time $t$ under $\pi_1$ is chosen to be 
\[ k_{\pi_1}(t) = \argmin_{k \in \left[0,  W_1(t)/m_e \right]} \frac{W_1(t)}{m} + k(1-m_e/m) \]
for $t\in[0,\tau]$, where $\tau  \doteq \min \{s\geq 0: W_1(s) = W_2(s) \} $ is the coupling time of the two systems. That is, $\tau$ is the first time the workloads of the two coupled processes $W_1$ and $W_2$ coincide. For $ t \geq \tau$, $k_{\pi_1}(t) = k^*_\gamma(W_1(t))$. 

It is easy to see that since $W_1$ and $W_2$ have continuous sample paths, $W_1(t) \leq W_2(t)$
 for $t\leq \tau$. Due to the choice of $k_{\pi_1}$, this further implies that
\begin{align*}
 \min_{k \in [0, W_1(t)/m_e ]} \left( \frac{W_1(t)}{m}+k\left(1-\frac{m_e}{m} \right) \right ) &= \min\left\{ \frac{W_1(t)}{m}, \frac{W_1(t)}{m_e} \right\} \\
& \leq \min \left\{ \frac{W_2(t)}{m}, \frac{W_2(t)}{m_e} \right\} \\
& \leq  \frac{W_2(t)}{m} + k^*_\gamma(W_2(t)) \left(1-\frac{m_e}{m} \right).
\end{align*}
For $t\geq \tau$, $W_1(t)$ is stochastically equal to $W_2(t)$. Therefore, the discounted cost of $\pi_1$ (with initial workload $w_1)$ is stochastically smaller than the discounted cost of $k^*_{\gamma}$ (with initial workload $w_2$). This implies $\tilde{v}_1 \leq V_\gamma(w_2)$, but $V_\gamma(w_1) \leq \tilde{v}_1$ (since $V_\gamma(w_1)$ is the optimal expected discounted cost). Therefore, $V_\gamma(w_1) \leq V_\gamma(w_2)$ when $w_1 \leq w_2$.

Since $ (V_\gamma(w_2) -  V_{\gamma}(w_1)) \geq 0$ for all $\gamma$, this also holds as $\gamma \downarrow 0$.  

\textbf{Note :} The only facts we relied on to argue monotonicity were $(i)$ continuity of sample paths, and $(ii)$ the cost of the cheapest action available in each state is monotonic in $w$. These appear to be weaker than the conditions typically used in the literature where the set of available actions is assumed to be independent of the state. Further, the cost is assumed to be non-decreasing in the state variable for each action.
\end{proofof}
\ \\
We now provide the proofs of Proposition~\ref{prop:infeasible_v} and \ref{prop:NR_convergence} for the analysis of our algorithms. We omit the proof of Proposition~\ref{prop:feasible_v} as it mirrors the proof of Proposition~\ref{prop:infeasible_v}. \\

\begin{proofof}{Proposition~\ref{prop:infeasible_v}}
Consider the diffusion control formulation for the Sd-LPS system but with a finite workload buffer of $W$. For the diffusion corresponding to this loss system, we have reflections at both $w=0$ and $w=W$. Therefore, for any policy for this loss system, the value function gradient is 0 at both these values \cite{Mandl_analytical}:
\[ G(0) = G(W) = 0.  \]
Therefore, \eqref{eqn:ODE_infeasible} defines the HJB equation for the value function gradient of the finite buffer system with workload buffer $W$, together with $G_v(0)=0$ and the additional boundary condition $G_v(W)=0$. Lemma~\ref{lemma:continuous_GvW} guarantees that ODE \eqref{eqn:ODE_infeasible} has a unique solution (by choosing the terminal condition $G_{v,W}(W)=0$).

We first show that for all $v < v^*$, there is a \emph{unique} value of $W$ such that $v$ is the average cost of the optimal finite buffer policy with workload buffer $W$. \\
Consider an arbitrary pair $W, v$ and solve the following ODE
\begin{equation}
  \label{eq:tech-ODE-1}
  v = \min_{k \in [0,w/m_e]} \left\{ \frac{w}{m} +  k \left( 1 - \frac{m_e}{m} \right) - \theta(k) G_{v,W}(w) \right\} + \frac{\sigma^2}{2} G'_{v,W}(w)  
\end{equation}
backwards with terminal condition $G_{v,W}(W)=0$ (note that this is the same ODE as \eqref{eqn:ODE_infeasible} but we do not enforce $G_{v,W}(0)=0$). Lemma~\ref{lemma:G0_Lipschitz} then shows that $G_{v,W}(0)$ is monotonic in $v$. Therefore, for each $W$, there exists a unique $v^*(W)$ such that $G_{v^*(W),W}(0)=0$ for the ODE above, with terminal condition $G_{v^*(W),W}(W)=0$. Further, Lipschitz continuity and Lemma~\ref{lemma:continuous_GvW} imply that the map $v^*(W)$ is continuous. From the foregoing discussion, we see that $v^*(W)$ denotes the cost of the optimal finite buffer policy with finite buffer $W$.

We next show that $v=v^*(W_1) \neq v^*(W_2)$ if $W_1 \neq W_2$. This would imply that two different workload buffer sizes must yield different optimal costs. Assume the contrapositive, and further $W_1 < W_2$. This implies that $G_{v,W_1}(w) = G_{v,W_2}(w)$  for $w\in [0,W_1]$ when the $G_{v,W}$ ODEs are evolved forward with initial condition $G_{v,W_1}=G_{v,W_2}=0$. Then by \eqref{eq:tech-ODE-1},
\begin{align*}
\left. G'_{v,W_2}(w) \right|_{w=W_1} = \left.G'_{v,W_1}(w) \right|_{w = W_1} &= \frac{2}{\sigma^2} \left[v - \min_{k \in [0,W_1/m_e]} \Big( \Delta_k(W_1) - \theta(k) G_{v,W_1}(W_1) \Big) \right] \\
&= \frac{2}{\sigma^2} \left[v - \min_{k \in [0,W_1/m_e]} \Delta_k(W_1) \right] \\
&= \frac{2}{\sigma^2} \left[v - \min \left\{  \frac{W_1}{m}, \frac{W_1}{m_e}\right\}  \right] <0.
\end{align*}
The last inequality is true because $\theta()$ is bounded from below, and hence for the optimal policy with buffer $W_1$, the average cost is strictly smaller than $\min\{W_1/m, W_1/m_e\}$. This implies $G_{v,W_2}(W_1 +\epsilon)<0$ for any $\epsilon > 0$. A similar argument as in Proposition~\ref{prop:monotonic_V} shows that the optimal value function for the finite buffer system is monotonic and hence $G_{v,W_2}(w) \geq 0, \ w\in[0,W_2]$, which contradicts $G_{v,W_2}(W_1)=0$ and $G'_{v,W_2}(W_1)<0$. 
 
Therefore, $\underline{W}(v)$ as defined in \eqref{eq:tech-W-v} (if it exists) is the unique buffer size corresponding to the optimal finite buffer policy with average cost $v$. 
To get a control on how $\underline{W}(v)$ grows as $v \uparrow v^*$ (where $v^*$ denotes the average cost of the infinite buffer control), 
we will next argue that $\underline{W}(v) = O \left(\log \frac{1}{v^*-v} \right)$ (and hence also finite). We will instead prove the following equivalent result: let $v^*_W$ denote the average cost of the optimal finite buffer control with workload buffer limit $W$, then $(v^* - v^*_W) = O(e^{-\beta W})$ as $W \to \infty$ for some constant $\beta > 0$.

Intuitively, the service rate of the optimal control must asymptotically approach $\hat{\theta}$ as the backlog builds up, and hence the distribution of the workload (and therefore number of jobs in the system) should decay at an exponential rate. Therefore, the effect of truncation at workload $W$ for a finite buffer system, that is $(v^*-v^*_W)$, should also be $O(e^{-\beta W})$ for some constant $\beta > 0$.

The proof will proceed in several steps:\\
{\bf Step 1:} For the optimal infinite buffer control, there exists a constant $\alpha$ such that the optimal value function gradient $G^*(w)$ satisfies
\begin{equation}
  \label{eq:tech-G-bdd}
  G^*(w) \leq \alpha + \frac{w}{m \hat{\theta}}.
\end{equation}
\emph{\hspace{0.1in} Proof:} Consider the upper envelope function $\overline{G}_{v^*}(\cdot)$ given by
\begin{align}
\overline{G}_{v^*}(w) &=  \frac{w}{m \hat{\theta}} + \left( \hat{k}\left(1-\frac{m_e}{m} \right) + \frac{\sigma^2}{2m\hat{\theta}} - v^* \right)\frac{1}{\hat{\theta}},
\quad  w \geq \hat{k} m_e.
\end{align} 
For $w \in[0, \hat{k}m_e]$,  $\overline{G}_{v^*}$ is obtained by solving ODE \eqref{eqn:Gf_ODE} backwards starting with the terminal condition 
\[ \overline{G}_{v^*}( \hat{k}m_e) = \left(\hat{k} - v^* + \frac{\sigma^2}{2m\hat{\theta}}\right) \frac{1}{\hat{\theta}}\ .\]  
Note that this is the same upper envelope function we use for detecting feasibility in the binary search algorithm. Now  $v^* \leq v_f(0)$ implies $\overline{G}_{v^*}(0) \geq 0$ since $\overline{G}_v(0)$ is monotonically decreasing and linear in $v$, and $\overline{G}_{v_f(0)}(0)=0$.\\
Assume that \eqref{eq:tech-G-bdd} does not hold. Then $G^*(\cdot)$ must cross $\overline{G}_{v^*}(\cdot)$ from below at some $ \hat{w} \geq 0$. But then, by following the fluid control for $ w \geq \hat{w}$ we get a feasible control with average cost $v^*$. Therefore, $G^*(w) = \overline{G}_{v^*}(w)$ for $w \geq \hat{w}$, and hence $G^*(w) \leq \overline{G}_{v^*}(w)$ for $w \geq 0$ contradicting our assumption.\\
{\bf Step 2:} Let $G^*_W(w)$ denote the value function gradient for the optimal finite buffer control with workload buffer $W$ and average cost $v^*_W$. Then 
\[ G^*_{W}(w) \leq G^*(w) \qquad \mbox{for } w \in [0,W] \]
and hence $G_W^*(w) \leq \alpha + \frac{w}{m\hat{\theta}}$.

\emph{\hspace{0.1in} Proof:} Compare the HJB equation for $G^*(w)$ and $G^*_W(w)$:
\begin{align*}
v^* &= \min_{k \in [0, w/m_e]} \left( \Delta_k(w) - \theta(k) G^*(w) \right) + \frac{\sigma^2}{2} (G^*)'(w) \\
v^*_W &= \min_{k \in [0, w/m_e]} \left( \Delta_k(w) - \theta(k) G_W^*(w) \right) + \frac{\sigma^2}{2} (G_W^*)'(w)
\end{align*}
Now, for any $ w \geq 0$, if $G^*(w) = G^*_W(w)$, then $(G^*)'(w) \geq (G^*_W)'(w)$ since the terms in the parentheses are equal and $v^* \geq v^*_W$. That is, $G^*_W(\cdot)$ can never cross $G^*(\cdot)$ from below. Since $G^*(0)=G^*_W(0)$, $G^*(w) \geq G^*_W(w)$ for all $w \geq 0$.\\
{\bf Step 3:} Define
\[ T_W(w) \doteq \int_{0}^w \theta(k^*_W(u)) du,  \]
where $k^*_W(w)$ is the optimal finite buffer control with workload buffer $W$. Then as $W \to \infty$, $T_W(W) = \Theta(W)$. 
That is, the integral of the drift over the interval $[0,W]$ for the family of controls $k^*_W(\cdot)$ parameterized by $W$ must asymptotically grow linearly in the buffer limit.

\emph{\hspace{0.1in} Proof:} 
We begin by rewriting the HJB equation for $G^*_W(w)$
\begin{align*}
\frac{\sigma^2}{2}(G^*_W)'(w) &= v^*_W - \min_{k \in [0, w/m_e]} \left( \Delta_k(w) - \theta(k) G_W^*(w) \right) \\
& \leq v^*_W - \frac{w}{m \vee m_e} + \theta(k^*_W(w)) G^*_W(w).
\intertext{Take the integration}
\int_{w=0}^W \frac{\sigma^2}{2}(G^*_W)'(w) dw &\leq W \cdot v^*_W - \frac{W^2}{2(m \vee m_e)} +  \left( \alpha + \frac{W}{m \hat{\theta}} \right)  \int_{0}^W \theta(k^*_W(w)) dw.\\
\intertext{This implies}
\frac{\sigma^2}{2} \left(G^*_W(W) - G^*_W(0) \right) & \leq W \cdot v^*_W - \frac{W^2}{2(m \vee m_e)} +  \left( \alpha + \frac{W}{m \hat{\theta}} \right)  T_W(W).
\end{align*}
The left-hand side of the above inequality is 0 (since the workload reflects at $w=0$ and $w=W$). The first term on the right-hand side grows linearly in $W$ since $v^*_W$ is bounded by $v^*$. The second term grows as $\Theta(W^2)$. If $T_W(W) = o(W)$ then the right-hand side becomes negative for $W$ large enough --  a contradiction. Therefore, $T_W(W)$ must grow at least linearly. By our assumptions, $\theta(w) \leq \hat{\theta} < \infty$. Therefore, $T_W(W) = \Theta(W)$.\\
{\bf Step 4:} Denote the density function of the workload under control $k^*_W$ by $f_W$. The preceding step implies 
\[ f_W(W) = O(e^{-\beta W}) \]
for some positive constant $\beta > 0$. That is, the density function for the workload under optimal finite buffer controls falls exponentially.

\emph{\hspace{0.1in} Proof:} The density function is given by
\begin{align*}
f_W(w) &= \kappa_W e^{-\int_0^w  \frac{\theta(k^*_W(u))}{\sigma^2/2}} du \\
&= \kappa_W e^{- \frac{T_W(w)}{\sigma^2/2}},
\end{align*}
where $\kappa_W$ is the normalization constant. Since $T_W(W) = \Theta(W)$ by the preceding step, $f_W(W) = O(e^{-\beta W)}$ for some $\beta > 0$.\\
{\bf Step 5:} Finally consider the following infinite buffer control, parameterized by workload $W$:
\begin{align*}
\tilde{k}_W(w) &= \begin{cases}
k^*_W(w) & w \leq W, \\
\hat{k} & w > W.
\end{cases}
\end{align*}
That is, we create a fluid continuation control with prefix $k^*_W(w)$. This results in a suboptimal infinite buffer control with average cost $\tilde{v}_W \geq v^* \geq v^*_W$. However, since the workload density decays exponentially under control $k^*_W(\cdot)$, and $\theta(\hat{k}) > 0$, the cost of $\tilde{k}_W(\cdot)$ is at most $O(e^{-\beta W})$ higher than $v^*_W$. That is
\[ (v^*-v^*_W) \leq (\tilde{v}_W - v^*_W) = O(e^{-\beta W}). \]

\end{proofof}


\begin{proofof}{Proposition~\ref{prop:NR_convergence}}
Recall the Newton-Raphson algorithm from Algorithm~\ref{alg:Newtons_method}: We first pick a large enough value of workload $W \geq \hat{k}m_e$ (which is not changed during subsequent iterations). The goal of the Newton-Raphson algorithm then is to find the average cost of the optimal dynamic policy under the restriction that the control for $w\geq W$ is the fluid control $\hat{k}$. With $v_n$ as our guess in the $n$th iteration, we backwards evolve the ODEs:
\begin{align}
\label{eqn:app_eq1}
v_n &= \min_{k \in [0, w/m_e]} \left[\frac{w}{m} + k \left(1-\frac{m_e}{m} \right) - \theta(k) G_{v_n}(w)\right] + \frac{\sigma^2}{2} G'_{v_n}(w) \\
1 &= - \theta(k_{v_n}(w)) g_{v_n}(w) + \frac{\sigma^2}{2} g'_{v_n}(w)
\intertext{for $w\in[0,W]$ with (terminal) boundary conditions:}
G_{v_n}(W) &= \left( \hat{k} \left( 1-\frac{m_e}{m} \right) -v_n + \frac{\sigma^2}{2\hat{\theta}}\right) \frac{1}{\hat{\theta}} + \frac{1}{m\hat{\theta}} W \\
\label{eqn:app_eq2}
g_{v_n}(W) &= -\frac{1}{\hat{\theta}}
\end{align}
Here $k_{v_n}(w)$ denotes the policy obtained while solving the ODE for $G_{v_n}$.

The updated guess for the $(n+1)$st iteration is
\[ v_{n+1} = v_n - \frac{G_{v_n}(0)}{g_{v_n}(0)}. \]
We develop our proof of the proposition in several steps.

{\bf Step 1 :} $v_n \geq v_f(W)$ for $n\geq 1$ 

\emph{\hspace{0.1in} Proof:} Let $\widetilde{G}_{v}(w)$ (parameterized by $v,w$) be given by the ODE
\[ v = \frac{w}{m} + k_{v_n}(w) \left( 1 - \frac{m_e}{m} \right) - \theta(k_{v_n}(w)) \widetilde{G}_{v}(w) + \frac{\sigma^2}{2} \widetilde{G}'_{v}(w)\]
for $w \in [0,W]$
with boundary condition 
\[ \widetilde{G}_{v}(W) =  \left( \hat{k} \left( 1-\frac{m_e}{m} \right) -v + \frac{\sigma^2}{2\hat{\theta}}\right) \frac{1}{\hat{\theta}} + \frac{1}{m\hat{\theta}} W \]
This is essentially the same ODE as \eqref{eqn:app_eq1} but with the $\min$ operator replaced by the fixed policy $k_{v_n}$. The first observation is that  $\widetilde{G}_{v}(0)$ is a linear function of $v$, and $\widetilde{G}_{v_n}(w) = G_{v_n}(w)$ for all $w \in [0,W]$. Further, denoting 
\[ \widetilde{g}_{v}(w) = \frac{d}{d v} \widetilde{G}_{v}(w) \]
it is easy to see that $\widetilde{g}_{v}(w) = g_{v_n}(w)$ . Therefore,
\[ v_{n+1} = v_{n} - \frac{G_{v_n}(0)}{g_{v_n}(0)} = v_{n} - \frac{\widetilde{G}_{v_n}(0)}{\widetilde{g}_{v_n}(0)} \]
Since $\widetilde{G}_v(w)$ is a linear function in $v$ for all $w$, the Newton-Raphson update for $\widetilde{G}_{v}(0)$ directly yields that value of $v$ for which $\widetilde{G}_{v}(0)=0$. But this must be the average cost of policy $k_{v_n}$. Therefore, $v_{n+1}$ is in fact the average cost of policy $k_{v_n}$. Since $k_{v_n}$ is a feasible policy in the set $\mathcal{F}_W$, its average cost must be no less than $v_f(W)$ and hence all the iterates $\{v_1, v_2, \ldots\}$ produced are larger than $v_f(W)$.

{\bf Step 2:} The iterates for average cost $\{v_1, v_2, \ldots \}$ form a strictly decreasing sequence.

\emph{\hspace{0.1in} Proof:} For this, we will show that $G_{v}(0)$ is monotonically decreasing and Lipschitz continuous in $v$ with derivative bounded away from 0. This would imply that for $v > v_f(W)$, $G_v(0) < 0$, as well as $g_v(0) < 0$, and hence $v_1 > v_2 > \ldots > v_f(W)$.

Consider $v_a > v_b$, and let $G_a,g_a,k_a$ and $G_b,g_b,k_b$ represent the solution of \eqref{eqn:app_eq1}-\eqref{eqn:app_eq2} and the optimal controls for $v_a$ and $v_b$, respectively. Our goal is to show $G_a(w) < G_b(w)$ for all $w \geq 0$. We rely on the following two facts:
\begin{enumerate}
\item Terminal condition:
\[ G_b(w) - G_a(w) = \frac{a-b}{\hat{\theta}} \quad w \geq W \]
\item Bounds on $G'_b(w) - G'_a(w)$ for $w \in [0,W]$: \footnote{We use primes to denote derivatives with respect to $w$. Derivatives with respect to $v$ are denoted with $\frac{\partial }{\partial v}$ notation.}
\[ G'_b(w) - G'_a(w) = - \frac{2}{\sigma^2} \left[ (v_a-v_b) - \min_{k\in [0,w/m_e]}\left( \Delta_k(w)-\theta(k)G_a(w) \right) - \min_{k\in[0,w/m_e]} (\Delta_k(w) - \theta(k)G_b(w)) \right]\]
where recall that $\Delta_k(w) = \frac{w}{m}+k\left(1-\frac{m_e}{m} \right)$. 
By the assumption $G_a(w) \leq G_b(w)$ and Lemma~\ref{lemma:bounding},
\begin{align*}
-\frac{2}{\sigma} \left[ (v_a-v_b)+d_\theta(G_b(w)-G_a(w)) \right] \leq G'_b(w)-G'_a(w) \leq -\frac{2}{\sigma^2} \left[ (v_a-v_b) - d_\theta \left( G_b(w)-G_a(w)\right) \right].
\end{align*}
\end{enumerate}
Combining these two facts, we get
\begin{align}
(v_a-v_b) \left[ \frac{1}{\hat{\theta}} + \frac{1}{d_\theta} \left(1-e^{-\frac{2d_\theta W}{\sigma^2}} \right) \right]
\leq G_b(0) - G_a(0) \leq
(v_a-v_b) \left[ \frac{1}{\hat{\theta}} + \frac{1}{d_\theta} \left(e^{\frac{2d_\theta W}{\sigma^2}} - 1\right) \right].
\label{eqn:G_Lipschitz}
\end{align}
Thus we have proved that $G_v(0)$ is monotonically decreasing in $v$ and therefore the Newton-Raphson iterates will form a monotonically decreasing sequence. Further, $G_v(w)$ is Lipschitz continuous in $v$ for all $w$, and therefore its derivative with respect to $v$ exists almost everywhere, and according to the first inequality of \eqref{eqn:G_Lipschitz} this derivative is bounded away from 0.

The properties proved so far are sufficient to prove that the Newton-Raphson algorithm converges to the optimal $v_f(W)$, and the convergence rate is at least linear.

{\bf Step 3:} The second derivative of $G_v(0)$ with respect to $v$ is finite in some neighborhood of $v_f(W)$, and hence the Newton-Raphson iterates converge quadratically to $v_f(W)$.

\emph{\hspace{0.1in} Proof:} 
A sufficient condition to show a quadratic convergence rate for the Newton-Raphson method is that the first derivative is non-zero and the second derivative is finite in some neighborhood of the root. We have already shown that the first condition holds. That is $g_v(0) > 0$ for all $v$. What remains to be done is to show that $\frac{\partial^2 G_v(0)}{\partial v^2} = \frac{\partial g_v(0)}{\partial v}$ is finite in a neighborhood of $v_f(W)$. \\
First consider the case where $\left| 1 - \frac{m_e}{m} \right| = 0$, or in other words $m_e=m$. As pointed out before, the optimal policy in this case is the fluid policy for which the optimal average cost is found in the initialization phase itself and the algorithm terminates after one step. Therefore we assume $\left| 1 - \frac{m_e}{m} \right| > 0$ in the remainder of the proof.\\ 
Abbreviating $\theta_a(w) \doteq \theta(k_a(w))$ and $\theta_b(w) \doteq \theta(k_b(w))$, consider the ODE for $g_a(w)$:
\begin{align*}
g'_a(w) &= \frac{2}{\sigma^2} \left( 1+\theta_a(w) g_a(w) \right)
\end{align*} 
with boundary condition $g_a(W) = -\frac{1}{\hat{\theta}}$. Therefore, we have
\begin{align}
|g'_a(w) - g'_b(w)| &= \frac{2}{\sigma^2} |\theta_a(w)g_a(w)-\theta_b(w)g_b(w)|\\
\label{eqn:gp_diff_bound} & \leq \frac{2}{\sigma^2} \left( |\theta_a(w)-\theta_b(w)| \left(|g_a(w)|+|g_b(w)|\right)  + 2d_\theta |g_a(w)-g_b(w)| \right).
\end{align}
As we have shown in the previous step, $|G_b(w)-G_a(w)| \leq \kappa(w)|v_a-v_b|$ for some bounded function $\kappa(w)$. Combined with our regularity assumptions on $\theta()$, this implies
\begin{align}
\label{eqn:theta_diff_bound}
|\theta_a(w) - \theta_b(w)| \leq |G_b(w)-G_a(w)| \frac{S_\theta^3}{D_\theta \left| 1 - \frac{m_e}{m} \right|} \leq |v_a-v_b| \frac{\kappa(w) S^3_\theta}{D_\theta \left|1 - \frac{m_e}{m} \right|}.
\end{align}
A rough justification of the first inequality is the following: Since
\[ k_{a}(w) = \argmin_{k \in [0, w/{m_e}]} \left( \frac{w}{m_e} - k \left( 1 - \frac{m_e}{m}\right)  + \theta(k) G_a(w) \right), \]
one of three cases occurs: (1) $k_a(w)=0$, (2) $k_a(w)=w/m_e$, or (3) $  (1-\frac{m_e}{m})/G_a(w) \in \partial \theta(k_a)$ where $\partial \theta(k_a)$ is the set of subderivatives of $\theta(k)$ at $k=k_a$. Consider the scenario where the third case occurs for $k_a(w)$ as well as $k_b(w)$ (other cases are easier to handle). By \eqref{eq:tech-1st-der-bd}, the absolute value of the first derivative of $\theta(k)$ is bounded by $S_\theta$. 
Therefore for the case $0 < k_a(w) < \frac{w}{m_e}$ to occur $|G_a(w)|, |G_b(w)| $ must be bounded away from 0. More precisely, $\min(|G_a(w)|, |G_b(w)|) \geq  \frac{\left| 1-\frac{m_e}{m} \right|}{S_\theta}$. 
Now, by \eqref{eq:tech-2nd-der-bd} $D_\theta>0$, we must have
\begin{align*}
\left| 1-\frac{m_e}{m} \right| \left| \frac{1}{G_a(w)} - \frac{1}{G_b(w)} \right| \geq D_\theta |k_a(w) - k_b(w)|,
\end{align*}
or
\begin{align*}
 |k_a(w) - k_b(w) | &\leq \frac{\left| 1-\frac{m_e}{m} \right|}{D_\theta} \left| \frac{1}{G_a(w)} - \frac{1}{G_b(w)} \right|  \\
 &= \frac{\left| 1-\frac{m_e}{m} \right| }{D_\theta} \left| \frac{G_a(w)-G_b(w)}{G_a(w)G_b(w)} \right|  \\
 & \leq  \frac{S^2_\theta}{D_\theta \left| 1-\frac{m_e}{m} \right| } \left| G_a(w)-G_b(w) \right|.
\end{align*}
Finally,
\begin{align*}
|\theta_a(w) - \theta_b(w)| &= |\theta(k_a(w)) - \theta(k_b(w))| \\
& \leq S_\theta \left| k_a(w) - k_b(w) \right| \\
&  \leq  \frac{S^3_\theta}{D_\theta \left| 1-\frac{m_e}{m} \right| } \left| G_a(w)-G_b(w) \right|. 
\end{align*}

Finally substituting \eqref{eqn:theta_diff_bound} into \eqref{eqn:gp_diff_bound} gives
\begin{align}
|g'_a(w) - g'_b(w)| & \leq \frac{2}{\sigma^2} \left( |v_a-v_b| \frac{\kappa(w) S^3_\theta}{D_\theta \left| 1 - \frac{m_e}{m} \right|} 
 \left(|g_a(w)|+|g_b(w)|\right)  + 2d_\theta |g_a(w)-g_b(w)| \right).
\end{align}
Now, a similar calculation to that in step 2 shows that $g_v(w)$ is Lipschitz continuous in $v$ for all $w$, and hence for $w=0$. We believe the lower bound on $\left| \frac{d^2 \theta(k)}{dk^2} \right|$ we used to prove quadratic convergence is an artifact of our rather crude proof technique, and that quadratic convergence holds even without this restriction.
\end{proofof}

\end{document}